\documentclass[11pt]{amsart}

\usepackage{amssymb}
\usepackage{fullpage}
\usepackage{setspace}
\usepackage{hyperref}
\usepackage{cleveref}
\usepackage{graphicx}
\usepackage{xcolor}
\usepackage{times}
\usepackage{subfigure}
\usepackage{wrapfig}
\makeatletter
\allowdisplaybreaks
\makeatother

\usepackage{paralist}
\usepackage[ruled,noend]{algorithm2e}

\usepackage[shortlabels]{enumitem}

\newtheorem{thm}{Theorem}[section]
\newtheorem{prop}[thm]{Proposition}
\newtheorem{lem}[thm]{Lemma}
\newtheorem{cor}[thm]{Corollary}
\newtheorem{fact}[thm]{Fact}

\theoremstyle{definition}
\newtheorem{definition}[thm]{Definition}

\newcommand{\iteralg}{\textsc{IterativeRound}\xspace}
\newcommand{\itersub}{\textsc{ReRoute}\xspace}
\newcommand{\altreroute}{\textsc{ConfigReRoute}\xspace}
\newcommand{\mainalg}{\textsc{PseudoApproximation}\xspace}
\newcommand{\outlieralg}{\textsc{OutliersPostProcess}}
\newcommand{\outlierpost}{\textsc{ComputePartial}}

\newcommand{\lpbasic}{LP_1}
\newcommand{\lpfacil}{LP_2}
\newcommand{\lpreroute}{LP_{iter}}
\newcommand{\lpextra}{LP_1'}

\newcommand{\I}{\mathcal{I}}

\newcommand{\eps}{\epsilon}

\newcounter{note}

\title{Structural Iterative Rounding for Generalized $k$-Median Problems}
\author{Anupam Gupta}
\thanks{Computer Science Department, Carnegie Mellon
                University, Pittsburgh, PA 15213. Research supported in part by NSF awards
  CCF-1907820, CCF1955785, and CCF-2006953.}
\author{Benjamin Moseley}\thanks{Tepper School of Business, Carnegie Mellon
                University, Pittsburgh, PA 15213. Moseley and Zhou were supported in part by a Google Research Award, an Infor Research Award, a Carnegie Bosch Junior Faculty Chair and NSF grants CCF-1824303,  CCF-1845146, CCF-1733873 and CMMI-1938909}
\author{Rudy Zhou}

\begin{document}

	\begin{abstract}
	
    This paper considers approximation algorithms for generalized
        $k$-median problems. This class of problems can be informally
        described as $k$-median with a constant number of extra
        constraints, and
        includes $k$-median with outliers, and knapsack median. Our
        first contribution is a pseudo-approximation algorithm for
        generalized $k$-median that outputs a $6.387$-approximate
        solution, with a constant number of fractional variables. The
        algorithm builds on the iterative rounding framework
        introduced by Krishnaswamy, Li, and Sandeep for $k$-median
        with outliers. The main technical innovation is allowing richer constraint sets in the iterative rounding and taking advantage of the structure of the resulting extreme points.
	
    Using our pseudo-approximation algorithm, we give improved approximation algorithms for $k$-median with outliers and knapsack median. This involves combining our pseudo-approximation with pre- and post-processing steps to round a constant number of fractional variables at a small increase in cost. Our algorithms achieve approximation ratios $6.994 + \eps$ and $6.387 + \eps$ for $k$-median with outliers and knapsack median, respectively. These improve on the best-known approximation ratio $7.081 + \eps$ for both problems \cite{DBLP:conf/stoc/KrishnaswamyLS18}.

	\end{abstract}

\maketitle

\section{Introduction}\label{sec_intro}

Clustering is a fundamental problem in combinatorial optimization,
where we wish to partition a set of data points into \emph{clusters}
such that points within the same cluster are more similar than points
across different clusters. In this paper, we focus on generalizations
of the \emph{$k$-median} problem. Recall that in this problem, we are
given a set $F$ of facilities, a set $C$ of clients, a metric $d$ on
$F \cup C$, and a parameter $k \in \mathbb{N}$.
The goal is to choose a set $S \subset F$ of
$k$ facilities to open to minimize the sum of \emph{connection costs}
of each client to its closest open facility. That is, to minimize the objective $\sum_{j \in C} d(j,S)$, where we define $d(j,S) = \min_{i \in S} d(i,j)$.

The $k$-median problem is well-studied from the perspective of
approximation algorithms, and many new algorithmic techniques have been
discovered while studying it.
Examples include linear program rounding~\cite{DBLP:journals/talg/ByrkaPRST17, DBLP:journals/siamcomp/LiS16}, primal-dual
algorithms~\cite{DBLP:journals/jacm/JainV01}, local search~\cite{DBLP:journals/siamcomp/AryaGKMMP04}, and large data techniques \cite{li2018distributed,MalkomesKCWM15,guha2017distributed,GuhaMMMO03,ImQMSZ20}.  Currently, the best approximation ratio for
$k$-median is $2.675 + \eps$ \cite{DBLP:journals/talg/ByrkaPRST17}, and there is a lower bound of
$1 + 2/e $ assuming $P \neq NP$ \cite{Jain:2002:NGA:509907.510012}. 

Recently, there has been significant interest in generalizations of
the $k$-median problem \cite{CharikarKMN01, DBLP:journals/mor/Krishnaswamy0NS15}.  One such generalization is the
\emph{knapsack median} problem. In knapsack median, each facility has
a non-negative weight, and we are given budget $B \geq 0$. The goal is
to choose a set of open facilities of total weight at most $B$ (instead
of having cardinality at most $k$) to minimize the same objective function.
That is, the open facilities must satisfy a knapsack constraint.
Another commonly-studied
generalization
is \emph{$k$-median with outliers}, also known as \emph{robust $k$-median}.
Here we
open
$k$ facilities $S$, as in
basic $k$-median,
but we no longer have to serve all the clients;
now, we are only required to serve at least $m$ clients $C' \subset C$ of our
choice. Formally, the objective function is now $\sum_{j \in C'} d(j,S)$.

Both knapsack median and $k$-median with outliers have proven to be
much more difficult than the standard $k$-median problem. Algorithmic
techniques that have been successful in approximating $k$-median
often lead to only a pseudo-approximation for these
generalizations---that is, they
violate the knapsack constraint or serve fewer than $m$ clients \cite{DBLP:journals/algorithmica/ByrkaPRSST18, CharikarKMN01,  DBLP:journals/talg/FriggstadKRS19,ImQMSZ20}.
Obtaining ``true'' approximation algorithms
requires new ideas beyond those of $k$-median.   Currently the best
approximation ratio for both problems is $7.081 +\eps$ due to the beautiful iterative rounding framework of Krishnaswamy, Li, and Sandeep \cite{DBLP:conf/stoc/KrishnaswamyLS18}.  The first and only other true approximation for $k$-median with outliers is a local search algorithm due to Ke Chen \cite{Chen08}.

\subsection{Generalized $k$-median}

Observe that both knapsack median and $k$-median with outliers
maintain the salient features of $k$-median; that is, the goal is to
open facilities to minimize the connection costs of served
clients. These variants differ in the way we put constraints on the
open facilities and served clients. In particular, in standard
$k$-median, we have a cardinality constraint on the open facilities,
whereas for knapsack median the open facilities are subject to a
knapsack constraint; in both cases we must serve all clients. For $k$-median with outliers, we are constrained to open at most $k$ facilities, and serve at least $m$ clients.

In this paper, we consider a further generalization of $k$-median that we call \emph{generalized $k$-median (GKM)}. As in $k$-median, our goal is to open facilities to minimize the connection costs of served clients. In GKM, the open facilities must satisfy $r_1$ given knapsack constraints, and the served clients must satisfy $r_2$ given coverage constraints. We define $r = r_1 + r_2$.

\subsection{Our Results}

The main contribution of this paper is a refined iterative
rounding algorithm for GKM. Specifically, we show how to round the
natural linear program (LP) relaxation of GKM to ensure all
except $O(r)$ of the variables are integral, and the objective
function is increased by at most a $6.387$-factor.  It is not
difficult to show that the iterative rounding framework in \cite{DBLP:conf/stoc/KrishnaswamyLS18}
can be extended to show a similar result. Indeed, a
$7.081$-approximation for GKM with at most $O(r)$ fractional
facilities is implicit in their work. 
The improvement in this work is the smaller loss in the objective
value.


Our improvement relies on analyzing the extreme points of certain set-cover-like LPs. These extreme points arise at the intermediate steps of our iterative rounding, and by leveraging their structural properties, we obtain our improved pseudo-approximation for GKM. This work reveals some of the structure of such extreme points, and it shows how this structure can lead to improvements. 
 
Our second contribution is improved ``true'' approximation algorithms for two special cases of GKM: knapsack median and $k$-median with outliers. For both problems, applying the pseudo-approximation algorithm for GKM gives a solution with $O(1)$ fractional facilities. Thus, the remaining work is to round a constant number of fractional facilities to obtain an integral solution. To achieve this goal, we apply known sparsification techniques \cite{DBLP:conf/stoc/KrishnaswamyLS18} to pre-process the instance, and then develop new post-processing algorithms to round the final $O(1)$ fractional facilities.

We show how to round these remaining variables for knapsack median at
arbitrarily small loss, giving a $6.387 + \eps$-approximation,
improving on the best $7.081 + \eps$-approximation.  For $k$-median
with outliers, a more sophisticated post-processing is needed to round
the $O(1)$ fractional facilities. This procedure loses more in the
approximation ratio. In the end, we obtain a $6.994 +
\eps$-approximation, modestly improving on the best known $7.081 +
\eps$-approximation. 



\subsection{Overview of Techniques}\label{sec_tech}

To illustrate our techniques, we first introduce a natural LP relaxations for GKM. The problem admits an integer program
formulation, with variables $\{x_{ij}\}_{i \in F, j \in C}$ and
$\{y_i\}_{i \in F}$, where $x_{ij}$ indicates that \emph{client
  $j$ connects to facility $i$} and $y_i$ indicates that
\emph{facility $i$ is open}. Relaxing the integrality constraints
gives the linear program relaxation $\lpbasic$.  


$$\begin{array}{lll|lll}
(\lpbasic) \min_{x,y} & ~~ &\sum_{i \in F} \sum_{j \in C} d(i,j) \,x_{ij}
                         \qquad\qquad &\qquad  (\lpfacil): &  \min_y
    & \sum_{i \in F} \sum_{j \in C: i \in F_j} d(i,j) \, y_i \\
&&\sum_{i \in F} x_{ij} \leq 1 \qquad \forall j \in C & && y(F_j) \leq 1 \qquad \forall j \in C\\
                      &&x_{ij} \leq y_i \qquad\forall i \in F, j \in C & & &  \\
    &&Wy \leq b & && Wy \leq b \\
&&\sum_{j \in C} a_j (\sum_{i \in F} x_{ij}) \geq c & && \sum_{j \in C} a_j
                                                 y(F_j) \geq
                                                 c \\
                  &&x_{ij}, y_i \in [0,1] \qquad \forall i \in F,j \in C & &&  y_i \in [0,1] \qquad \forall i \in F
\end{array}
$$

We focus on $\lpbasic$ for now. The linear program $\lpbasic$ is the standard $k$-median LP with the
extra side constraints. Note that $\sum_{i \in F} x_{ij} \leq 1$ may seem opposite
to the intuition that we want clients to get ``enough'' coverage from
the facilities, but that will be guaranteed by the coverage constraints
below.

The constraint $Wy \leq b$ corresponds to the $r_1$ knapsack
constraints on the facilities $y$, where
$W \in \mathbb{R}_+^{r_1 \times F}$ and $b \in
\mathbb{R}_+^{r_1}$. These $r_1$ packing constraints can be
thought of as a multidimensional knapsack constraint over the
facilities, and ensure that ``few'' facilities are opened.
Next, $\sum_{j \in C} a_j (\sum_i x_{ij}) \geq c$ corresponds to the $r_2$
coverage constraints on the clients, where
$a_j \in \mathbb{R}_+^{r_2}$ for all $j \in C$ and
$c \in \mathbb{R}_+^{r_2}$. These coverage constraints ensure that
``enough'' clients are served. E.g., having one packing constraint
$\sum_{i \in F} y_i \leq k$ and one covering constraint
$\sum_{j \in C} \sum_{i \in F} x_{ij} \geq m$ ensures that at least $m$
clients are covered by at most $k$ facilities; this is the $k$-median with outliers problem.

\medskip \noindent \textbf{Reducing the variables in the LP: } We get
$\lpfacil$ by eliminating the $x$ variables from $\lpbasic$, thereby reducing the
number of constraints. The idea
from~\cite{DBLP:conf/stoc/KrishnaswamyLS18} is to prescribe a set
$F_j \subseteq F$ of permissible facilities for each client $j$ such that $x_{ij}$
is implicitly set to $y_i \mathbf{1}(i \in F_j)$. 
The details of this
reduction and the procedure for creating $F_j$ are given in 
 \Cref{prop_LPfacil}. 
 Using this procedure, $\lpfacil$ is also a relaxation for GKM. Note that in $\lpfacil$, we use the notation $y(F') = \sum_{i \in F'} y_i$ for $F' \subset F$. 

Now consider solving $\lpfacil$ to obtain an optimal extreme point
$\bar{y}$. 
There must be $|F|$ linearly independent tight constraints at $\bar{y}$, and
we call these constraints the \emph{basis} for $\bar{y}$.  
The tight constraints of interest are the $y(F_j) \leq 1$ constraints;
in general, there are at most $\lvert C \rvert$ such tight
constraints,
and we have little structural understanding of the $F_j$-sets. 

 \noindent \textbf{Prior Iterative Rounding Framework: } Consider the
 family of $F_j$ sets corresponding to tight constraints, so
 $\mathcal{F} = \{F_j \mid j \in C,\, \bar{y}(F_j) = 1 \}$. If
 $\mathcal{F}$ is a family of disjoint sets , then
 the tight constraints of $\lpfacil$ form a face of a partition matroid polytope intersected with at most $r$ side constraints (the knapsack and coverage constraints). Using ideas from, e.g.,~\cite{DBLP:conf/stoc/KrishnaswamyLS18, DBLP:journals/mp/GrandoniRSZ14}, we can show that $\bar{y}$ has at most $O(r)$ fractional variables. 
 
 Indeed, the goal of the iterative rounding framework in
 \cite{DBLP:conf/stoc/KrishnaswamyLS18} is to control the set family
 $\mathcal{F}$ to obtain an optimal extreme point where $\mathcal{F}$
 is a disjoint family. To achieve this goal, they iteratively round an
 auxiliary LP based on $\lpfacil$,  where they have the constraint
 $y(F_j) = 1$ for all clients $j$ in a special set $C^* \subset
 C$. Roughly, they regulate what clients are added to $C^*$ and delete
 constraints $y(F_j) \leq 1$ for some clients. The idea is that a
 client $j$ whose constraint is deleted must be close to some client
 $j'$ in $C^*$.  Since $y(F_{j'}) = 1$ we can serve $j$ with the
 facility for $j'$, and the cost is small if $j'$'s facility is close to $j$.  

To get intuition, assume each client $j$ can pay the farthest distance to
a facility in $F_j$, and call this the \emph{radius} of
$F_j$. (Precisely, clients may not be able to afford this distance,
  but we use this assumption to highlight the ideas behind our
  algorithmic decisions.) For simplicity, assume all radii are powers of two. Over time, this radius shrinks if some $y$
  variables in $F_j$ are set to zero. Consider applying the
  following iterative steps until none are applicable, in which case $C^*$ corresponds to the tight
  constraints: (1) delete a constraint for $j \notin  C^*$ if the
  radius of $F_j$ is at least that of some $F_{j'}$ for $j' \in C^*$ and $F_j \cap F_{j'} \neq \emptyset$. (2) add $j \notin C^*$  to $C^*$ if  $y(F_j) = 1$ and for every $j'\in C^*$ such that $F_j \cap F_{j'} \neq \emptyset$ it is the case that $F_{j'}$ has a radius strictly  larger than $F_{j}$. If added then remove all $j'$ from $C^*$ where $j$'s radius is \emph{half} or less of the radius of $j'$ and $F_j \cap F_{j'} \neq \emptyset$.
  
  

The approximation ratio is bounded by how much a client $j$ with a
deleted constraint pays to get to a facility serving a client in
$C^*$. After removing $j$'s constraint, the case to worry about is if $j$'s
closest client $j' \in C^*$ is later removed from $C^*$. This happens only if $j''$  is added to $C^*$, with $F_{j''}$ having half the radius of $F_{j'}$. Thus every time we remove $j$'s closest client in $C^*$, we guarantee that $j$'s cost only increases geometrically. The approximation ratio is proportional to the total distance that $j$ must travel and can be directly related to the distance of ``ball-chasing'' though these $F_j$ sets. See Figure~\ref{fig:chase}.


\begin{wrapfigure}{r}{0.5\textwidth}
\begin{minipage}{0.5\textwidth}
    \includegraphics[width=0.95\textwidth]{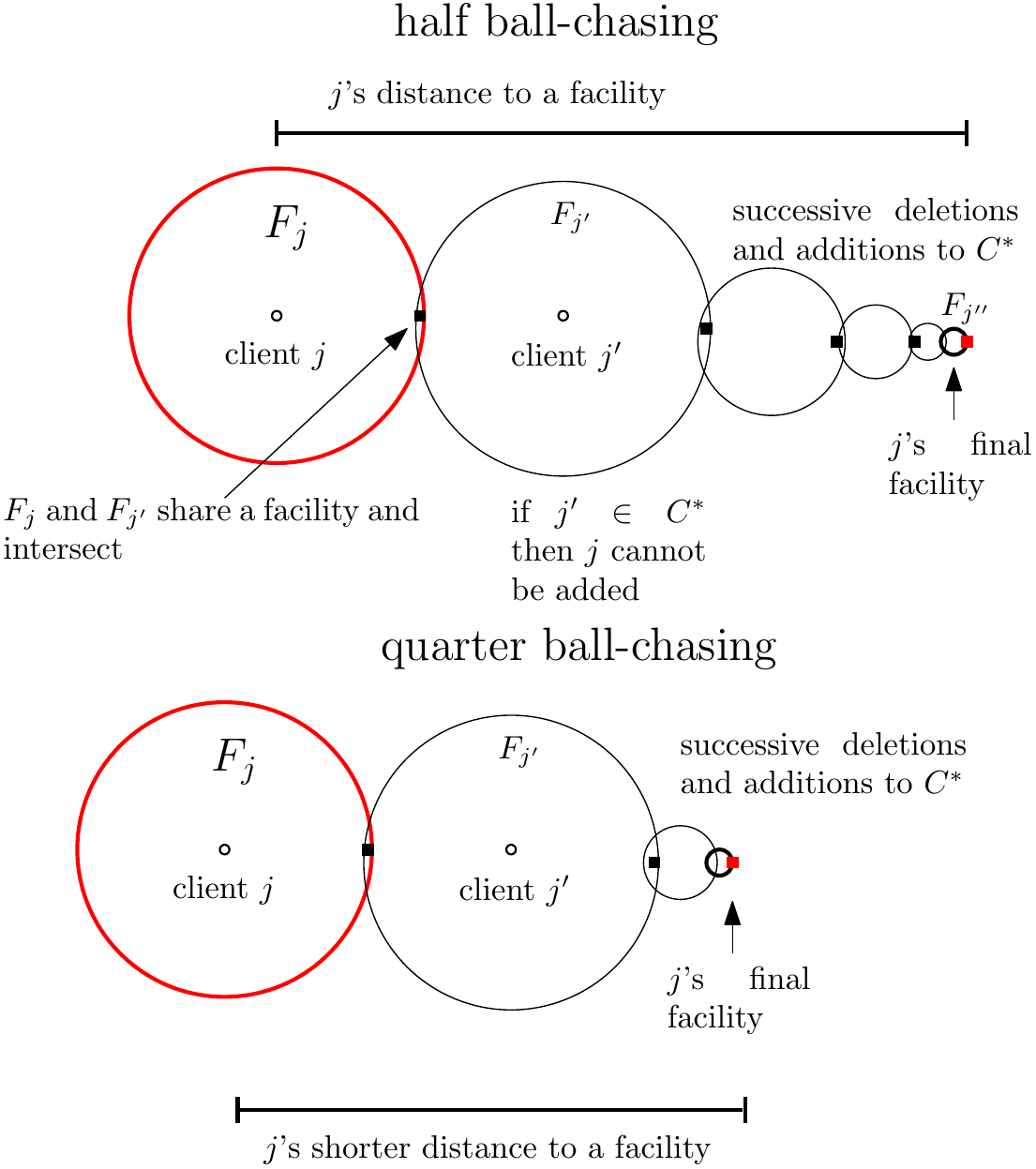}
      \caption{{\small Half and quarter ball chasing} \label{fig:chase}}
\end{minipage}\vspace{-.8cm}
\end{wrapfigure}

 \noindent \textbf{New Framework via Structured Extreme Points: } The
 target of our framework is to ensure that the radii decreases in the
 ball-chasing using a smaller factor, in particular
 \emph{one-quarter}.  This will give closer facilities for clients
 whose constraints are deleted and a better approximation ratio.  See
 Figure~\ref{fig:chase}. To achieve this ''quarter ball-chasing," we can simply change half to
 one-quarter in step (2) above.
 
 
 Making this change immediately decreases the approximation ratio; however, the challenge is that $\mathcal{F}$ is no longer disjoint. Indeed, it can be the case that $j,j' \in C^*$ such that $F_j \cap F_{j'} \neq \emptyset$ if their radii differ by only a one half factor. Instead, our quarter ball-chasing algorithm maintains that $\mathcal{F}$ is not disjoint, but has a \emph{bipartite intersection graph}.


The main technical challenge now is obtaining an extreme point with
$O(r)$ fractional variables, which is no longer guaranteed as when
$\mathcal{F}$ was disjoint. Indeed, if $\mathcal{F}$ has bipartite
intersection graph, then the tight constraints form a face of the
intersection of two partition matroid polytopes intersected with at
most $r$ side constraints. In general, we \emph{cannot upper bound the
  number of fractional variables} arising in the extreme points of
such polytopes. However, such extreme points have a nice combinatorial
structure:  
the intersection graph can be decomposed into $O(r)$ disjoint paths. We exploit this ``chain decomposition'' of extreme points arising in our iterative rounding to discover clients $j$ that can be removed from $C^*$ even if there is not a $j' \in C^*$ where $F_{j'}$ has one quarter of the radius of $F_j$. We continue this procedure until we are left with only $O(r)$ fractional variables.

The main technical contribution of this work is showing how the problem can be reduced to structural characterization of extreme points corresponding to bipartite matching. This illustrates some of the structural properties of polytopes defined by $k$-median-type problems. We hope that this helps lead to other structural characterizations of these polytopes and ultimately improved algorithms.

\subsection{Organization}

In \S \ref{sec_iteroverview}, we introduce the auxiliary LP for GKM
that our iterative rounding algorithm operates on. We note that this
is the same LP used in the algorithm of \cite
{DBLP:conf/stoc/KrishnaswamyLS18}. Then \S
\ref{sec_iter}--\ref{sec_pseudoalg} give the
pseudo-approximation for GKM. In particular, \S \ref{sec_iter}
describes the basic iterative rounding phase, where we iteratively
update the auxiliary LP such that $\mathcal{F}^* = \{F_j \mid j \in
C^*\}$ has a bipartite intersection graph. In \S
\ref{sec_rerouteleaf}, we characterize the structure of the resulting extreme points 
and use it 
to define a new iterative operation, 
which allows us to reduce the number of fractional variables to
$O(r)$. Finally, in \S \ref{sec_pseudoalg}, we combine the
algorithms 
from \S \ref{sec_iter} and \S\ref{sec_rerouteleaf}
to obtain our pseduo-approximation algorithm for GKM.

We then obtain true approximations for knapsack median and $k$-median with outliers: 
in \S \ref{sec_trueapprox}, we describe our framework to turn pseudo-approximation algorithms into true approximations for both problems, and apply it to knapsack median. Then in \S \ref{sec_postoutlier}, we give a more involved application of the same framework to $k$-median with outliers.

\section{Auxiliary LP for Iterative Rounding}\label{sec_iteroverview}

In this section, we construct the auxiliary LP, $\lpreroute$, that our algorithm will use. We note that we use the same relaxation used in \cite{DBLP:conf/stoc/KrishnaswamyLS18}. Recall the two goals of iterative rounding, outlined in \S \ref{sec_tech}; we want to maintain a set of clients $C^* \subset C$ such that $\{F_j \mid j \in C^*\}$ has bipartite intersection graph, and $C^*$ should provide a good set of open facilities for the clients that are not in $C^*$. Thus, we want to define $\lpreroute$ to accommodate moving clients in and out of $C^*$, while having the LP faithfully capture how much we think the clients outside of $C^*$ should pay in connection costs. For all missing proofs in this section, see \S \ref{sec_appendix_LP}.

\subsection{Defining $F$-balls}

Our starting point is $\lpfacil$, so we assume that we have sets $F_j \subset F$ for all $j \in C$. The next proposition states that such sets can be found efficiently so that $\lpfacil$ is a relaxation of GKM.

\begin{prop}\label{prop_LPfacil}
	There exists a polynomial time algorithm that given GKM instance $\I$, duplicates facilities and outputs sets $F_j \subseteq F$ for $j \in C$ such that $Opt(LP_2) \leq Opt(\I)$.
\end{prop}

In \S \ref{sec_tech}, we assumed the radii of the $F_j$ sets were powers of two. To formalize this idea, we discretize the distances to powers of $\tau > 1$ (up to some random offset.) The choice of $\tau$ is to optimize the final approximation ratio. The main ideas of the algorithm remain the same if we discretize to powers of, say $2$, with no random offset. Our discretization procedure is the following: 

Fix some $\tau > 1$ and sample the random offset $\alpha \in [1, \tau)$ such that $\log_e \alpha$ is uniformly distributed in $[0, \log_e \tau)$. Without loss of generality, we may assume that the smallest non-zero inter-point distance is $1$. Then we define the possible discretized distances, $L(-2) = -1, L(-1) = 0, \dots, L(\ell) = \alpha \tau^\ell$ for all $\ell \in \mathbb{N}$.

For each $p,q \in F \cup C$, we round $d(p,q)$ \emph{up} to the next largest discretized distance. Let $d'(p,q)$ denote the rounded distances. Observe that $d(p,q) \leq d'(p,q)$ for all $p,q \in F \cup C$. See \S \ref{sec_appendix_LP} for proof of the following proposition, which we use to bound the cost of discretization.

\begin{prop}\label{lem_disc}
	For all $p,q \in F \cup C$, we have $\mathbb{E}[d'(p,q)] = \frac{\tau - 1}{\log_e \tau} d(p,q)$
\end{prop}

Now using the discretized distances, we can define the \emph{radius level} of $F_j$ for all $j \in C$ by:
\[\ell_j = \min\limits_{\ell \geq -1} \{\ell \mid d'(j,i) \leq L(\ell)
  \quad \forall i \in F_j\}.\]
One should imagine that $F_j$ is a ball of radius $L(\ell_j)$ in terms of the $d'$-distances. Thus, we will often refer to $F_j$ as the \emph{$F$-ball of client $j$}. Further, to accommodate ``shrinking" the $F_j$ sets, we define the \emph{inner ball of $F_j$} by: 
\[B_j = \{i \in F_j \mid d'(j,i) \leq L(\ell_j - 1)\}.\]
Note that we defined $L(-2) = -1$ so that  if $\ell_j = -1$, then $B_j = \emptyset$.

\subsection{Constructing $\lpreroute$}

Our auxiliary LP will maintain three sets of clients: $C_{part}, C_{full}$, and $C^*$. $C_{part}$ consists of all clients, whom we have not yet decided whether we should serve them or not. Then for all clients in $C_{full}$ and $C^*$, we decide to serve them fully. The difference between the clients in $C_{full}$ and $C^*$ is that for the former, we remove the constraint $y(F_j) = 1$ from the LP, while for the latter we still require $y(F_j) = 1$. Thus although we commit to serving $C_{full}$, such clients rely on $C^*$ to find an open facility to connect to. Using the discretized distances, radius levels, inner balls, and these three sets of clients, we are ready to define $\lpreroute$:

\begin{align*}\tag{$\lpreroute$}
	\min\limits_y~~ &\sum\limits_{j \in C_{part}} \sum\limits_{i \in F_j} d'(i,j) y_i + \sum\limits_{j \in C_{full \cup C^*}} (\sum\limits_{i \in B_j} d'(i,j) y_i + (1 - y(B_j)) L(\ell_j)) \\
	\text{s.t.}~~ &y(F_j) \leq 1 \quad \forall j \in C_{part} \\
	&y(B_j) \leq 1 \quad \forall j \in C_{full} \\
	&y(F_j) = 1 \quad \forall j \in C^* \\
	&Wy \leq b\\
	&\sum\limits_{j \in C_{part}} a_j y(F_j) \geq c - \sum\limits_{j \in C_{full} \cup C^*} a_j\\
	&0 \leq y \leq 1
\end{align*}

Note that we use the \emph{rounded} distances in the definition of $\lpreroute$ rather than the original distances. Keeping this in mind, if $C_{part} = C$ and $C_{full}, C^* = \emptyset$, then $\lpreroute$ is the same as $\lpfacil$ up to the discretized distances, so the following proposition is immediate.

\begin{prop}\label{prop_LPreroute}
	Suppose $C_{part} = C$ and $C_{full}, C^* = \emptyset$. Then $\mathbb{E}[Opt(\lpreroute)] \leq \frac{\tau - 1}{\log_e \tau} Opt(\lpfacil)$.
\end{prop}

We now take some time to parse the definition of $\lpreroute$. Initially, all clients are in $C_{part}$. For clients in $C_{part}$, we are not sure yet whether we should serve them or not. Thus for these clients, we simply require $y(F_j) \leq 1$, so they can be served any amount, and in the objective, the contribution of a client from $C_{part}$ is exactly its connection cost (up to discretization) to $F_j$.

The clients in $C_{full}$ correspond to the ``deleted" constraints in \S \ref{sec_tech}. Importantly, for $j \in C_{full}$, we do not require that $y(F_j) = 1$; rather, we relax this condition to $y(B_j) \leq 1$. Recall that we made the assumption that every client can pay the radius of its $F_j$ set in \S \ref{sec_tech}. To realize this idea, we require that each $j \in C_{full}$ pays its connection costs to $B_j$ in the objective. Then, to serve $j$ fully, $j$ must find $(1 - y(B_j))$ units of open facility to connect to beyond $B_j$. Now $j$ truly pays its radius, $L(\ell_j)$, for this $(1 - y(B_j))$ units of connections in $\lpreroute$, so we can do ``ball-chasing" to $C^*$ to find these facilities. In this case, we say that we \emph{re-route} the client $j$ to some \emph{destination}.


For clients in $C^*$, we require $y(F_j) = 1$. Note that the contribution of a $j \in C^*$ to the objective of $\lpreroute$ is exactly its connection cost to $F_j$. The purpose of $C^*$ is to provide destinations for $C_{full}$.

Finally, because we have decided to fully serve all clients in $C_{full}$ and $C^*$, regardless of how much they are actually served in their $F$-balls, we imagine that they every $j \in C_{full} \cup C^*$ contributes $a_j$ to the coverage constraints, which is reflected in $\lpreroute$.

\subsection{Properties of $\lpreroute$}

Throughout our algorithm, we will modify the data of $\lpreroute$ - we will move clients between $C_{part}, C_{full}$, and $C^*$ and modify the $F$-balls and radius levels. However, we still want the data of $\lpreroute$ to satisfy some consistent properties, which we call our \emph{Basic Invariants}.

\begin{definition}[Basic Invariants] \label{invar_basic}
	We call the following properties our \emph{Basic Invariants}:
	
	\begin{enumerate}
		\item $C_{part} \cup C_{full} \cup C^*$ partitions $C$.	\label{invar_partition}
		\item For all $j \in C$, we have $d'(j,i) \leq L_{\ell_j}$ for all $i \in F_j$.	\label{invar_F}
		\item For all $j \in C$, we have $B_j = \{i \in F_j \mid d'(j,i) \leq L_{\ell_j - 1} \}$.	\label{invar_B}
		\item For all $j \in C$, we have $\ell_j \geq -1$.	\label{invar_ell}
		\item (Distinct Neighbors) For all $j_1, j_2 \in C^*$, if $F_{j_1} \cap F_{j_2} \neq \emptyset$, then $\lvert \ell_{j_1} - \ell_{j_2} \rvert = 1$. In words, if the $F$-balls of two clients in $C^*$ intersect, then they differ by exactly one radius level.	\label{invar_neighbor}
	\end{enumerate}
\end{definition}

We want to emphasize Basic Invariant \ref{invar_basic}(\ref{invar_neighbor}), which we
call the \emph{Distinct Neighbors Property}. It is not difficult to see that the Distinct Neighbors Property implies that $\{F_j \mid j \in C^* \}$ has bipartite intersection graph.

\begin{definition}[Intersection Graph]\label{def_itersectiongraph}
    Let $\mathcal{F} = \{F_j \mid j \in C^*\}$ be a set family indexed by $C^*$. The intersection graph of $\mathcal{F}$ is the undirected graph with vertex set $C^*$ such that two vertices $j$ and $j'$ are connected by an edge if any only if $F_j \cap F_{j'} \neq \emptyset$.
\end{definition}

\begin{prop}\label{prop_bipartite}
	Suppose $\lpreroute$ satisfies the Distinct Neighbors Property. Then the intersection graph of $\mathcal{F} = \{F_j \mid j \in C^*\}$ is bipartite.
\end{prop}

The following proposition will also be useful.

\begin{prop}\label{prop_degree}
	Suppose $\lpreroute$ satisfies the Distinct Neighbors Property. Then each facility is in at most two $F$-balls for clients in $C^*$.
\end{prop}

We summarize the relevant properties of $\lpreroute$ in the following lemma. The algorithm described by the lemma is exactly the steps we took in this section.

\begin{lem}\label{lem_makeLPbasic}
	There exists a polynomial time algorithm that takes as input a GKM instance $\I$ and outputs $\lpreroute$ such that $\mathbb{E}[Opt(\lpreroute)] \leq \frac{\tau - 1}{\log_e \tau} Opt(I)$ and $\lpreroute$ satisfies all Basic Invariants.
\end{lem}

\subsection{Notations}

Throughout this paper, we will always index facilities using $i$ and clients using $j$.

For any client $j \in C$, we say that $j$ is \emph{supported on} facility $i \in F$ if $i \in F_j$. Then for any $C' \subset C$, we let $F(C') \subset F$ be the set of all facilities supported on at least one client in $C'$.

Given a setting of the $y$-variables of $\lpreroute$, we say a facility $i$ is \emph{fractional} (with respect to the given $y$-variables) if $y_i < 1$. Otherwise, facility $i$ is \emph{integral}. Similarly, we say a client $j$ is fractional if $F_j$ contains only fractional facilities, and $j$ is integral otherwise. Using these definitions, for any $F' \subset F$, we can partition $F'$ into $F'_{<1} \cup F'_{=1}$, where $F'_{<1}$ is the subset of fractional facilities and $F'_{=1}$ is the subset of integral facilities. An analogous partition holds for a subset of clients $C' \subset C$, so we have $C' = C'_{<1} \cup C'_{=1}$.

\section{Basic Iterative Rounding Phase}\label{sec_iter}

In this section, we describe the iterative rounding phase of our
algorithm. This phase has two main goals: (a)~to simplify the
constraint set of $\lpreroute$, and (b)~to decide which clients to
serve and how to serve them. To make these two decisions, we
repeatedly solve $\lpreroute$ to obtain an optimal extreme point, and
then use the structure of tight constraints to update $\lpreroute$,
and reroute clients accordingly.

\subsection{The Algorithm}

Our algorithm repeatedly solves $\lpreroute$ to obtain an optimal
extreme point $\bar{y}$, and then performs one of the following three
possible updates, based on the tight constraints:

\begin{enumerate}
    \item If some facility $i$ is set to zero in $\bar{y}$, we delete it from the instance.
    \item If constraint $\bar{y}(F_j) \leq 1$ is tight for some $j \in
      C_{part}$, then we decide to fully serve client $j$ by moving
      $j$ to either $C_{full}$ or $C^*$. Initially, we add $j$ to $C_{full}$ then run  Algorithm \ref{alg_reroute} to decide if $j$ should be in $C^*$ instead. 
    \item If constraint $\bar{y}(B_j) \leq 1$ is tight for some $j \in C_{full}$,  we shrink $F_j$ by one radius level (so $j$'s new $F$-ball is exactly $B_j$.) Then we possibly move $j$ to $C^*$ by running Algorithm \ref{alg_reroute} for $j$.
\end{enumerate}

These steps are made formal in Algorithms~\ref{alg_iterativeround}
(\iteralg) and \ref{alg_reroute} (\itersub). \iteralg relies on the
subroutine \itersub, which gives our criterion for moving a client to
$C^*$. This criterion for adding clients to $C^*$ is the key way in
which our algorithm differs from that of~\cite{DBLP:conf/stoc/KrishnaswamyLS18}. In~\cite{DBLP:conf/stoc/KrishnaswamyLS18},
the criterion used ensures that $\{F_j \mid j \in C^*\}$ is a family
of disjoint sets. In contrast, we allow $F$-balls for clients in $C^*$
to intersect, as long as they satisfy the Distinct Neighbors Property
from \Cref{invar_basic}(\ref{invar_neighbor}). Thus, our algorithm
allows for rich structures in the set system
$\{F_j \mid j \in C^*\}$.

\begin{algorithm}
	\DontPrintSemicolon
	\LinesNumbered
	\SetAlgoLined
	
	\KwIn{$\lpreroute$ satisfying all Basic Invariants}
	\KwResult{Modifies $\lpreroute$ and outputs an optimal extreme point of $\lpreroute$}
	\BlankLine
	
	\Repeat{termination}{
		Solve $\lpreroute$ to obtain optimal extreme point $\bar{y}$.\;
		\If{there exists a facility $i \in F$ such that $\bar{y}_i \geq 0$ is tight}
			{
			Delete $i$ from $F$.\;
			}
		\ElseIf{there exists a client $j \in C_{part}$ such that $y(F_j) \leq 1$ is tight}
			{
			Move $j$ from $C_{part}$ to $C_{full}$.\;
			$\textsc{ReRoute}(j)$\;
			}
		\ElseIf{there exists a client $j \in C_{full}$ such that $\bar{y}(B_j) \leq 1$ is tight}
			{
			Update $F_j \gets B_j$ and decrement $\ell_j$ by $1$.\;
			Update $B_j \gets \{i \in F_j \mid d'(j,i) \leq L(\ell_j - 1) \}$.\;
			$\textsc{ReRoute}(j)$\;
			}
		\Else {Output $\bar{y}$ and Terminate.}
	}
	\caption{\iteralg \label{alg_iterativeround}}
\end{algorithm}

\begin{algorithm}
	\DontPrintSemicolon
	\LinesNumbered
	\SetAlgoLined
	
	\KwIn{Client $j \in C_{full}$}
	\KwResult{Decide whether to move $j$ to $C^*$ or not}

	\If {$\ell_j \leq \ell_{j'} - 1$ for all $j' \in C^*$ such that $F_j \cap F_{j'} \neq \emptyset$}
	{
		Move $j$ from $C_{full}$ to $C^*$.\;
		For all $j' \in C^*$ such that $F_j \cap F_{j'} \neq \emptyset$ and $\ell_{j'} \geq \ell_j + 2$, move $j'$ from $C^*$ to $C_{full}$.\;
	}
	
	\caption{\itersub \label{alg_reroute}}	
\end{algorithm}

The modifications made by \iteralg do not increase $Opt(\lpreroute)$,
so upon termination of our algorithm, we have an optimal extreme point
$\bar{y}$ to $\lpreroute$ such that $\lpreroute$ is still a relaxation
of GKM and no non-negativity constraint, $C_{part}$-constraint, or
$C_{full}$-constraint is tight for $\bar{y}$. This is formalized in
the following theorem, whose proof is similar to~\cite{DBLP:conf/stoc/KrishnaswamyLS18}, and is
deferred to Appendix~\ref{sec_iteranalysis}.

\begin{thm}\label{thm_mainiter}
  \iteralg~ is a polynomial time algorithm that maintains all Basic
  Invariants, weakly decreases $Opt(\lpreroute)$, and outputs an
  optimal extreme point to $\lpreroute$ such that no $C_{part}$-,
  $C_{full}$-, or non-negativity constraint is tight.
\end{thm}


Recall the goals from the beginning of the section: procedure \iteralg
achieves goal (a) of making $\{F_j \mid j \in C^*\}$ simpler while
maintaining the Distinct Neighbors Property. Since we moved facilities
between $C^*$ and $C_{full}$, achieving goal~(b) means deciding which
facilities to open, and guaranteeing that each client has a
``close-by'' open facility. (Recall from \S\ref{sec_iteroverview} that
$C^*$ is the set of clients such that their $F_j$-balls are guaranteed
to contain an open facility, and $C_{full}$ are the clients which are
guaranteed to be served but using facilities opened in $C^*$.)

Here's the high-level idea of how we achieve goal (b). Suppose we move $j$
from $C_{full}$ to $C^*$ and some $j'$ from $C^*$ to $C_{full}$; we
want to find a good destination for $j'$. We claim $j$'s facility is a
good destination for $j'$. Indeed, since $j$ is now in $C^*$, we can
use the constraint $y(F_j) = 1$ to bound the distance of $j'$ to this
unit of facility by
$L(\ell_{j'}) + 2L(\ell_j) \leq (1 + \frac{2}{\tau^2}) L(\ell_{j'})$,
using the facts that $\ell_{j'} \geq \ell_j + 2$ and
$F_j \cap F_{j'} \neq \emptyset$, which are guaranteed by \itersub.
Of course, if $j$ is removed from $C^*$ later, we re-route it to some
client that is at least two radius levels smaller, and can send $j'$
to that client. This corresponds to the ``quarter ball-chasing" of \S \ref{sec_tech}. Indeed every further re-routing step for $j'$ has
geometrically decreasing cost, which give a cost of
$O(1) L(\ell_{j'})$.  We defer the formal analysis to
\Cref{thm_mainreroute}, after we combine \iteralg~ with another (new)
iterative operation, which we present in the next section.


\section{Iterative Operation for Structured Extreme Points}\label{sec_rerouteleaf}


In this section, we achieve two goals: (a)~we show that the structure
of the extreme points of $\lpreroute$ obtained from
\Cref{thm_mainiter} are highly structured, and admit a \emph{chain
  decomposition}. Then, (b) we exploit this chain decomposition to
define a \emph{new} iterative operation that is applicable whenever $\bar{y}$
has ``many'' (i.e., more than $O(r)$) fractional variables.  We emphasize that this characterization of the extreme points is what enables the new iterative rounding algorithm.
  
\subsection{Chain Decomposition}
\label{sec:chain-decomposition}

A chain is a sequence of clients in $C^*$ 
where the $F$-ball of each client $j$ 
contains exactly two facilities---one shared with the previous ball 
and other with the next. 
\begin{definition}[Chain]\label{def_chain}
	A \emph{chain} is a sequence of clients $(j_1, \dots, j_p) \subseteq C^*$ satisfying:
	\begin{itemize}
		\item $\lvert F_{j_q} \rvert = 2$ for all $q \in [p]$, and
		\item $F_{j_q} \cap F_{j_{q+1}} \neq \emptyset$ for all $q \in [p - 1]$.
	\end{itemize}
\end{definition}

Our chain decomposition is a partition of the \emph{fractional}
$C^*$-clients given in the next theorem, which is our main structural
characterization of the extreme points of $\lpreroute$. (Recall that a
client $j$ is fractional if all facilities in $F_j$ are fractional; we
denote the fractional clients in $C^*$ by $C^*_{< 1}$.)

\begin{thm}[Chain Decomposition]\label{thm_chaindecomp}
  Suppose $\lpreroute$ satisfies all Basic Invariants. Let $\bar{y}$
  be an extreme point of $\lpreroute$ such that no $C_{part}$-,
  $C_{full}$-, or non-negativity constraint is tight. Then there
  exists a partition of $C^*_{<1}$ into at most $3r$ chains, along with
  a set of at most $2r$ violating clients (clients that are not in any chain.)
\end{thm}

The proof relies on analyzing the extreme points of a set-cover-like
polytope with $r$ side constraints; we defer it to
\S\ref{sec_pathdecomp} and proceed instead to define the new iterative
operation.

\subsection{Iterative Operation for Chain Decompositions}
\label{sec:iter-oper-chain}

Composing~\Cref{thm_mainiter} and \Cref{thm_chaindecomp}, consider 
an optimal extreme point $\bar{y}$ of $\lpreroute$, and a chain
decomposition.
We show that if the number of fractional variables in $\bar{y}$ is
sufficiently large, there exists a useful structure in the chain
decomposition, which we call a \emph{candidate configuration}.

\begin{definition}[Candidate Configuration]\label{def_config}
	Let $\bar{y}$ be an optimal extreme point of $\lpreroute$. A candidate configuration is a pair of two clients $(j,j') \subset C^*_{<1}$ such that:
	\begin{enumerate}
		\item $F_j \cap F_{j'} \neq \emptyset$
		\item $\ell_{j'} \leq \ell_j - 1$
        \item Every facility in $F_j$ and $F_{j'}$ is in at exactly two $F$-balls for clients in $C^*$
        \item $\lvert F_j \rvert = 2$ and $\lvert F_{j'} \rvert = 2$ 
	\end{enumerate}
\end{definition}

\begin{lem}\label{lem_configfacil}
  Suppose $\lpreroute$ satisfies all Basic Invariants, and let
  $\bar{y}$ be an optimal extreme point of $\lpreroute$ such that no
  $C_{part}$-, $C_{full}$-, or non-negativity constraint is tight. If
  $\lvert F_{<1} \rvert \geq 15r$, then there exist a candidate
  configuration in $C^*_{<1}$.
\end{lem}

Our new iterative operation is easy to state: Find  a candidate configuration $(j,j')$ and move $j$ from $C^*$ to $C_{full}$.

\begin{algorithm}
	\DontPrintSemicolon
	\LinesNumbered
	\SetAlgoLined
	
	\KwIn{An optimal extreme point $\bar{y}$ to $\lpreroute$ s.t.\
          there exists an candidate configuration}
	\KwResult{Modify $\lpreroute$}
	
	Let $(j,j') \subset C^*_{<1}$ be any candidate configuration.\;
	Move $j$ from $C^*$ to $C_{full}$.\;
	\caption{\altreroute \label{alg_altreroute}}	
\end{algorithm}

The first two properties of candidate configurations are used to
re-route $j$ to $j'$. Observe
a key difference between \itersub~ and \altreroute: In the former,
if a client $j'$ is moved from $C^*$ to $C_{full}$, there exists a client $j \in C^*$ such that $F_{j'} \cap F_j \neq \emptyset$ and $\ell_j \leq \ell_{j'} - 2$. Thus we re-route $j'$ to a client at least two radius levels smaller. This corresponds to ``quarter ball-chasing." On the other hand, in \altreroute, we only guarantee
a client of at least one radius level smaller, which corresponds to ``half ball-chasing." This raises the worry
that if all 
re-routings 
are due to \altreroute, any potential gains by \itersub are not
realized in the worst case. However we show that, roughly speaking, the last two properties of candidate configurations guarantee that 
the more expensive re-routings of \altreroute~ happen at most half the time.
The main properties of \altreroute~ appear in the next theorem (whose
proof is in Appendix~\ref{sec_altanalysis}).

\begin{thm}\label{thm_mainaltreroute}
   \altreroute~ is a polynomial-time algorithm that maintains all Basic Invariants and weakly decreases $Opt(\lpreroute)$.
\end{thm}

Again, we defer the analysis of the re-routing cost of \altreroute~to
\S\ref{thm_mainreroute}, where we analyze the interactions between
\altreroute~ and \itersub, and present our final pseudo-approximation
algorithm next.
%


\section{Pseudo-Approximation Algorithm for GKM}\label{sec_pseudoalg}

The pseudo-approximation algorithm for GKM 
combine the iterative rounding algorithm \iteralg~ from
\S\ref{sec_iter} with the re-routing operation \altreroute~ from
\S\ref{sec_rerouteleaf} to construct a solution to $\lpreroute$.

\begin{thm}[Pseudo-Approximation Algorithm for GKM]\label{thm_pseudooverview}
	There exists an polynomial time randomized algorithm \mainalg that takes as input an instance $\I$ of GKM and outputs a feasible solution to $\lpbasic$ with at most $O(r)$ fractional facilities and expected cost at most $6.387 \cdot Opt(\I)$.
\end{thm}

\begin{algorithm}
	\DontPrintSemicolon
	\LinesNumbered
	\SetAlgoLined
	
	\KwIn{$\lpreroute$ satisfying all Basic Invariants}
	\KwResult{Modifies $\lpreroute$ and outputs an optimal extreme point of $\lpreroute$}
	\BlankLine
	\Repeat{Termination} {
		Run \iteralg~ to obtain an optimal extreme point $\bar{y}$ of $\lpreroute$\;
		\If{there exists a candidate configuration }{
			Run \altreroute
		}
		\Else{
			Output $\bar{y}$ and Terminate	
		}
	}
	\caption{\mainalg\label{alg_mainalg}}
\end{algorithm}

There are two main components to analyzing \mainalg. First, we 
show that the output extreme point has $O(r)$ fractional
variables. Second, we 
 bound the re-routing cost. 
The first part follows directly by combining the analogous theorems for \iteralg~ and \altreroute. We defer its proof to Appendix \ref{appendix_pseudoanalysis}.

\begin{thm}\label{thm_pseudo}
    \mainalg~ is a polynomial time algorithm that maintains all Basic Invariants, weakly decreases $Opt(\lpreroute)$, and outputs an optimal extreme point of $\lpreroute$ with at most $15r$ fractional variables.
\end{thm}

\subsection{Analysis of Re-Routing Cost}

We now bound the re-routing cost by analyzing how $C^*$
evolves throughout \mainalg. This is one of the main technical
contributions of our paper, and it is where our richer $C^*$-set and
relaxed re-routing rules are used. 
\cite{DBLP:conf/stoc/KrishnaswamyLS18} prove an analogous result about the re-routing cost of
their algorithm. In the language of the following theorem statement, they
show that $\alpha = \frac{\tau + 1}{\tau - 1}$ for the case $\beta =
1$. We improve on this factor by analyzing the interactions between
\itersub~ and \altreroute . Interestingly, analyzing each of \itersub~
and \altreroute~ separately would not yield any improvement over
\cite{DBLP:conf/stoc/KrishnaswamyLS18} in the worst case, even with our richer set $C^*$. It is
only by using the properties of candidate configurations and analyzing
sequences of calls to \itersub~ and \altreroute that we get an improvement.

\begin{thm}[Re-Routing Cost]\label{thm_mainreroute}
    Upon termination of \mainalg, let $S \subset F$ be a set of open facilities and $\beta \geq 1$ such that $d(j,S) \leq \beta L(\ell_j)$ for all $j \in C^*$. Then for all $j \in C_{full} \cup C^*$, $d(j, S) \leq (2 + \alpha) L(\ell_j)$, where $\alpha = \max(\beta, 1 + \frac{1 + \beta}{\tau}, \frac{\tau^3 + 2\tau^2 + 1}{\tau^3 - 1})$.
\end{thm}

We will need the following discretized version of the triangle inequality. 

\begin{prop}\label{prop_distintersect}
    Let $j, j' \in C$ such that $F_j$ and $F_{j'}$ intersect. Then $d(j,j') \leq L(\ell_j) + L(\ell_{j'})$.
\end{prop}
\begin{proof}
    Let $i \in F_j \cap F_{j'}$. Then using the triangle inequality we can bound:
    \begin{gather*}
      d(j,j') \leq d(j,i) + d(i,j') \leq d'(j,i) + d'(i,j') \leq
      L(\ell_j) + L(\ell_{j'}). \qedhere
    \end{gather*}
\end{proof}

\begin{figure}[t]
\begin{center}
  \includegraphics[width=.5\textwidth]{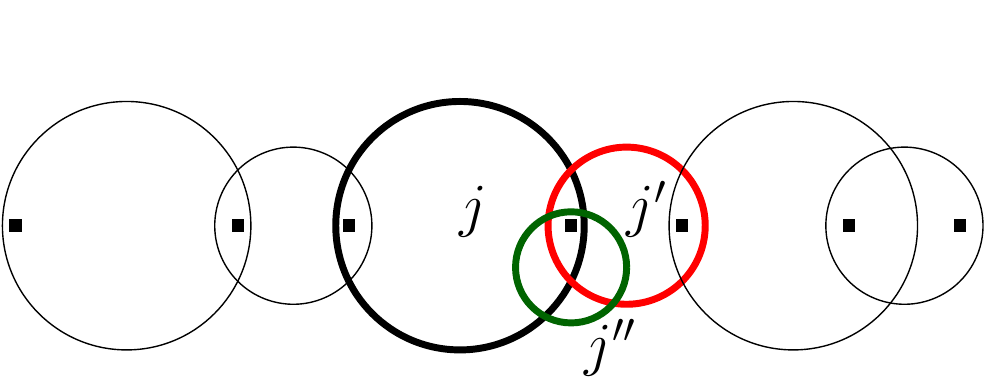}
  \caption{A chain of balls in $C^*$, where squares indicate facilities. First $j$ is removed from $C^*$ as part of candidate configuration $(j,j')$, so $j'$ has strictly smaller radius than $j$. Then $j''$ is added to $C^*$, which has strictly smaller radius than $j'$. This gives $j$ a destination that is at least two radius levels smaller.}\label{fig:chain}
\end{center}
\end{figure}

The next lemma analyzes the life-cycle of a client that enters $C^*$
at some point in \mainalg. Our improvement over \cite{DBLP:conf/stoc/KrishnaswamyLS18} comes from
this lemma.

\begin{lem}\label{lem_chase}
     Upon termination of \mainalg, let $S \subset F$ be a set of open facilities and $\beta \geq 1$ such that $d(j,S) \leq \beta L(\ell_j)$ for all $j \in C^*$. Suppose client $j$ is added to $C^*$ at radius level $\ell$ during \mainalg~ (it may be removed later.) Then upon termination of \mainalg, we have $d(j,S) \leq \alpha L(\ell)$, where $\alpha = \max(\beta, 1 + \frac{1 + \beta}{\tau}, \frac{\tau^3 + 2\tau^2 + 1}{\tau^3 - 1})$.
\end{lem}
\begin{proof}
    Consider a client $j$ added to $C^*$ with radius level $\ell$. If $j$ remains in $C^*$ until termination, the lemma holds for $j$ because $\alpha \geq \beta$. Thus, consider the case where $j$ is later removed from $C^*$ in \mainalg. Note that the only two operations that can possibly cause this removal are \itersub~ and \altreroute.
        We  prove the lemma by induction on $\ell = -1, 0, \dots$. If $\ell = -1$, then $j$ remains in $C^*$ until termination because it has the smallest possible radius level and both \itersub~ and \altreroute~ remove a client from $C^*$ only if there exists another client with strictly smaller radius level.
    
    Similarly, if $\ell = 0$, we note that \itersub~ removes a client from $C^*$ only if there exists another client with radius level at least two smaller, which is not possible for $j$. Thus, if $j$ does not remain in $C^*$ until termination, there must exist some $j'$ that is later added to $C^*$ with radius level at most $\ell - 1 = -1$ such that $F_j \cap F_{j'} \neq \emptyset$. We know that $j'$ remains in $C^*$ until termination since it is of the lowest radius level. Thus:
    \[d(j,S) \leq d(j,j') + d(j', S) \leq L(0) + L(-1) + \beta L(-1) = L(0).\]
    
    Now consider $\ell > 0$ where $j$ can possibly be removed from $C^*$ by either \itersub~ or \altreroute. In the first case, $j$ is  removed by \itersub, so there exists $j'$ that is added to $C^*$ such that $\ell_{j'} \leq \ell - 2$ and $F_j \cap F_{j'} \neq \emptyset$. Applying the inductive hypothesis to $j'$, we can bound:
    \[d(j,S) \leq d(j,j') + d(j', S) \leq L(\ell) + L(\ell - 2) + \alpha L(\ell - 2) \leq (1 + \frac{1 + \alpha}{\tau^2})L(\ell).\]
    It is easy to verify by routine calculations that $1 + \frac{1 + \alpha}{\tau^2} \leq \alpha$ given that $\alpha \geq \frac{\tau^3 + 2\tau^2 + 1}{\tau^3 - 1}$.
    
    For our final  case,  suppose $j$ is  removed by \altreroute. Then  there exists $j' \in C^*$ such that $F_j \cap F_{j'} \neq \emptyset$ and $\ell_{j'} \leq \ell - 1$. Further, $\lvert F_{j'} \rvert = 2$. If $j'$ remains in $C^*$ until termination, then:
    \[d(j,S) \leq d(j,j') \leq L(\ell) + L(\ell - 1) + \beta L(\ell - 1) \leq (1 + \frac{1 + \beta}{\tau})L(\ell).\]
    Otherwise,  $j'$ is  removed by \itersub~ at an even later time because some $j''$ is added to $C^*$ such that $\ell_{j''} \leq \ell_{j'} - 2$ and $F_{j'} \cap F_{j''} \neq \emptyset$. Applying the inductive hypothesis to $j''$, we can bound:
    \[d(j,S) \leq d(j, j') + d(j', j'') + d(j'', S) \leq (1 + \frac{2}{\tau} + \frac{1 + \alpha}{\tau^3})L(\ell).\]
     where $\alpha \geq \frac{\tau^3 + 2\tau^2 + 1}{\tau^3 - 1}$ implies $1 + \frac{2}{\tau} + \frac{1 + \alpha}{\tau^3} \leq \alpha$.
    
    Now, we consider the case where $j'$ is later removed by \altreroute. To analyze this case, consider when $j$ was removed by \altreroute. At this time, we have $\lvert F_{j'} \rvert = 2$ by definition of Candidate Configuration. Because $F_j \cap F_{j'} \neq \emptyset$,  consider any facility $i \in F_j \cap F_{j'}$. When $j$ is removed from $C^*$ by \altreroute, we have that $i$ is in exactly two $F$-balls for clients in $C^*$,  exactly $F_j$ and $F_{j'}$. However, after removing $j$ from $C^*$, $i$ is only in one $F$-ball for clients in $C^*$ - namely $F_{j'}$.
    
    Later, at the time $j'$ is removed by \altreroute, it must be the case that $\lvert F_{j'} \rvert = 2$ still, so $F_{j'}$ is unchanged between the time that $j$ is removed and the time that $j'$ is removed. Thus the facility $i$ that was previously in $F_j \cap F_{j'}$ must still be present in $F_{j'}$. Then this facility must be in exactly two $F$-balls for clients in $C^*$, one of which is $j'$. It must be the case that the other $F$-ball containing $i$, say $F_{j''}$, was added to $C^*$ between the removal of $j$ and $j'$.
    
    Note that the only operation that adds clients to $C^*$ is \itersub, so we consider the time between the removal of $j$ and $j'$ when $j''$ is added to $C^*$. Refer to Figure~\ref{fig:chain}. At this time, we  have $j' \in C^*$, and $F_{j'} \cap F_{j''} \neq \emptyset$ because of the facility $i$. Then it must be the case that $j''$ has strictly smaller radius level than $j'$, so $\ell_{j''} \leq \ell_{j'} - 1 \leq \ell - 2$. To conclude the proof, we note that $F_j \cap F_{j''} \neq \emptyset$ due to the facility $i$, and apply the inductive hypothesis to $j''$:
    \[d(j,S) \leq d(j,j'') + d(j'', S) \leq (1 + \frac{1 + \alpha}{\tau^2})L(\ell,)\]
    which is at most $\alpha L(\ell)$.
\end{proof}

Now using the above lemma, we can prove Theorem \ref{thm_mainreroute}.

\begin{proof}[Proof of Theorem \ref{thm_mainreroute}]
    Consider any client $j$ that is in $C_{full} \cup C^*$ upon termination of \mainalg. It must be the case that $\itersub(j)$ was called at least once during \mainalg. Consider the time of the last such call to $\itersub(j)$. If $j$ is added to $C^*$ at this time, note that its radius level from now until termination remains unchanged, so applying Lemma \ref{lem_chase} gives that $d(j,S) \leq \alpha L(\ell_j)$, as required.
        Otherwise, if $j$ is not added to $C^*$ at this time, then there must exist some $j' \in C^*$ such that $F_j \cap F_{j'} \neq \emptyset$ and $\ell_{j'} \leq \ell_j$. Then applying Lemma \ref{lem_chase} to $j'$, we have:
        \begin{gather*}
          d(j,S) \leq d(j,j') + d(j', S) \leq L(\ell_j) + L(\ell_{j'})
          + \alpha L(\ell_{j'}) \leq (2 + \alpha)L(\ell_j). \qedhere
        \end{gather*}
\end{proof}

\subsection{Putting it all Together: Pseudo-Approximation for GKM}

In this section, we prove Theorem \ref{thm_pseudooverview}. In particular, we use the output of \mainalg~ to construct a setting of the $x$-variables with the desired properties.

\begin{proof}[Proof of Theorem \ref{thm_pseudooverview}]
	Given as input an instance $\I$ of GKM, our algorithm is first to run the algorithm guaranteed by Lemma \ref{lem_makeLPbasic} to construct $\lpreroute$ from $\lpbasic$ such that $\mathbb{E}[Opt(\lpreroute)] \leq \frac{\tau - 1}{\log_e \tau} Opt(\I)$ and $\lpreroute$ satisfies all Basic Invariants. Note that we will choose $\tau > 1$ later to optimize our final approximation ratio. Then we run \mainalg~ on $\lpreroute$, which satisfies all Basic Invariants, so by Theorem \ref{thm_pseudo}, \mainalg~ outputs in polynomial time $\lpreroute$ along with an optimal solution $\bar{y}$ with $O(r)$ fractional variables.
	
	Given $\bar{y}$, we define a setting $\bar{x}$ for the
        $x$-variables: for all $j \in C_{part}$, connect $j$ to all
        facilities in $F_j$ by setting $\bar{x}_{ij} = \bar{y}_i$ for
        all $i \in F_j$. For all $j \in C^*$, we have $\bar{y}(F_j) =
        1$, so connect $j$ to all facilities in $F_j$. Finally, 
	to connect every $j \in C_{full}$ to one unit of open facilities, we use the following modification of Theorem \ref{thm_mainreroute}:
	
	\begin{prop}\label{prop_fracreroute}
		When \mainalg\ terminates, for all $j \in C_{full}
                \cup C^*$, there exists one unit of open facilities
                with respect to $\bar{y}$ within distance $(2 +
                \alpha) L(\ell_j)$ of $j$, where $\alpha = \max(1, 1 + \frac{2}{\tau}, \frac{\tau^3 + 2\tau^2 + 1}{\tau^3 - 1})$.
	\end{prop}

	The proof of the above proposition is analogous to that of Theorem \ref{thm_mainreroute} in the case $\beta = 1$, so we omit it. To see this, note that for all $j \in C^*$, we have $\bar{y}(F_j) = 1$. This implies that each $j \in C^*$ has one unit of fractional facility within distance $L(\ell_j)$. Following an analogous inductive argument as in Lemma \ref{lem_chase} gives the desired result.
        
        By routine calculations, it is easy to see that $\alpha = \frac{\tau^3 + 2\tau^2 + 1}{\tau^3 - 1}$ for all $\tau > 1$. Now, for all $j \in C_{full}$, we connect $j$ to all facilities in $B_j$. We want to connect $j$ to one unit of open facilities, so to find the remaining $1 - \bar{y}(B_j)$ units, we connect $j$ to an arbitrary $1 - \bar{y}(B_j)$ units of open facilities within distance $(2 + \alpha) L(\ell_j)$ of $j$, whose existence is guaranteed by Proposition \ref{prop_fracreroute}. This completes the description of $\bar{x}$.
	
	It is easy to verify that $(\bar{x}, \bar{y})$ is feasible for $\lpbasic$, because $\bar{y}$ satisfies all knapsack constraints, and every client's contribution to the coverage constraints in $\lpbasic$ is exactly its contribution in $\lpreroute$. Thus it remains to bound the cost of this solution. We claim that $\lpbasic(\bar{x}, \bar{y}) \leq (2 + \alpha) Opt(\lpreroute)$, because each client in $C_{part}$ and $C^*$ contributes the same amount to $\lpbasic$ and $\lpreroute$ (up to discretization), and each client  $j \in C_{full}$ has connection cost at most $2 + \alpha$ times its contribution to $\lpreroute$.
	
	In conclusion, the expect cost of the solution $(\bar{x}, \bar{y})$ to $\lpbasic$ is at most:
	\[(2 + \alpha)\, \mathbb{E}[Opt(\lpreroute)] \leq \frac{\tau
            -1}{\log_e \tau} \left(2 + \frac{\tau^3 + 2\tau^2 +
              1}{\tau^3 - 1}\right)\; Opt(\I).\]
	Choosing $\tau > 1$ to minimize $\frac{\tau -1}{\log_e \tau} (2 + \frac{\tau^3 + 2\tau^2 + 1}{\tau^3 - 1})$ gives $\tau = 2.046$ and $\frac{\tau -1}{\log_e \tau} (2 + \frac{\tau^3 + 2\tau^2 + 1}{\tau^3 - 1}) = 6.387$.
\end{proof}

\section{From Pseudo-Approximation to Approximation}\label{sec_trueapprox}

In this section, we leverage the pseudo-approximation algorithm for GKM defined in Section \ref{sec_pseudoalg} to construct improved approximation algorithms for two special cases of GKM: knapsack median and $k$-median with outliers.

Recall that knapsack median is an instance of GKM with a single arbitrary knapsack constraint and a single coverage constraint that states we must serve every client in $C$. Similarly, $k$-median with outiers is an instance of GKM with a single knapsack constraint, stating that we can open at most $k$ facilities, and a single coverage constraint, stating that we must serve at least $m$ clients. Note that both special cases have $r = 2$.

Our main results for these two problems are given by the following theorems:

\begin{thm}[Approximation Algorithm for Knapsack Median]\label{thm_mainknapsack}
    There exists a polynomial time randomized algorithm that takes as input an instance $\I$ of knapsack median and parameter $\epsilon \in (0, 1/2)$ and in time $n^{O(1/ \epsilon)}$, outputs a feasible solution to $\I$ of expected cost at most $(6.387 + \epsilon) Opt(\I)$.
\end{thm}

\begin{thm}[Approximation Algorithm for $k$-Median with Outliers]\label{thm_mainoutliers}
    There exists a polynomial time randomized algorithm that takes as input an instance $\I$ of $k$-median with outliers and parameter $\epsilon \in (0,1/2)$ and in time $n^{O(1/ \epsilon)}$, outputs a feasible solution to $\I$ of expected cost at most $(6.994 + \epsilon) Opt(\I)$. 
\end{thm}

\subsection{Overview}

The centerpiece for both of our approximation algorithms is the pseudo-approximation algorithm \mainalg~ for GKM. For both of these special cases, we can obtain via \mainalg~ a solution to $\lpreroute$ with only $O(1)$ fractional facilities and bounded re-routing cost. Now our remaining task is to turn this solution with $O(1)$ fractional facilities into an integral one.

Unfortunately, the basic LP relaxations for knapsack median and $k$-median with outliers have an unbounded integrality gap.  To overcome this bad integrality gap, we use known sparsification tools to pre-process the given instance. Our main technical contribution in this section is a post-processing algorithm that rounds the final $O(1)$ fractional variables at a small cost increase for the special cases of knapsack median and $k$-median with outliers.

Thus our approximation algorithms for knapsack median and $k$-median with outliers consist of a known pre-processing algorithm \cite{DBLP:conf/stoc/KrishnaswamyLS18}, our new pseudo-approximation algorithm, and our new post-processing algorithm.

\subsection{Approximation Algorithm for Knapsack Median}

To illustrate our approach, we give the pre- and post-processing algorithms for knapsack median, which is the simpler of the two variants. Our pre-processing instance modifies the data of the input instance, so the next definition is useful to specify the input instance and pre-processed instance.

\begin{definition}[Instance of Knapsack Median]\label{def_instanceknapsack}
    An instance of knapsack median is of the form $\I = (F,C,d,w,B)$, where $F$ is a set of facilities, $C$ is a set of clients, $d$ is a metric on $F \cup C$, $w \in \mathbb{R}_+^F$ is the weights of the facilities, and $B \geq 0$ is the budget.
\end{definition}

Note that for knapsack median, the two side constraints in $\lpreroute$ are the knapsack constraint, $\sum\limits_{i \in F} w_i y_i \leq B$, and the coverage constraint, $\sum\limits_{j \in C_{part}} y(F_j) \geq \lvert C \rvert - \lvert C_{full} \cup C^* \rvert$.

We utilize the same pre-processing as in \cite{DBLP:conf/stoc/KrishnaswamyLS18}. Roughly speaking, given a knapsack median instance $\I$, we first handle the expensive parts of the optimal solution using enumeration. Once we pre-open the facilities and decide what clients should be assigned there for this expensive part of the instance, we are left with a sub-instance, say $\I'$. In $\I'$, our goal is to open some more facilities to serve the remaining clients.

Roughly speaking, $\I'$ is the ``cheap" part of the input instance. Thus, when we construct $\lpreroute$ for this sub-instance, we initialize additional invariants which we call our \emph{Extra Invariants}.

To state our Extra Invariants, we need to define the \emph{$P$-ball centered at $p$ with radius $r$} for any $P \subset F \cup C$, $p \in F \cup C$, and $r \geq 0$, which is the set:
\[B_P(p,r) = \{q \in P \mid d(p,q) \leq r\}.\]

\begin{definition}[Extra Invariants for Knapsack Median] \label{invar_kextra}
	Let $\rho, \delta \in (0,1/2)$, $U \geq 0$, $S_0 \subset F$, and $R \in \mathbb{R}_+^C$ be given. Then we call the following properties our \emph{Extra Invariants}:	
	\begin{enumerate}
		\item For all $i \in S_0$, there exists a dummy client $j(i) \in C^*$ such that $F_{j(i)} = \{i' \in F \mid i' \text{ collocated with $i$}\}$ with radius level $\ell_{j(i)} = - 1$. We let $C_0 \subset C$ be the collection of these dummy clients. \label{invar_kguessfacil}
		\item For all $i \in F$ that is not collocated with some $i' \in S_0$, we have $\sum\limits_{j \mid i \in F_j} d(i,j) \leq 2\rho U$	\label{invar_ksparsefacil}
		\item For all $j \in C$, we have $L(\ell_j) \leq \tau R_j$	\label{invar_kupperbound}
		\item For all $j \in C$ and $r \leq R_j$, we have:
        $\lvert B_C(j, \delta r) \rvert r \leq \rho U.$
         \label{invar_ksparserad}
	\end{enumerate}
\end{definition}

Extra Invariant \ref{invar_kextra}(\ref{invar_kguessfacil}) guarantees that we open the set of guessed facilities $S_0$ in our final solution. Then for all non-guessed facilities, so the set $F \setminus S_0$, Extra Invariant \ref{invar_kextra}(\ref{invar_ksparsefacil}) captures the idea that these facilities are ``cheap." Taken together, Extra Invariants \ref{invar_kextra}(\ref{invar_kupperbound}) and \ref{invar_kextra}(\ref{invar_ksparserad}) capture the idea that all remaining clients are ``cheap."

The next theorem describes our pre-processing algorithm for knapsack median, which is a convenient re-packaging of the pre-processing used in \cite{DBLP:conf/stoc/KrishnaswamyLS18}. The theorem essentially states that given $\rho, \delta$, and $U$, we can efficiently guess a set $C \setminus C'$ of clients and $S_0$ of facilities that capture the expensive part of the input instance $\I$. Then when we construct $\lpreroute$ for the cheap sub-instance, we can obtain the Extra Invariants, and the cost of extending a solution of the sub-instance to the whole instance is bounded with respect to $U$, which one should imagine is $Opt(\I)$.

\begin{thm}[Pre-Processing for Knapsack Median]\label{thm_kpre}
	Let $\I = (F,C,d,w,B)$ be an instance of knapsack median. Then, given as input instance $\I$, parameters $\rho, \delta \in (0, 1/2)$, and an upper bound $U$ on $Opt(\I)$, there exists an algorithm that runs in time $n^{O(1/ \rho)}$ and outputs $n^{O(1/ \rho)}$-many sub-instances of the form $\I' = (F,C' \subset C,d,w, B)$ along with the data for $\lpreroute$ on $\I'$, a set of facilities $S_0 \subset F$, and a vector $R \in \mathbb{R}_+^{C'}$ such that:
	\begin{enumerate}
		\item $\lpreroute$ satisfies all Basic and Extra Invariants
		\item $\frac{\log_e \tau}{\tau - 1}\mathbb{E}[Opt(\lpreroute)] + \frac{1 - \delta}{1+ \delta} \sum\limits_{j \in C \setminus C'} d(j, S_0) \leq U$
	\end{enumerate}
\end{thm}

The proof is implicit in \cite{DBLP:conf/stoc/KrishnaswamyLS18}. For completeness, we prove the analogous theorem for $k$-median with outliers, \Cref{thm_outpre}, in \S \ref{appendix_preoutlier}.

We will show that if $\lpreroute$ satisfies the Extra Invariants for knapsack median, then we can give a post-processing algorithm with bounded cost. It is not difficult to see that \mainalg~ maintains the Extra Invariants as well, so we use the Extra Invariants in our post-processing.

\begin{prop}\label{prop_kmainextra}
	\mainalg~ maintains all Extra Invariants for knapsack median.
\end{prop}

Now we move on to describing our post-processing algorithm. Suppose we run the pre-processing algorithm guaranteed by Theorem \ref{thm_kpre} to obtain $\lpreroute$ satisfying all Basic- and Extra Invariants. Then we can run \mainalg~ to obtain an optimal extreme point of $\lpreroute$ with $O(1)$ fractional facilities, and $\lpreroute$ still satisfies all Basic- and Extra Invariants.

It turns out, to round these $O(1)$ fractional facilities, it suffices to open one facility in each  $F$-ball for clients in $C^*$. Then we can apply Theorem \ref{thm_mainreroute} to bound the re-routing cost. The main difficulty in this approach is that we must also round some fractional facilities down to zero to maintain the knapsack constraint.

Note that closing a facility can incur an unbounded multiplicative cost in the objective. To see this, consider a fractional facility $i$ that is almost open, so $\bar{y}_i \sim 1$. Then suppose there exists $j \in C_{full}$ such that $i \in B_j$ and $d(j,i) \ll L(\ell_j)$. Then $j$'s contribution to the objective of $\lpreroute$ is $\sim d(j,i)$. However, if we close $i$, then $j$'s contribution increases to $L(\ell_j) \gg d(j,i)$.

To bound the cost of closing facilities, we use the Extra Invariants. In particular, we use the next technical lemma, which states that if we want to close down a facility $i$, and every client $j$ that connects to $i$ has a back-up facility to go to within distance $O(1) L(\ell_j)$, then closing $i$ incurs only a small increase in cost. For proof, see \S \ref{sec_trueproofs}.

\begin{lem}\label{lem_kclosefacil}
	Suppose $\lpreroute$ satisfies all Basic and Extra Invariants for knapsack median, and let $S \subset F$ and $\alpha \geq 1$. Further, consider a facility $i \notin S \cup S_0$ and set of clients $C' \subset C$ such that for all $j \in C'$, we have $i \in F_j$ and there exists some facility in $S$ within distance $\alpha L(\ell_j)$ of $j$. Then $\sum\limits_{j \in C'} d(j,S) = O(\frac{\rho}{\delta}) U$.
\end{lem}

By the next proposition, rounding a facility up to one does not incur any cost increase, because every client must be fully connected.

\begin{prop}\label{prop_allfull}
    Upon termination of \mainalg~ on a knapsack median instance, we have $C_{part} = \emptyset$.
\end{prop}
\begin{proof}
    We observe that the single coverage constraint in $\lpreroute$ for a knapsack median instance is of the form:
    \[\sum\limits_{j \in C_{part}} y(F_j) \geq \lvert C \rvert - \lvert C_{full} \cup C^* \rvert = \lvert C_{part} \rvert\]
    , where we use the fact that $C_{part}$, $C_{full}$, and $C^*$ partition $C$ due to Basic Invariant \ref{invar_basic}(\ref{invar_partition}). Combining this with the constraint $y(F_j) \leq 1$ for all $j \in C_{part}$ gives that $y(F_j) = 1$ for all $j \in C_{part}$ for any feasible solution to $\lpreroute$. By assumption, no $C_{part}$-constraint is tight upon termination of \mainalg, so the proposition follows.
\end{proof}

To summarize, the goal of our post-processing algorithm is to find an integral setting of the $O(1)$ fractional facilities in the output of \mainalg~ such that the knapsack constraint is satisfied and there is an open facility in each $F$-ball for clients in $C^*$.

\begin{lem}\label{lem_kexists}
    Upon termination of \mainalg~ on a knapsack median instance, let $\bar{y}$ be the outputted extreme point of $\lpreroute$, and suppose $\lpreroute$ satisfies all Basic- and Extra Invariants. Then there exists an integral setting of the fractional facilities such that the knapsack constraint is satisfied, there is an open facility in each $F$-ball for clients in $C^*$, and every facility in $S_0$ is open.
\end{lem}
\begin{proof}
    Consider the following LP:
    \[\mathcal{LP} = \min\limits_y \{\sum\limits_{i \in F} w_i y_i \mid y(F_j) = 1\quad \forall j \in C^* ,\, y_i = 1 \quad \forall i \in F_{=1},\, y \in [0,1]^F\}\]
    The first constraint states that we want one open facility in each $F$-ball for clients in $C^*$, and the second states that our solution should agree on the integral facilities in $\bar{y}$.
    
    Because $\lpreroute$ satisfies all Basic Invariants, the intersection graph of $\{F_j \mid j \in C^*\}$ is bipartite by Proposition \ref{prop_bipartite}. Then the feasible region of $\mathcal{LP}$ is a face of the intersection of two partition matroids (each side of the biparitition of $\{F_j \mid j \in C^*\}$ defines one parititon matroid), and thus $\mathcal{LP}$ is integral.
    
    To conclude the proof, we observe that $\bar{y}$ is feasible for $\mathcal{LP}$, so $Opt(\mathcal{LP}) \leq \sum\limits_{i \in F} w_i \bar{y}_i \leq B$. Thus there exists an integral setting of facilities that opens one facility in each  $F$-bal for all clients in $C^*$, agrees with all of $\bar{y}$'s integral facilities, and has total weight at most $B$. Finally, by Extra Invariant \ref{invar_kextra}(\ref{invar_kguessfacil}), $C_0 \subset C^*$, so we open every facility in $S_0$.
\end{proof}

Thus, in light of Lemma \ref{lem_kexists}, our post-processing algorithm is to enumerate over all integral settings of the fractional variables to find one that satisfies the knapsack constraint, opens one facility in each $F$-ball for clients in $C^*$, and opens $S_0$. Combining our post-processing algorithm with \mainalg~ gives the following theorem.

\begin{thm}\label{thm_knapsackapprox}
    There exists a polynomial time algorithm that takes as input $\lpreroute$ for knapsack median instance $\I$ satisfying all Basic- and Extra Invariants and outputs a feasible solution to $\I$ such that the solution opens all facilities in $S_0$ and has cost at most $(2 + \alpha) Opt(\lpreroute) + O(\rho / \delta) U$, where $\alpha = \frac{\tau^3 + 2\tau^2 + 1}{\tau^3 - 1}$.
\end{thm}
\begin{proof}
    Our algorithm is to run \mainalg~ on $\lpreroute$ and then run our post-processing algorithm, which is to enumerate over all integral settings of the fractional variables, and then output the feasible solution that opens $S_0$ of lowest cost (if such a solution exists.)
    
    Let $\bar{y}$ be the optimal extreme point of $\lpreroute$ output by \mainalg, which has $O(1)$ fractional variables by \Cref{thm_pseudo}. Because $\bar{y}$ has $O(1)$ fractional variables, our post-processing algorithm is clearly efficient, which establishes the runtime of our overall algorithm.
    
    Note that upon termination, $\lpreroute$ still satisfies all Basic- and Extra Invariants. Then by Lemma \ref{lem_kexists}, there exists an integral setting of the fractional variables that is feasible, opens $S_0$, and opens a facility in each $F$-ball for clients in $C^*$. It suffices to bound the cost of this solution. Let $S \subset F$ denote the facilities opened by this integral solution, so $d(j, S) \leq L(\ell_j)$ for all $j \in C^*$. Applying Lemma \ref{thm_mainreroute} with $\beta = 1$, we obtain that $d(j,S) \leq (2 + \alpha) L(\ell_j)$ for all $j \in C_{full} \cup C^*$, where $\alpha = \max(1, 1 + \frac{2}{\tau}, \frac{\tau^3 + 2\tau^2 + 1}{\tau^3 - 1})$. It is easy to check that $\alpha = \frac{\tau^3 + 2\tau^2 + 1}{\tau^3 - 1}$ for all $\tau > 1$.
    
    To bound the cost of the solution $S$ relative to $Opt(\lpreroute)$, we must bound the cost of closing the $O(1)$-many facilities in $F_{<1} \setminus S$. We recall that by Proposition \ref{prop_allfull}, we have $C = C_{full} \cup C^*$, so all clients must be fully connected in $\lpreroute$.
    
    First we consider any client $j \in C$ that is not supported on any facility in $F_{<1} \setminus S$. Such a client is not affected by closing $F_{<1} \setminus S$, so if $F_j$ is empty, then $d(j, S) \leq (2 + \alpha) L(\ell_j)$, which is at most $(2 + \alpha)$ times $j$'s contribution to $\lpreroute$. Otherwise, $F_j$ contains an integral facility in $S$ to connect to, so $d(j,S)$ is at most $j$'s contribution to $\lpreroute$.
    
    It remains to consider the clients whose $F$-balls contain a facility in $F_{<1} \setminus S$. Because there are only $O(1)$-many facilities in $F_{<1} \setminus S$, it suffices to show that for each $i \in F_{<1} \setminus S$, the additive cost of connecting all clients supported on $i$ is at most $O(\rho / \delta) U$. Here we apply Lemma \ref{lem_kclosefacil} to the set of clients $C' = \{j \in C \mid i \in F_j\}$ to obtain $\sum\limits_{j \in C'} d(j,S) = O(\rho / \delta) U$.
    
    To summarize, the cost of connecting the clients not supported on $F_{<1} \setminus S$ is at most $(2 + \alpha) Opt(\lpreroute)$, and the cost of the remaining clients is $O(\rho / \delta)U$, as required.
\end{proof}

Now our complete approximation for knapsack median follows from combining the pre-processing with the above theorem and tuning parameters.

\begin{proof}[Proof of Theorem \ref{thm_mainknapsack}]
    Let $\epsilon' > 0$. We will later choose $\epsilon'$ with respect to the given $\epsilon$ to obtain the desired approximation ratio and runtime. First, we choose parameters $\rho, \delta \in (0,1/2)$ and $U \geq 0$ for our pre-processing algorithm guaranteed by Theorem \ref{thm_kpre}. We take $\rho = {\epsilon'}^2$ and $\delta = \epsilon'$. We require that $U$ is an upper bound on $Opt(\I)$. Using a standard binary search idea, we can guess $Opt(\I)$ up to a multiplicative $(1 + \epsilon')$-factor in time $n^{O(1 / \epsilon')}$, so we guess $U$ such that $Opt(\I) \leq U \leq (1 + \epsilon') Opt(\I)$.
    
    With these choices of parameters, we run the algorithm guaranteed by Theorem \ref{thm_kpre} to obtain $n^{O(1/ \epsilon')}$ many sub-instances such that one such sub-instance is of the form $\I' = (F, C' \subset C, d, w, B)$, where $\lpreroute$ for $\I'$ satisfies all Basic- and Extra Invariants, and we have:
    
    \begin{equation}\label{eq_kcost}
       \frac{\log_e \tau}{\tau - 1}\mathbb{E}[Opt(\lpreroute)] + \frac{1 - \epsilon'}{1+ \epsilon'} \sum\limits_{j \in C \setminus C'} d(j, S_0) \leq U 
    \end{equation}
    
    Then for each sub-instance output by the pre-processing, we run the algorithm guaranteed by Theorem \ref{thm_knapsackapprox} to obtain a solution to each sub-instance. Finally, out of these solutions, we output the one that is feasible for the whole instance with smallest cost. This completes the description of our approximation algorithm for knapsack median. The runtime is $n^{O(1/ \epsilon')}$, so it remains to bound the cost of the output solution and to choose the parameters $\epsilon'$ and $\tau$.
    
    To bound the cost, it suffices to consider the solution output on the instance $\I'$ where $\lpreroute$ satisfies all Basic- and Extra Invariants and Equation \ref{eq_kcost}. By running the algorithm guaranteed by Theorem \ref{thm_knapsackapprox} on this $\lpreroute$, we obtain a feasible solution $S \subset F$ to $\I'$ such that $S_0 \subset S$, and the cost of connecting $C'$ to $S$ is at most $(2 + \alpha)Opt(\lpreroute) + O(\epsilon') U$, where $\alpha =  \frac{\tau^3 + 2\tau^2 + 1}{\tau^3 - 1}$. To extend this solution on the sub-instance to a solution on the whole instance $\I$, we must connect $C \setminus C'$ to $S$. Because $S_0 \subset S$, applying Equation \ref{eq_kcost} allows us to upper bound the expected cost of connecting $C$ to $S$ by:
    \[(2 + \alpha)\mathbb{E}[Opt(\lpreroute)] + O(\epsilon') U + \sum\limits_{j \in C \setminus C'} d(j, S_0) \leq (2 + \alpha) \frac{\tau - 1}{\log_e \tau} \frac{1 + \epsilon'}{1 - \epsilon'}U + O(\epsilon')U.\]
    Now choosing $\tau > 1$ to minimize $(2 + \frac{\tau^3 + 2\tau^2 + 1}{\tau^3 - 1}) \frac{\tau - 1}{\log_e \tau}$ gives $\tau = 2.046$ and $\frac{\tau -1}{\log_e \tau} (2 + \frac{\tau^3 + 2\tau^2 + 1}{\tau^3 - 1}) = 6.387$. Thus the expected cost of this solution is at most $6.387\frac{1 + \epsilon'}{1 - \epsilon'}U + O(\epsilon')U$, where $U \leq (1 + \epsilon') Opt(\I)$. Finally, by routine calculations, we can choose $\epsilon' = \theta(\epsilon)$ so that expected cost is at most $(6.387 + \epsilon) Opt(\I)$, as required. Note that the runtime of our algorithm is $n^{O(1/ \epsilon')} = n^{O(1/ \epsilon)}$.
\end{proof}

\subsection{Approximation Algorithm for $k$-Median with Outliers}

Our approximation algorithm for $k$-median with outliers follows the same general steps as our algorithm for knapsack median. We state the analogous Extra Invariants for $k$-median with outliers and pre-processing algorithm here. The only differences between the Extra Invariants for knapsack median and $k$-median with outliers is in the final Extra Invariant.

\begin{definition}[Instance of $k$-Median with Outliers]\label{def_instanceoutlier}
    An instance of $k$-median with outliers is of the form $\I = (F,C,d,k,m)$, where $F$ is a set of facilities, $C$ is a set of clients, $d$ is a metric on $F \cup C$, $k$ is the number of facilities to open, and $m$ is the number of clients to serve.
\end{definition}

Note that for $k$-median with outliers, the two side constraints in $\lpreroute$ are the knapsack constraint, $y(F) \leq k$, and the coverage constraint, $\sum\limits_{j \in C_{part}} y(F_j) \geq m - \lvert C_{full} \cup C^* \rvert$.

\begin{definition}[Extra Invariants for $k$-Median with Outliers] \label{invar_outextra}
	Let $\rho, \delta \in (0,1/2)$, $U \geq 0$, $S_0 \subset F$, and $R \in \mathbb{R}_+^C$ be given. Then we call the following properties our \emph{Extra Invariants}:	
	\begin{enumerate}
		\item For all $i \in S_0$, there exists a dummy client $j(i) \in C^*$ such that $F_{j(i)} = \{i' \in F \mid i' \text{ colocated with $i$}\}$ with radius level $\ell_{j(i)} = - 1$. We let $C_0 \subset C$ be the collection of these dummy clients. \label{invar_outguessfacil}
		\item For all $i \in F$ that is not collocated with some $i' \in S_0$, we have $\sum\limits_{j \mid i \in F_j} d(i,j) \leq \rho (1 + \delta) U$	\label{invar_outsparsefacil}
		\item For all $j \in C$, we have $L(\ell_j) \leq \tau R_j$	\label{invar_outupperbound}
		\item For every $t > 0$ and $p \in F \cup C$, we have:
		\[\lvert \{j \in B_C(p, \frac{\delta t}{4 + 3 \delta}) \mid R_j \geq t \} \rvert \leq \frac{\rho (1 + 3\delta / 4)}{1 - \delta/4} \frac{U}{t}.\] \label{invar_outsparserad}
	\end{enumerate}
\end{definition}

Again, the pre-processing of \cite{DBLP:conf/stoc/KrishnaswamyLS18} gives the next theorem. For proof, see \S \ref{appendix_preoutlier}.

\begin{thm}[Pre-Processing for $k$-Median with Outliers]\label{thm_outpre}
	Let $\I = (F,C,d,k,m)$ be an instance of $k$-median with outliers with optimal solution $(S^*,C^*)$. Then, given as input instance $\I$, parameters $\rho, \delta \in (0, 1/2)$, and an upper bound $U$ on $Opt(\I)$, there exists an algorithm that runs in time $n^{O(1/ \rho)}$ and outputs $n^{O(1/ \rho)}$-many sub-instances of the form $\I' = (F,C' \subset C,d,k,m' = m - \lvert C^* \setminus C' \rvert)$ along with the data for $\lpreroute$ on $\I'$, a set of facilities $S_0 \subset F$, and a vector $R \in \mathbb{R}_+^{C'}$ such that:
	\begin{enumerate}
		\item $\lpreroute'$ satisfies all Basic and Extra Invariants
		\item $\frac{\log_e \tau}{(\tau - 1)(1 + \delta/2)}\mathbb{E}[Opt(\lpreroute)] + \frac{1 - \delta}{1+ \delta} \sum\limits_{j \in C^* \setminus C'} d(j, S_0) \leq U$
	\end{enumerate}
\end{thm}

It is easy to check that \mainalg~ maintains all Extra Invariants for $k$-median with outliers as well, and we have an analogous technical lemma to bound the cost of closing facilities. For proof of the lemma, see \S \ref{sec_trueproofs}.

\begin{prop}\label{prop_outmainextra}
	\mainalg~ maintains all Extra Invariants for $k$-median with outliers.
\end{prop}

\begin{lem}\label{lem_outclosefacil}
	Suppose $\lpreroute$ satisfies all Basic and Extra Invariants for $k$-median with outliers, and let $S \subset F$ and $\alpha \geq 1$. Further, consider a facility $i \notin S \cup S_0$ and set of clients $C' \subset C$ such that for all $j \in C'$, we have $i \in F_j$ and there exists some facility in $S$ within distance $\alpha L(\ell_j)$ of $j$. Then $\sum\limits_{j \in C'} d(j,S) = O(\frac{\rho}{\delta}) U$.
\end{lem}

Now we focus on the main difference between the two algorithms: the post-processing. In particular, the coverage constraint of $k$-median with outliers introduces two difficulties in rounding the final $O(1)$ fractional facilities: (a)~we are no longer guaranteed that $C_{part} = \emptyset$, and (b)~we must satisfy the coverage constraint.

The difficulty with (a)~is that now rounding a facility up to one can also incur an unbounded multiplicative cost in the objective. To see this, consider a fractional facility $i$ that is almost closed, so $\bar{y}_i \sim 0$. Consider rounding this facility up to one. Then for a client $j \in C_{part}$ that fractionally connects to $i$ in the solution $\bar{y}$, if we fully connect $j$ to $i$, this costs $d(j,i) \gg d(j,i) \bar{y}_i$. The solution here is to use Extra Invariant \ref{invar_outextra}(\ref{invar_outsparsefacil}) to bound the additive cost of opening facilities.

The more troublesome issue is (b). Note that the same approach that we used to prove that there exists a good integral setting of the $O(1)$ fractional variables in \Cref{lem_kexists} does not work here because putting the coverage constraint in the objective of the LP could result in a solution covering the same client multiple times. Our solution to (b) is a more sophisticated post-processing algorithm that first re-routes clients in $C_{part}$. After re-routing, we carefully pick facilities to open that do not double-cover any remaining $C_{part}$-clients. We defer the details of our post-processing algorithm for \S \ref{sec_postoutlier}. For now, we present the guarantees of our pseudo-approximation combined with post-processing:

\begin{thm}\label{thm_outlierapprox}
    There exists a polynomial time algorithm that takes as input $\lpreroute$ for $k$-median with outliers instance $\I$ satisfying all Basic- and Extra Invariants and outputs a feasible set of facilities $S \supset S_0$ such that the cost of connecting $m$ clients to $S$ is at most $(2 + \alpha) Opt(\lpreroute) + O(\rho / \delta) U$, where $\alpha = \max(3 + 2\tau^{-c}, 1 + \frac{4 + 2\tau^{-c}}{\tau}, \frac{\tau^3 + 2\tau^2 + 1}{\tau^3 - 1})$ for any constant $c \in \mathbb{N}$.
\end{thm}

Combining the pre-processing of \Cref{thm_outpre} with \Cref{thm_outlierapprox} and tuning parameters gives our final approximation algorithm for $k$-median with outliers. The proof of \Cref{thm_mainoutliers} is analogous to \Cref{thm_mainknapsack}, so we defer it to \S \ref{sec_trueproofs}.

\section{Post-Processing for $k$-Median with Outliers}\label{sec_postoutlier}

In this section we develop the post-processing algorithm for $k$-median with outliers that is guaranteed by \Cref{thm_outlierapprox}. The structure of our algorithm is recursive. First, we give a procedure to round at least one fractional facility or serve at least one client. Then we recurse on the remaining instance until we obtain an integral solution.

\subsection{Computing Partial Solutions}

In this section, we show how to round at least one fractional facility or serve at least one client. We interpret this algorithm as computing a \emph{partial solution} to the given $k$-median with outliers instance.

The main idea of this algorithm is to re-route clients in $C_{part}$. In particular, we maintain a subset $\bar{C} \subset C^*$ such that for every client in $\bar{C}$, we guarantee to open an integral facility in their $F$-ball. We also maintain a subset $C_{covered} \subset C_{part}$ of $C_{part}$-clients that we re-route; that is, we guarantee to serve them even if no open facility is in their $F$-balls. Crucially, every client in $C_{part} \setminus C_{covered}$ is supported on at most one $F$-ball for clients in $\bar{C}$. Thus, we do not have to worry about double-covering those clients when we round the facilities in $F(\bar{C})$.

The partial solution we output consists of one open facility for each client in $\bar{C}$ (along with the facilities that happen to be integral already), and we serve the clients in $C_{full}$, $C^*$, $C_{covered}$, and the $C_{part}$-clients supported on our open facilities. See Algorithm \ref{alg_outlierpost} (\outlierpost) for the formal algorithm to compute partial solutions. Note that $c \in \mathbb{N}$ is a parameter of \outlierpost.

\begin{algorithm}
	\DontPrintSemicolon
	\LinesNumbered
	\SetAlgoLined
	
	\KwIn{$\lpreroute$ and $\bar{y}$ output by \mainalg~ on a $k$-median with outliers instance such that $C^*_{<1} \not\subset C_0$}
	\KwOut{Output a partial solution $S \subset F$ and modify $\lpreroute$}
	\BlankLine
	Initialize $C_{covered} = \emptyset$ and $\bar{C} = \{j \in C^*_{<1} \mid F_j \cap S_0 = \emptyset\}$\;
	\For{all clients $\bar{j} \in \bar{C}$ in increasing order of $\ell_{\bar{j}}$} {
		For all $j \in \bar{C}$ such that $j \neq \bar{j}$ and $F_j \cap F_{\bar{j}} \neq \emptyset$, remove $j$ from $\bar{C}$\;
		\While {there exists a client $j' \in C_{part} \setminus C_{covered}$ such that $F_{j'}$ intersects $F_{\bar{j}}$ and $F_j$ for some other $j \in \bar{C}$} {
			\If{$\ell_{j'} \leq \ell_{\bar{j}} - c$} {
				Remove $j$ from $\bar{C}$\;
			}	
			\Else{
				Add $j'$ to $C_{covered}$\;
			}	
		}
	}
	For all $i \in F$, we define $w_i = \lvert \{j \in C_{part} \setminus C_{covered} \mid i \in F_j\} \rvert$\;
	Construct the set $\bar{S} \subset F(\bar{C})$ by greedily picking the facility $i \in F_j$ with largest $w_i$ for each $j \in \bar{C}$\;
	Define the set $C_{part}(\bar{S}) = \{j \in C_{part} \mid F_j \cap \bar{S} \neq \emptyset\}$\;
	Define the partial set of facilities by $S = (\bar{S} \cup F_{= 1}) \setminus S_0$ and the partial set of clients by $C' = C_{part}(\bar{S}) \cup C_{covered} \cup C_{full} \cup (C^* \setminus C_0)$\;
	Update $\lpreroute$ by deleting $S$ and $F(\bar{C})$ from $F$, deleting all clients in $C'$ from $C$, decrementing $k$ by $\lvert S \rvert$, and decrementing $m$ by $\lvert C' \rvert$\;
	Output the partial solution $S$\;
	\caption{\outlierpost\label{alg_outlierpost}}
\end{algorithm}

We note that to define a solution for $k$-median with outliers, it suffices to specify the set of open facilities $S$, because we can choose the clients to serve as the $m$ closest clients to $S$. Thus when we output a partial solution, we only output the set of open facilities.

We summarize the performance of \outlierpost~ with the next theorem, which we prove in \S \ref{proof_subpost}. In the next section, we use \Cref{thm_subpost} to define our recursive post-processing algorithm.

\begin{thm}\label{thm_subpost}
  	Let $\lpreroute$ and $\bar{y}$ be  the input to \outlierpost. Then let $S$ be the partial solution output by \outlierpost~ and $\lpreroute^1$ be the modified LP. Then $\lpreroute^1$ satisfies all Basic- and Extra Invariants and we have:
  	\[Opt(\lpreroute^1) + \frac{1}{2 + \alpha} \sum\limits_{j \in C'} d(j, S \cup S_0) \leq Opt(\lpreroute) + O(\frac{\rho}{\delta}) U,\]
  	where $\alpha = \max(3 + 2\tau^{-c}, 1 + \frac{4 + 2\tau^{-c}}{\tau}, \frac{\tau^3 + 2\tau^2 + 1}{\tau^3 - 1})$.
\end{thm}

We want $\lpreroute^1$ to satisfy all Basic- and Extra Invariants so we can  continue recursing on $\lpreroute^1$. The second property of \Cref{thm_subpost} allows us to extend a solution computed on $\lpreroute^1$ with the partial solution $S$.

\subsection{Recursive Post-Processing Algorithm}

To complete our post-processing algorithm, we recursively apply \outlierpost~ until we have an integral solution.

The main idea is that we run \mainalg~ to obtain an optimal extreme point with $O(1)$ fractional variables. Then using this setting of the $y$-variables, we construct a partial solution consisting of some open facilities along with the clients that they serve. However, if there are still fractional facilities remaining, we recurse on $\lpreroute$ (after the modifications by \outlierpost.) Our final solution consists of the union of all recursively computed partial solutions. See Algorithm \ref{alg_outlier} (\outlieralg.)

\begin{algorithm}
	\DontPrintSemicolon
	\LinesNumbered
	\SetAlgoLined
	
	\KwIn{$\lpreroute$ for a $k$-median with outliers instance satisfying all Basic- and Extra Invariants}
	\KwOut{Solution $S \subset F$}
	\BlankLine
	Run \mainalg~ to obtain extreme point $\bar{y}$ to $\lpreroute$\;
	\If{$\bar{y}$ is integral}{
		Output the solution $F_{=1}$\;
	}
	\ElseIf{$C^*_{<1} \subset C_0$}{
		By \Cref{lem_extremedisjoint}, $\bar{y}$ has at most two fractional variables, say $a,b \in F$.\;
		Without loss of generality, we may assume $\lvert \{j \in C_{part} \mid a \in F_j \} \rvert \geq \lvert \{j \in C_{part} \mid b \in F_j \} \rvert$\;
		Output the solution $F_{=1} \cup \{a\}$\;	
	}
	
	\Else {
		Run \outlierpost~ to obtain partial solution $S'$ and update $\lpreroute$\;
		Run \outlierpost~ on updated $\lpreroute$ to obtain partial solution $S''$\;
		Output solution $S' \cup S''$\;
	}			
	\caption{\outlieralg\label{alg_outlier}}
\end{algorithm}

For proof of the next lemma, see \S \ref{appendix_postproofs}.

\begin{lem}\label{lem_extremedisjoint}
	If $C^*_{<1} \subset C_0$, then $\bar{y}$ has at most two fractional variables.
\end{lem}

\subsection{Analysis of \outlieralg}

In this section, we show that \outlieralg~ satisfies the guarantees of \Cref{thm_outlierapprox}. All missing proofs can be found in \S \ref{appendix_postproofs}. We let $\bar{y}$ be the output of \mainalg~ in the first line of \outlieralg~ and $\alpha = \max(3 + 2\tau^{-c}, 1 + \frac{4 + 2\tau^{-c}}{\tau}, \frac{\tau^3 + 2\tau^2 + 1}{\tau^3 - 1})$. First, we handle the base cases of \outlieralg.

The easiest base is when $\bar{y}$ is integral. Here, we do not need the Extra Invariants - all we need is that every client $j \in C_{full} \cup C^*$ has an open facility within distance $(2 + \alpha)L(\ell_j)$.

\begin{lem}\label{lem_base1}
	If $\bar{y}$ is integral, then the output solution $F_{=1}$ is feasible, contains $S_0$, and connecting $m$ clients costs at most $(2 + \alpha) Opt(\lpreroute)$.
\end{lem}

Now, in the other base case, $\bar{y}$ is not integral, but we know that $C^*_{<1} \subset C_0$. By \Cref{lem_extremedisjoint}, we may assume without loss of generality $\bar{y}$ has exactly two fractional facilities, say $a, b \in F_{<1}$. Further, we may assume that the $k$-constraint is tight, because opening more facilities can only improve the objective value. It follows that $\bar{y}_a + \bar{y}_b = 1$. For the sake of analysis, we define the sets $C_{part}(a) = \{j \in C_{part} \mid a \in F_j\}$ and $C_{part}(b) = \{j \in C_{part} \mid b \in F_j\}$ where we may assume $\lvert C_{part}(a) \rvert \geq \lvert C_{part}(b) \rvert$.

It remains to bound the cost of the solution $F_{=1} \cup \{a\}$. One should imagine that we obtain this solution from $\bar{y}$ by closing down the facility $b$ and opening up $a$. First we handle the degenerate case where $a \in S_0$. In this case, $a$ and $b$ must both be co-located copies of a facility in $S_0$, so opening one or the other does not change the cost. Thus, we may assume without loss of generality that $a \notin S_0$. Here, we need to Extra Invariants to bound the cost of opening $a$ and closing $b$



\begin{lem}\label{lem_base2cost}
	Suppose $a \notin S_0$. Then the output solution $F_{=1} \cup \{a\}$ is feasible, contains $S_0$, and connecting $m$ clients costs at most $(2 + \alpha) Opt(\lpreroute) + O(\frac{\rho}{\delta}) U.$
\end{lem}

This completes the analysis of the base cases. To handle the recursive step, we apply \Cref{thm_subpost}.

\begin{proof}[Proof of Theorem \ref{thm_outlierapprox}]
    First, we show that \outlieralg~ terminates in polynomial time. It suffices to show that the number of recursive calls  to \outlierpost~ is polynomial. To see this, note that for each recursive call, it must be the case that $C^*_{<1} \not\subset C_0$. In particular, there exists some non-dummy client in $C^* \setminus C_0$. Thus, we are guaranteed to remove at least once client from $C$ in each recursive call. 

	Now it remains to show that the output solution is feasible, contains $S_0$, and connecting $m$ clients costs at most $(2 + \alpha) Opt(\lpreroute) + O(\frac{\rho}{\delta}) U$. Let $\bar{y}$ be the extreme point computed by \mainalg~ in the first line of \outlieralg. If $\bar{y}$ is integral or $C^*_{<1} \subset C_0$, then we are done by the above lemmas.
	
	Then it remains to consider the case where $C^*_{<1} \not\subset C_0$. Let $\lpreroute$ denote the input to \outlieralg~ and $\lpreroute^1$ the updated $\lpreroute$ at the end of \outlierpost~ as in the statement of Theorem \ref{thm_subpost}. We note that Theorem \ref{thm_subpost} implies that $\lpreroute^1$ satisfies all Basic and Extra Invariants, so $\lpreroute^1$ is a valid input to \outlieralg. Then we may assume inductively that the recursive call to \outlieralg~ on $\lpreroute^1$ outputs a feasible solution $S''$ to $\lpreroute^1$ such that $S_0 \subset S''$ and the cost of connecting $m^1$ clients from $C^1$ to $S''$ is at most $(2 + \alpha) Opt(\lpreroute^1) + O(\frac{\rho}{\delta}) U$.
	
	Further, let $S'$ be the partial solution output by \outlierpost~ on $\lpreroute$. Now we combine the solutions $S'$ and $S''$ to obtain the solution output by \outlieralg. First, we check that $S' \cup S''$ is is feasible. This follows, because $\lvert S'' \rvert \leq k^1 \leq k - \lvert S' \rvert$ by definition of \outlierpost. Also, $S_0 \subset S''$ by the inductive hypothesis.
	
	It remains to bound the cost of connecting $m$ clients to $S' \cup S''$. Consider serving the $m^1$ closest clients in $C^1$ with $S''$ and $C'$ with $S' \cup S_0$. Because $m_1 = m - \lvert C' \rvert$, this is enough clients. Connecting the $m^1$ closest clients in $C^1$ to $S''$ costs at most $(2 + \alpha) Opt(\lpreroute^1) + O(\frac{\rho}{\delta}) U$ by the inductive hypothesis. Now we use the guarantee of \Cref{thm_subpost}, which we recall is:
    \[Opt(\lpreroute^1) + \frac{1}{2 + \alpha} \sum\limits_{j \in C'} d(j, S' \cup S_0) \leq Opt(\lpreroute) + O(\frac{\rho}{\delta}) U.\]
    Thus, the total connection cost is at most:
    \[(2 + \alpha) Opt(\lpreroute^1) + O(\frac{\rho}{\delta})U + \sum\limits_{j \in C'} d(j, S' \cup S_0) \leq (2 + \alpha) Opt(\lpreroute) + O(\frac{\rho}{\delta})U.\]
    Note that the additive $O(\frac{\rho}{\delta})U$ terms which we accrue in each recursive call are still $O(\frac{\rho}{\delta})U$ overall. This is because we keep recursing on a subset of the remaining \emph{fractional} facilities -- which is always $O(1)$ -- and we open/close each fractional facility at most once over all recursive calls. Thus, we can bound the additive cost of each opening/closing by $O(\frac{\rho}{\delta})U$.
\end{proof}

\subsection{Proof of \Cref{thm_subpost}}\label{proof_subpost}

For all missing proofs in this section, see \S \ref{appendix_postproofs}. We let $\lpreroute$ and $\bar{y}$ denote the input to \outlierpost~ and $\lpreroute^1$ the updated LP that is output at the end \outlierpost. We begin with three properties of \outlierpost~ that will be useful throughout our analysis.

The first is immediate by definition of \outlierpost.

\begin{prop}\label{prop_partnice}
	Upon termination of \outlierpost, the set family $\{F_j \mid j \in \bar{C}\}$ is disjoint, and every client $j \in C_{part} \setminus C_{covered}$, $F_j$ intersects at most one $F$-ball for clients in $\bar{C}$.
\end{prop}

\begin{prop}\label{prop_subinvar}
	\outlierpost~ initializes and maintains the invariants that $C_{covered} \subset C_{part}$ and $\bar{C} \subset \{j \in C^*_{<1} \mid F_j \cap S_0 = \emptyset\}$
\end{prop}
\begin{proof}
	We initialize $C_{covered} = \emptyset$ and only add clients from $C_{part}$ to $C_{covered}$. Similarly, we initialize $\bar{C} = \{j \in C^*_{<1} \mid F_j \cap S_0 = \emptyset\}$ and only remove clients from $\bar{C}$.
\end{proof}

\begin{lem}\label{lem_greedystay}
	Every $\bar{j} \in \bar{C}$ that is reached by the \textsc{For} loop remains in $\bar{C}$ until termination.
\end{lem}

Now we are ready to prove both properties of \Cref{thm_subpost}.  It is not difficult to see that $\lpreroute^1$ satisfies all Basic- and Extra Invariants by construction.
 
\begin{lem}\label{lem_afterinvar}
	$\lpreroute^1$ satisfies all Basic- and Extra Invariants.
\end{lem}

Now it remains to show 	$Opt(\lpreroute^1) + \frac{1}{2 + \alpha} \sum\limits_{j \in C'} d(j, S \cup S_0) \leq Opt(\lpreroute) + O(\frac{\rho}{\delta}) U$. To do this, we partition $C$ into $C^1$ and $C' = C_{part}(\bar{S}) \cup C_{covered} \cup C_{full} \cup (C^* \setminus C_0)$. For each client in $C^1$, we show that its contribution to the objective of $\lpreroute^1$ is at most its contribution to $\lpreroute$. Then for each client $j \in C'$, either $d(j,S \cup S_0)$ is at most $2 + \alpha$ times $j$'s contribution to $Opt(\lpreroute)$ or we can charge $j$'s connection cost to an additive $O(\frac{\rho}{\delta}) U$ term.

First, we focus on $C^1$. For these clients, it suffices to show that $\bar{y}$ (restricted to $F^1$) is feasible for $\lpreroute^1$. This is because for all $j \in C^1$, either $j \in C_{part}^1 \subset C_{part}$ or $j \in C_0$. The clients in $C_0$ contribute zero to the cost of $\lpreroute^1$ and $\lpreroute$. This is because both $\lpreroute$ and $\lpreroute^1$ satisfy Extra Invariant \ref{invar_outextra}(\ref{invar_outguessfacil}), so every dummy client $j(i) \in C_0$ is co-located with one unit of open facility corresponding to $i \in S_0$.

Thus it remains to consider the clients $j \in C_{part}^1$. We recall that $C_{part}^1 \subset C_{part}$ and $F_j^1 \subset F_j$ for all $j \in C^1$, so each $j \in C_{part}^1$ costs less in $\lpreroute^1$ than in $\lpreroute$.

To complete the cost analysis of $C^1$, we go back to prove feasibility. The main difficulty is showing that the coverage constraint is still satisfied. Recall that we construct $\bar{S}$ by greedily opening the facility in each $F$-ball for clients in $\bar{C}$ that covers the most $C_{part} \setminus C_{covered}$-clients. \Cref{prop_partnice} ensures that this greedy choice is well-defined (because $\{F_j \mid j \in \bar{C}\}$ is disjoint), and that we do not double-cover any $C_{part} \setminus C_{covered}$-clients.

Then by definition of greedy, we show that our partial solution covers more clients than the fractional facilities we delete. This proposition is the key to showing that the coverage constraint is still satisfied.

\begin{prop}\label{prop_greedyoutlier}
    Upon termination of \outlierpost, we have $\lvert C_{part}(\bar{S}) \rvert \geq \sum\limits_{j \in \bar{C}} \sum\limits_{i \in F_j} w_i \bar{y}_i$.
\end{prop}
\begin{proof}
    For each $j \in \bar{C}$, let $i(j) \in \bar{S}$ be the unique open facility in $F_j$. By definition of $C_{part}(\bar{S})$, for all $j \in C_{part}(\bar{S})$ we have $F_j \cap \bar{S} \neq \emptyset$. Further, by Proposition \ref{prop_partnice}, $F_j$ intersects exactly one $F$-ball among clients in $\bar{C}$, so each $i(j)$ for $j \in \bar{C}$ covers a unique set of clients. This implies the equality:
	\[\lvert C_{part}(\bar{S}) \rvert = \sum\limits_{j \in \bar{C}} w_{i(j)}.\]
	Combining this equality with the facts that for all $j \in \bar{C}$, we have $\bar{y}(F_j) = 1$ and $w_{i(j)} \geq w_i$ for all $i \in F_j$ gives the desired inequality:
	\[\lvert C_{part}(\bar{S}) \rvert = \sum\limits_{j \in \bar{C}} w_{i(j)} \geq \sum\limits_{j \in \bar{C}} w_{i(j)} \bar{y}(F_j) \geq \sum\limits_{j \in \bar{C}} \sum\limits_{i \in F_j} w_i \bar{y}_i.\]
\end{proof}

Finally, we can complete the analysis of $C^1$. It is easy to check that constraints except for the coverage constraint are satisfied. To handle the coverage constraint, we use \Cref{prop_greedyoutlier}.

\begin{lem}\label{lem_afterfeasible}
	$\bar{y}$ restricted to $F^1$ is feasible for $\lpreroute^1$.
\end{lem}

For the final property, we must upper bound the connection cost of $C' = C_{part}(\bar{S}) \cup C_{covered} \cup C_{full} \cup (C^* \setminus C_0)$ to $S \cup S_0$. We bound the connection cost in a few steps. First, we bound the re-routing cost of $C_{full} \cup C^*$. Second, we show that every $j \in C_{covered}$ has a ''back-up" facility within $O(1) L(\ell_j)$. This is used to bound the additive cost of connecting $C_{covered}$. Finally, for the clients in $C_{part}(\bar{S})$, we guarantee to open a facility in their $F$-balls, so we also bound their additive cost.

The next lemma and corollary allow us to bound the cost of $C_{full} \cup C^*$.

\begin{lem}\label{lem_fullextrastep}
	Upon termination of \outlierpost, for all $j \in C^*$, we have $d(j, S \cup S_0) \leq (3 + \frac{2}{\tau^c}) L(\ell_j)$.
\end{lem}
\begin{proof}
	There are a few cases to consider. First, if $j \in C^*_{=1}$, then the lemma is trivial because by definition there exists an integral facility in $F_j$. Otherwise, if $j \in C^*_{<1}$, but $F_j \cap S_0 \neq \emptyset$, then again the lemma is trivial.
	
	Thus it remains to consider clients $j \in \{j \in C^*_{<1} \mid F_j \cap S_0 = \emptyset\}$. We note that such a client $j$ is initially in $\bar{C}$. If $j$ remains in $\bar{C}$ until termination, then we are done, because we are guaranteed to open a facility in $F_j$ for all $j \in \bar{C}$ (this is exactly the set $\bar{S}$ of facilities.)
	
	For the final case, we suppose client $j$ is removed from $\bar{C}$ in the iteration where we consider $\bar{j} \in \bar{C}$. Then either $F_j \cap F_{\bar{j}} \neq \emptyset$ or there exists $j' \in C_{part} \setminus C_{covered}$ such that $F_{j'}$ intersects both $F_j$ and $F_{\bar{j}}$ and $\ell_{j'} \leq \ell_{\bar{j}} - c$. Note that because the \textsc{For} loop considers clients in increasing order of radius level, we have $\ell_j \geq \ell_{\bar{j}}$
	
	In the former case, by the Distinct Neighbors Property, we have $\ell_j \geq \ell_{\bar{j}}  + 1$. Further, by \Cref{lem_greedystay}, we know that $\bar{j}$ remains in $\bar{C}$ until termination, so $d(\bar{j}, S) \leq L(\ell_{\bar{j}})$. Then we can upper bound:
	\[d(j,S) \leq d(j, \bar{j}) + d(\bar{j}, S) \leq L(\ell_j) + L(\ell_{\bar{j}}) + L(\ell_{\bar{j}}) \leq (1 + \frac{2}{\tau}) L(\ell_j) < 3L(\ell_j)\]
	, where in the final inequality, we use the fact that $\tau > 1$.
	
	In the latter case, we again have $d(\bar{j}, S) \leq L(\ell_{\bar{j}})$.  Then we can bound the distance from $j$ to $S$ by first going from $j$ to $j'$, then from $j'$ to $\bar{j}$, and finally from $\bar{j}$ to $S$, where $\ell_{j'} \leq \ell_{\bar{j}} - c$ and $\ell_{\bar{j}} \leq \ell_j$:
	\[d(j, S) \leq d(j, j') + d(j', \bar{j}) + d(\bar{j}, S) \leq L(\ell_j) + L(\ell_{j'}) + L(\ell_{j'}) + L(\ell_{\bar{j}}) + L(\ell_{\bar{j}}) \leq (3 + \frac{2}{\tau^c})L(\ell_j).\]
\end{proof}
\begin{cor}\label{cor_outreroute}
    Upon  termination of \outlierpost, for all $j \in C_{full} \cup C^*$, we have $d(j, S \cup S_0) \leq (2 + \alpha) L(\ell_j)$, where $\alpha = \max(3 + 2\tau^{-c}, 1 + \frac{4 + 2\tau^{-c}}{\tau}, \frac{\tau^3 + 2\tau^2 + 1}{\tau^3 - 1})$
\end{cor}
\begin{proof}
    Apply \cref{thm_mainreroute} with $\beta = 3 + 2\tau^{-c}$.
\end{proof}

Similarly, the next lemma ensures that $C_{covered}$ can also be re-routed.

\begin{lem}\label{lem_covered}
	Upon termination of \outlierpost, for all $j \in C_{covered}$, we have $d(j,\bar{S}) \leq (1 + 2\tau^c)L(\ell_j)$.
\end{lem}
\begin{proof}
	Because $j \in C_{covered}$, it must be the case that we put $j$ in $C_{covered}$ in some iteration of the \textsc{For} loop where we consider client $\bar{j} \in \bar{C}$. Thus we have $\ell_j \geq \ell_{\bar{j}} - c  + 1$ and $F_j \cap F_{\bar{j}} \neq \emptyset$. Also, by Lemma \ref{lem_greedystay}, $\bar{j}$ remains in $\bar{C}$ until termination, so $d(\bar{j}, \bar{S}) \leq L(\ell_{\bar{j}})$. Using these facts, we can bound:
	\[d(j,S) \leq d(j, \bar{j}) + d(\bar{j}, S) \leq L(\ell_j) + L(\ell_{\bar{j}}) + L(\ell_{\bar{j}}) \leq (1 + 2 \tau^{c-1})L(\ell_j).\]
\end{proof}

Using the above lemmas, we are ready to bound the connection cost of our partial solution. Note that we bound the cost of serving $C'$ not with $S$, which is the partial solution we output, but rather with $S \cup S_0$. Thus, we are implicitly assuming that $S_0$ will be opened by the later recursive calls.

We recall that $C' = C_{part}(\bar{S}) \cup C_{covered} \cup C_{full} \cup (C^* \setminus C_0)$ and $S \cup S_0 = F_{=1} \cup \bar{S} \cup S_0$.
	
To begin, we bound the cost of connecting $C_{part}(\bar{S})$ to $\bar{S}$. By definition, every client $j \in C_{part}(\bar{S})$ has some facility from $\bar{S}$ in its $F$-ball. Further, $\bar{S} \subset F(\bar{C})$ by definition, and $F(\bar{C}) \cap S_0 = \emptyset$ using Proposition \ref{prop_partnice}. Thus we  can apply Extra Invariant \ref{invar_outextra}(\ref{invar_outsparsefacil}) to each facility in $\bar{S}$. Further, we know that $\lvert \bar{S} \rvert = O(1)$, because there are only $O(1)$ fractional facilities, and every facility in $\bar{S}$ is fractional by definition. Then we can bound:
\[\sum\limits_{j \in C_{part}(\bar{S})} d(j, \bar{S}) \leq \sum\limits_{i \in \bar{S}} \sum\limits_{j \in C_{part}(\bar{S}) \mid i \in F_j} d(j,i) \leq O(\rho) U\]
, where we apply Extra Invariant \ref{invar_outextra}(\ref{invar_outsparsefacil}) for each $i \in \bar{S}$. Thus, we have shown that the connection cost of $C_{part}(\bar{S})$ is at most an additive $O(\rho) U$.
	
Now we move on to the rest of $C'$, that is - the clients in $C_{covered}$, $C_{full}$ and $C^* \setminus C_0$. For the clients in $C_{covered}$, we know by Lemma \ref{lem_covered} that every client $j \in C_{covered}$ has an open facility in $\bar{S}$ at distance at most $(1 + 2\tau^c) L(\ell_j)$. Further, by definition of $C_{covered}$, each $j \in C_{covered}$ is supported on a fractional facility not in $S_0$. To see this, note that for all $j \in C_{covered}$, there exists $\bar{j} \in \bar{C}$ such that $F_j \cap F_{\bar{j}} \neq \emptyset$, and $F_{\bar{j}} \cap S_0 = \emptyset$.
	
Then we can use Lemma \ref{lem_outclosefacil} for each fractional facility $i \notin S_0$ to bound the cost of connecting all $C_{covered}$-clients supported on $i$ to $\bar{S}$. For each such fractional facility $i \notin S_0$, the connection cost of these clients is at most $O(\frac{\rho}{\delta}) U$. By summing over all $O(1)$-many fractional $i \notin S_0$, we connect all of $C_{covered}$ at additive cost at most $O(\frac{\rho}{\delta}) U$.
	
Finally, we handle the clients in $C_{full} \cup (C^* \setminus C_0)$. For convenience, we denote this set of clients by $\hat{C}$. Using Lemma \ref{lem_fullextrastep}, every $j \in \hat{C}$ has a facility in $S \cup S_0$ at distance at most $(2 +\alpha) L(\ell_j)$. There are a few cases to consider. For the first case, we consider the clients $\{j \in \hat{C} \mid F_j \setminus (S_0 \cup F_{=1}) \neq \emptyset\}$, that is the set of $\hat{C}$-clients whose $F$-balls contain some cheap, fractional facility. By an analogous argument as for $C_{covered}$, we can apply Lemma \ref{lem_outclosefacil} to each fractional $i \notin S_0$ to bound the cost of connecting $\{j \in \hat{C} \mid F_j \setminus (S_0 \cup F_{=1}) \neq \emptyset\}$ to $S \cup S_0$ by $O(\frac{\rho}{\delta}) U$.
	
Then it remains to consider the clients $j \in \hat{C}$ such that $F_j \subset F_{=1} \cup S_0$. For such a client $j$, if $F_j = \emptyset$, then $j$'s contribution to the objective of $\lpreroute$ is exactly $L(\ell_j)$, so connecting $j$ to $S \cup S_0$ costs at most $(2 +\alpha)$ times $j$'s contribution to the objective of $\lpreroute$.
	
Similarly, if $F_j$ happens to contain a facility in $F_{=1} \cup S_0$, then we can simply connect $j$ to the closest such facility in $F_j$. Note that $F_{=1} \cup S_0 \subset S \cup S_0$, so $j$'s connection cost in this case it as most its contribution to the objective in $\lpreroute$. In conclusion, the connection cost of $\{j \in \bar{C} \mid F_j \subset F_{=1} \cup S_0 \}$ is at most $(2 +\alpha) Opt(\lpreroute)$. Summing the costs of these different groups of clients completes the proof.

\section{Chain Decompositions of Extreme Points}\label{sec_pathdecomp}

In this section, we prove a more general version of Theorem \ref{thm_chaindecomp} that applies to set-cover-like polytopes with $r$ side constraints. In particular, we consider polytopes of the form:
	\[\mathcal{P} = \{y \in \mathbb{R}^F \mid y(F_j) =1\quad \forall j \in C^*,\, 0 \leq y \leq 1,\, Ay \leq b\}\]
, where $F$, $C^*$, and the $F$-balls are defined as in $\lpreroute$, and $Ay \leq b$ is an arbitrary system of $r$ linear inequalities. We note that other than the $C_{part}$-, and $C_{full}$-constraints, $\mathcal{P}$ generalizes the feasible region of $\lpreroute$ by taking the system $Ay \leq b$ to be the $r_1$ knapsack constraints and $r_2$ coverage constraints.

Although we phrase $\mathcal{P}$ in terms of facilities and clients, one can interpret $\mathcal{P}$ as a set cover polytope with side constraints as saying that we must choose elements in $F$ to cover each set in the family $\mathcal{F} = \{F_j \mid j \in C^*\}$ subject to the constraints $Ay \leq b$. The main result of this section is that if $\mathcal{F}$ has bipartite intersection graph, there exists a chain decomposition of the extreme points of $\mathcal{P}$.

Note that $\mathcal{P}$ can also also be interpreted as the intersection of two partition matroid polytopes or as a bipartite matching polytope both with $r$ side constraints. Our chain decomposition theorem shares some parallels with the work of Grandoni, Ravi, Singh, and Zenklusen, who studied the structure of bipartite matching polytopes with $r$ side constraints \cite{DBLP:journals/mp/GrandoniRSZ14}.

\begin{thm}[General Chain Decomposition]\label{thm_chaingeneral}
	Suppose we have a polytope:
	\[\mathcal{P} = \{y \in \mathbb{R}^F \mid y(F_j) =1\quad \forall j \in C^*,\, 0 \leq y \leq 1,\, Ay \leq b\},\]
	such that $F$ is a finite ground set of elements (facilities), $\{F_j \subset F \mid j \in C^*\}$ is a set family indexed by $C^*$ (clients), and $Ay \leq b$ is a system of $r$ linear inequalities. Further, let $\bar{y}$ be an extreme point for $\mathcal{P}$ such that no non-negativity constraint is tight. If $\mathcal{F}$ has bipartite intersection graph, then $C^*_{<1}$ admits a partition into $3r$ chains along with at most $2r$ violating clients (clients that are not in any chain.)
\end{thm}

We will use the following geometric fact about extreme points of polyhedra:

\begin{fact}\label{fact_extremepoints}
	Let $\mathcal{P}$ be a polyhedron in $\mathbb{R}^n$. Then $x \in \mathbb{R}^n$ is an extreme point of $\mathcal{P}$ if and only if there exist $n$ linearly independent constraints of $\mathcal{P}$ that are tight at $x$. We call such a set of constraints a \emph{basis} for $x$.
\end{fact}

Theorem \ref{thm_chaindecomp} follows almost immediately from Theorem \ref{thm_chaingeneral}.

\begin{proof}[Proof of Theorem \ref{thm_chaindecomp}]
	Let $\bar{y}$ be an extreme point of $\lpreroute$ such that no $C_{part}$-, $C_{full}$-, or non-negativity constraint is tight, and suppose $\lpreroute$ satisfies the Distinct Neighbors Now consider the polytope:
	\[\mathcal{P} = \{y \in \mathbb{R}^F \mid y(F_j) =1\quad \forall j \in C^*,\, 0 \leq y \leq 1,\, Ay \leq b\},\]
	, where $Ay \leq b$ consists of the $r_1$ knapsack constraints and $r_2$ coverage constraints of $\lpreroute$. We claim that $\bar{y}$ is an extreme point of $\mathcal{P}$. To see this, note that $\bar{y}$ is an extreme point of $\lpreroute$, so fix any basis for $\bar{y}$ using tight constraints of $\lpreroute$. By assumption, this basis uses no $C_{part}$-, $C_{full}$-, or non-negativity constraint. In particular, it only uses constraints of $\lpreroute$ that are also present in $\mathcal{P}$, so this basis certifies that $\bar{y}$ is an extreme point of $\mathcal{P}$.
	
	Further, by Proposition \ref{prop_bipartite}, the set family $\{F_j \mid j \in C^*\}$ has bipartite intersection graph. Then we can apply Theorem \ref{thm_chaingeneral} to $\bar{y}$ and polytope $\mathcal{P}$, which gives the desired result.
\end{proof}

\subsection{Proof of Theorem \ref{thm_chaingeneral}}

Now we go back to prove the more general chain decomposition theorem: Theorem \ref{thm_chaingeneral}. For all missing proofs, see \S \ref{appendix_pathdecomp}. Throughout this section, let 
\[\mathcal{P} = \{y \in \mathbb{R}^F \mid y(F_j) =1\quad \forall j \in C^*,\, 0 \leq y \leq 1,\, Ay \leq b\}\]
be a polytope satisfying the properties of Theorem \ref{thm_chaingeneral}. In particular, the intersection graph of $\mathcal{F} = \{F_j \mid j \in C^*\}$ is bipartite. Further, let $\bar{y}$ be an extreme point of $\mathcal{P}$ such that no non-negativity constraint is tight for $\bar{y}$.

The crux of our proof is the next lemma, which allows us to bound the complexity of the intersection graph with respect to the number of side constraints $r$. We prove the lemma by constructing an appropriate basis for $\bar{y}$. The next definition is useful for constructing a basis.

\begin{definition}
	For any subset $C' \subset C^*$, let $dim(C')$ denote the maximum number of linearly independent $C'$-constraints, so the constraint set $\{y(F_j) = 1 \mid j \in C'\}$.
\end{definition}

\begin{lem}\label{lem_fractionalextreme}
	Let $\bar{y}$ be an extreme point of $\mathcal{P}$ such that no non-negativity constraint is tight. Then the number of fractional facilities in $\bar{y}$ satisfies $\lvert F_{<1} \rvert \leq dim(C^*_{<1}) + r$ (recall that $r$ is the number of constraints of $Ay \leq b$.)
\end{lem}

Now, to find a chain decomposition of $C^*_{<1}$, first we find the violating clients. We note that every $F$-ball contains at least two facilities. The violating clients will be those clients whose $F$-balls contain strictly more than two facilities, so we let $V = \{j \in C^*_{<1} \mid \lvert F_j \rvert > 2 \}$ be the set of violating clients. It remains to bound the size of $V$, which follows from a standard counting argument.

\begin{prop}\label{prop_violating}
	$\lvert V \rvert \leq 2r$
\end{prop}

Now that we have decided on the violating clients, it remains to partition $C^*_{<1} \setminus V$ into the desired chains. Importantly, for all $j \in C^*_{<1} \setminus V$, we have $\lvert F_j \rvert = 2$. To find our chains, we consider the intersection graph of $C^*_{<1}$, so the intersection graph of the set family $\{F_j \mid j \in C^*_{<1}\}$. We let $G$ denote this graph. Note that $G$ is a subgraph of the standard intersection graph, so it is also bipartite by assumption.

We consider deleting the vertices $V$ from $G$, which breaks $G$ into some connected components, say $H_1, \dots, H_\ell$. Let $V_k$ denote the vertex set of $H_k$, so we have that $V_k \cup \dots \cup V_k$ partitions $C^*_{<1} \setminus V$. Further, for all $k \in [\ell]$, every $F$-ball for clients in $V_k$ contains exactly two facilities, and every facility is in at most two $F$-balls. Translating these statements into properties of the intersection graph, we can see that every vertex of $H_k$ has degree at most two, and $H_k$ is connected, so we can conclude that each $H_k$ is a path or even cycle (we eliminate the odd cycle case because the intersection graph is bipartite.)

\begin{prop}\label{prop_compchain}
	Each $V_k$ is a chain.
\end{prop}

To complete the proof, it remains to upper bound the number of chains. To do this, we first split the inequality given by Lemma \ref{lem_fractionalextreme} into the contribution by each $H_k$. Importantly, we observe that the $F(V_k)$'s are disjoint for all $k$ because the $V_k$'s correspond to distinct connected components. Then we can write:
\[\sum\limits_{k \in [\ell]} \lvert F(V_k) \rvert \leq \lvert F_{<1} \rvert \leq dim(C^*_{<1}) + r \leq \sum\limits_{k \in [\ell]} dim(V_k) + dim(V) + r \leq \sum\limits_{k \in [\ell]} dim(V_k) + 3r.\]
The way to interpret this inequality is that each chain, $V_k$, has a budget of $dim(V_k)$ fractional facilities to use in its chain, but we have an extra $3r$ facilities to pay for any facilities beyond each $V_k$'s allocated budget. We will show that each chain uses at least one extra facility from this $3r$ surplus, which allows us to upper bound $\ell$ by $3r$.

For the path components, we use the fact that every path has one more vertex than edge.

\begin{prop}\label{prop_path}
	If $H_k$ is a path, then $\lvert F(V_k) \rvert > dim(V_k)$.
\end{prop}

For the even cycle components, we show that the corresponding $C^*$-constraints are not linearly independent.

\begin{prop}\label{prop_cycle}
	If $H_k$ is an even cycle, then $\lvert F(V_k) \rvert > dim(V_k)$.
\end{prop}

Applying the above two propositions, we complete the proof by bounding $\ell$:
\[\sum\limits_{k \in [\ell]} dim(V_k) + 3r \geq \sum\limits_{k \in [\ell]} \lvert F(V_k) \rvert \geq \sum\limits_{k \in [\ell]} dim(V_k) + \ell \Rightarrow 3r \geq \ell\]

\bibliographystyle{alpha}
\bibliography{ref}

\newcommand{\etalchar}[1]{$^{#1}$}
\begin{thebibliography}{GMM{\etalchar{+}}03}

\bibitem[AGK{\etalchar{+}}04]{DBLP:journals/siamcomp/AryaGKMMP04}
Vijay Arya, Naveen Garg, Rohit Khandekar, Adam Meyerson, Kamesh Munagala, and
  Vinayaka Pandit.
\newblock Local search heuristics for k-median and facility location problems.
\newblock {\em {SIAM} J. Comput.}, 33(3):544--562, 2004.

\bibitem[BPR{\etalchar{+}}17]{DBLP:journals/talg/ByrkaPRST17}
Jaroslaw Byrka, Thomas~W. Pensyl, Bartosz Rybicki, Aravind Srinivasan, and Khoa
  Trinh.
\newblock An improved approximation for \emph{k}-median and positive
  correlation in budgeted optimization.
\newblock {\em {ACM} Trans. Algorithms}, 13(2):23:1--23:31, 2017.

\bibitem[BPR{\etalchar{+}}18]{DBLP:journals/algorithmica/ByrkaPRSST18}
Jaroslaw Byrka, Thomas~W. Pensyl, Bartosz Rybicki, Joachim Spoerhase, Aravind
  Srinivasan, and Khoa Trinh.
\newblock An improved approximation algorithm for knapsack median using
  sparsification.
\newblock {\em Algorithmica}, 80(4):1093--1114, 2018.

\bibitem[Che08]{Chen08}
Ke~Chen.
\newblock A constant factor approximation algorithm for $k$-median clustering
  with outliers.
\newblock In {\em SODA}, pages 826--835, 2008.

\bibitem[CKMN01]{CharikarKMN01}
Moses Charikar, Samir Khuller, David~M. Mount, and Giri Narasimhan.
\newblock Algorithms for facility location problems with outliers.
\newblock In S.~Rao Kosaraju, editor, {\em Proceedings of the Twelfth Annual
  Symposium on Discrete Algorithms, January 7-9, 2001, Washington, DC, {USA}},
  pages 642--651. {ACM/SIAM}, 2001.

\bibitem[FKRS19]{DBLP:journals/talg/FriggstadKRS19}
Zachary Friggstad, Kamyar Khodamoradi, Mohsen Rezapour, and Mohammad~R.
  Salavatipour.
\newblock Approximation schemes for clustering with outliers.
\newblock {\em {ACM} Trans. Algorithms}, 15(2):26:1--26:26, 2019.

\bibitem[GLZ17]{guha2017distributed}
Sudipto Guha, Yi~Li, and Qin Zhang.
\newblock Distributed partial clustering.
\newblock In {\em Proceedings of the 29th ACM Symposium on Parallelism in
  Algorithms and Architectures}, pages 143--152. ACM, 2017.

\bibitem[GMM{\etalchar{+}}03]{GuhaMMMO03}
Sudipto Guha, Adam Meyerson, Nina Mishra, Rajeev Motwani, and Liadan
  O'Callaghan.
\newblock Clustering data streams: Theory and practice.
\newblock {\em TKDE}, 15(3):515--528, 2003.

\bibitem[GRSZ14]{DBLP:journals/mp/GrandoniRSZ14}
Fabrizio Grandoni, R.~Ravi, Mohit Singh, and Rico Zenklusen.
\newblock New approaches to multi-objective optimization.
\newblock {\em Math. Program.}, 146(1-2):525--554, 2014.

\bibitem[IQM{\etalchar{+}}20]{ImQMSZ20}
Sungjin Im, Mahshid~Montazer Qaem, Benjamin Moseley, Xiaorui Sun, and Rudy
  Zhou.
\newblock Fast noise removal for k-means clustering.
\newblock In {\em The 23rd International Conference on Artificial Intelligence
  and Statistics, {AISTATS} 2020, 26-28 August 2020, Online [Palermo, Sicily,
  Italy]}, pages 456--466, 2020.

\bibitem[JMS02]{Jain:2002:NGA:509907.510012}
Kamal Jain, Mohammad Mahdian, and Amin Saberi.
\newblock A new greedy approach for facility location problems.
\newblock In {\em Proceedings of the Thiry-fourth Annual ACM Symposium on
  Theory of Computing}, STOC '02, pages 731--740, New York, NY, USA, 2002. ACM.

\bibitem[JV01]{DBLP:journals/jacm/JainV01}
Kamal Jain and Vijay~V. Vazirani.
\newblock Approximation algorithms for metric facility location and
  \emph{k}-median problems using the primal-dual schema and lagrangian
  relaxation.
\newblock {\em J. {ACM}}, 48(2):274--296, 2001.

\bibitem[KKN{\etalchar{+}}15]{DBLP:journals/mor/Krishnaswamy0NS15}
Ravishankar Krishnaswamy, Amit Kumar, Viswanath Nagarajan, Yogish Sabharwal,
  and Barna Saha.
\newblock Facility location with matroid or knapsack constraints.
\newblock {\em Math. Oper. Res.}, 40(2):446--459, 2015.

\bibitem[KLS18]{DBLP:conf/stoc/KrishnaswamyLS18}
Ravishankar Krishnaswamy, Shi Li, and Sai Sandeep.
\newblock Constant approximation for k-median and k-means with outliers via
  iterative rounding.
\newblock In Ilias Diakonikolas, David Kempe, and Monika Henzinger, editors,
  {\em Proceedings of the 50th Annual {ACM} {SIGACT} Symposium on Theory of
  Computing, {STOC} 2018, Los Angeles, CA, USA, June 25-29, 2018}, pages
  646--659. {ACM}, 2018.

\bibitem[LG18]{li2018distributed}
Shi Li and Xiangyu Guo.
\newblock Distributed $ k $-clustering for data with heavy noise.
\newblock In {\em Advances in Neural Information Processing Systems}, pages
  7838--7846, 2018.

\bibitem[LS16]{DBLP:journals/siamcomp/LiS16}
Shi Li and Ola Svensson.
\newblock Approximating k-median via pseudo-approximation.
\newblock {\em {SIAM} J. Comput.}, 45(2):530--547, 2016.

\bibitem[MKC{\etalchar{+}}15]{MalkomesKCWM15}
Gustavo Malkomes, Matt Kusner, Wenlin Chen, Kilian Weinberger, and Benjamin
  Moseley.
\newblock Fast distributed $k$-center clustering with outliers on massive data.
\newblock In {\em NIPS}, pages 1063--1071, 2015.

\end{thebibliography}

\appendix

\section{Missing Proofs from  \S \ref{sec_iteroverview}: Construction of $\lpreroute$}\label{sec_appendix_LP}

\begin{proof}[Proof of Proposition \ref{prop_LPfacil}]
	Let $\mathcal{I}$ be the given instance of GKM and $(x^*, y^*)$ be an optimal solution to $\lpbasic$.
	
	Observe that if $x^*_{ij} \in \{0, y^*_i\}$ for all $i \in F, j \in C$, then we can define $F_j = \{i \in F \mid x^*_{ij} > 0\}$ for all $j \in C$. It is easy to verify in this case that $y^*$ is feasible for $\lpfacil$ and achieves the same objective value in $\lpfacil$ as $(x^*, y^*)$ achieves in $\lpbasic$, which completes the proof.
	
	Thus our goal is to duplicate facilities in $F$ and re-allocate the $x$- and $y$-values appropriately until $x^*_{ij} \in \{0, y^*_i\}$ for all $i \in F, j \in C$. To prevent confusion, let $F$ denote the original set of facilities, and let $F'$ denote the modified set of facilities, where make $n = \lvert C \rvert$ copies of each facility in $F$, so for each $i \in F$, we have copies $i_1, \dots, i_n \in F'$.
	
	Now we define $x' \in [0,1]^{F' \times C}$ and $y' \in [0,1]^{F'}$ with the desired properties. For each $i \in F$, we assume without loss of generality that $0 \leq x_{i1} \leq x_{i2} \leq \dots \leq x_{in} \leq y_i$. We define $x'_{i_1 1}, \dots, x'_{i_n n}$ and $y'_{i_1}, \dots, y'_{i_n}$ recursively:
	
	 Let $y'_{i_1} = x_{i1}$ and $x'_{i_1 j} = x_{ij}$ for all $j \in [n]$.
	 
	 Now for $k > 1$, let $y'_{i_k} = x_{ik} - x_{i(k-1)}$ and $x'_{i_k j} = \begin{cases} 0 &, j < k\\
	 y'_{i_k} &, j \geq k
	 \end{cases}$\quad for all $j \in [n]$.
	 
	 It is easy to verify that $(x',y')$ is feasible for $\lpbasic$ (after duplicating facilities) and $x'_{ij} \in \{0, y'_i\}$ for all $i \in F', j \in C$, as required. Further, it is clear that this algorithm is polynomial time.
\end{proof}

\begin{proof}[Proof of Proposition \ref{lem_disc}]
    If $d(p,q) = 0$, then the claim is trivial. Suppose $d(p,q) \geq 1$. We can rewrite $d(p,q) = \tau^{\ell + f}$ for some $\ell \in \mathbb{N}, f \in [0,1)$. Also, for convenience we define $\beta = \log_\tau \alpha$. Because $\log_e \alpha$ is uniformly distributed in $[0, \log_e \tau)$, it follows that $\beta$ is uniformly distributed in $[0,1)$.

	It follows, $d(p,q)$ is rounded to $\alpha \tau^\ell = \tau^{\ell + \beta}$ exactly when $\beta \geq f$, and otherwise $d(p,q)$ is rounded to $\tau^{\ell + \beta + 1}$ when $\beta < f$. Thus we compute:
	
	\begin{align*}
		\mathbb{E}[d'(p,q)] &= \int_{\beta = 0}^f \tau^{\ell + \beta + 1} ~d\beta + \int_{\beta = f}^1 \tau^{\ell + \beta} ~d\beta\\
		&= \frac{1}{\log_e \tau} (\tau^{\ell + \beta + 1}\rvert_{\beta = 0}^f + \tau^{\ell + \beta}\rvert_{\beta = f}^1)\\
		&= \frac{1}{\log_e \tau} (\tau^{\ell + f + 1} - \tau^{\ell + 1} + \tau^{\ell + 1} - \tau^{\ell + f})\\
		&= \frac{1}{\log_e \tau} (\tau^{\ell + f + 1}  - \tau^{\ell + f})\\
		&= \frac{\tau - 1}{\log_e \tau} d(p,q).
	\end{align*}	
\end{proof}

\begin{proof}[Proof of Proposition \ref{prop_bipartite}]
	Assume for contradiction that the intersection graph of $\mathcal{F}$ is not bipartite, so there exists an odd cycle, say $j_1 \rightarrow \dots \rightarrow j_\ell \rightarrow j_1$ such that each vertex $j_1, \dots, j_\ell \in C^*$. Further, along each edge $j_k \rightarrow j_{k+1}$, we have $F_{j_k} \cap F_{j_{k+1}} \neq \emptyset$, so $\ell_{j_k}$ and $\ell_{j_{k+1}}$ differ by exactly one. In particular, the radius level can either increase by one or decrease by one along each edge.
	
	Consider traversing the cycle starting from $j_1$ all the way to $j_\ell$ and then back to $j_1$, and count the number of increases and decreases along the way. The number of increases and decreases must be equal when we return to $j_1$, but this cycle has an odd number of edges, so the number of increases and decreases cannot be the same. This is a contradiction.
\end{proof}

\begin{proof}[Proof of Proposition \ref{prop_degree}]
	Assume for contradiction that there exists a facility $i$ such that $i \in F_{j_1} \cap F_{j_2} \cap F_{j_3}$ for distinct clients $j_1, j_2, j_3 \in C^*$. Then the intersection graph of $C^*$ contains an odd cycle $j_1 \rightarrow j_2 \rightarrow j_3 \rightarrow j_1$. This contradicts the fact that the intersection graph is bipartite.
\end{proof}

\begin{proof}[Proof of Lemma \ref{lem_makeLPbasic}]
	Our algorithm is to first run the algorithm guaranteed by Lemma \ref{prop_LPfacil} to obtain $\lpfacil$ and the $F$-balls such that $Opt(\lpfacil) \leq Opt(\mathcal{I})$. Then we follow the construction in \S \ref{sec_iteroverview}- that is, we randomly discretize the distances to  obtain $d'$, define the $F$- and $B$- balls and radius levels, and initialize $C_{part} = C$, $C_{full} = \emptyset$ and $C^* = \emptyset$. This completes the description of $\lpreroute$.
	
	By Proposition \ref{prop_LPreroute}, we have $\mathbb{E}[Opt(\lpreroute)] \leq \frac{\tau - 1}{\log_e \tau} Opt(\lpfacil) \leq \frac{\tau - 1}{\log_e \tau} Opt(\mathcal{I})$, as required. Finally, it is easy to check that $\lpreroute$ satisfies all Basic Invariants.
\end{proof}

\section{Missing Proofs from \S \ref{sec_iter} - \ref{sec_pseudoalg}: Analysis of \mainalg}\label{appendix_pseudoanalysis}

In this section, we present all missing proofs from the analyses of \mainalg~ and its sub-routines \iteralg~ and \altreroute.

\subsection{Analysis of \iteralg}\label{sec_iteranalysis}

The goal of this section is to prove Theorem \ref{thm_mainiter}. First we show that \iteralg~ maintains all Basic Invariants. It is easy to see that the first three Basic Invariants are maintained by \iteralg, so we only prove the last two.

\begin{lem}[Basic Invariant \ref{invar_basic}(\ref{invar_ell})]\label{lem_iterlevels}
	\iteralg~ maintains the invariant that $\ell_j \geq -1$ for all $j \in C$.
\end{lem}
\begin{proof}
	Consider any $j \in C$. Suppose the invariant holds at the beginning of \iteralg, so initially we have $\ell_j \geq -1$. Note that a necessary condition for decreasing $\ell_j$ is that $\bar{y}(B_j) \leq 1$ is tight at some iteration of \iteralg, and in this case we decrease $\ell_j$ by one.
	
	Suppose $\ell_j = -1$. Then $L(\ell_j - 1) = -1$, so $B_j = \emptyset$. Thus it cannot be the case that  $\bar{y}(B_j) \leq 1$ is tight. We conclude that for all $j \in C$, we never decrease $\ell_j$ beyond negative one.
\end{proof}

\begin{lem}[Basic Invariant \ref{invar_basic}(\ref{invar_neighbor}): Distinct Neighbors Property]\label{lem_iterneighbors}
	\iteralg~ maintains the Distinct Neighbors Property.
\end{lem}
\begin{proof}
	It suffices to show that \itersub~ maintains the Distinct Neighbors Property, because in \iteralg, the only time $C^*$ is modified is when \itersub~ is called. Thus we consider an arbitrary call to $\itersub(j)$. To prevent confusion, let $C^*$ denote the status of $C^*$ before the call to \itersub. If we do not move $j$ from $C_{full}$ to $C^*$, then the invariant is clearly maintained, so we may assume that we move $j$ from $C_{full}$ to $C^*$.
	
    Then for all $j' \in C^*$ such that $F_j \cap F_{j'} \neq \emptyset$, we have $\ell_{j'} \geq \ell_j + 1$. Finally, after adding $j$ to $C^*$, we remove all such $j'$ with $\ell_{j'} \geq \ell_j + 2$, so for all remaining clients in $C^*$ whose $F$-balls intersect $j$'s, their radius level is exactly one larger than $j$'s.
\end{proof}

To show that \iteralg~ weakly decreases $Opt(\lpreroute)$, it suffices to show that each iteration weakly decreases $Opt(\lpreroute)$. We show that in any iteration of \iteralg, the $\bar{y}$ computed at the beginning of the iteration is still feasible after we update $\lpreroute$ in that iteration. Then, we show that $\bar{y}$ achieves the same objective value before and after the updates.

\begin{lem}\label{lem_iterdecrease}
    Each iteration of \iteralg~ weakly decreases $Opt(\lpreroute)$.
\end{lem}
\begin{proof}
	Consider any iteration of \iteralg. Let $\bar{y}$ be the optimal solution to $\lpreroute$ computed at the beginning of the iteration. There are three possible modifications we can make to $\lpreroute$ in this iteration: remove a facility from $F$, move a client from $C_{part}$ to $C_{full}$, or shrink a $F$-ball for a client in $C_{full}$. For each operation, we show that $\bar{y}$ is still feasible and achieves the same objective value afterwards.
	
	For the first operation, if we remove a facility $i$ from $F$, then it must be the case that $\bar{y}_i = 0$. Thus it is immediate that $\bar{y}$ (restricted to $F \setminus \{i\}$) is feasible after deleting $i$, and achieves the same objective value.
	
	Otherwise, suppose there exists $j \in C_{part}$ such that $\bar{y}(F_j) = 1$, and we move $j$ from $C_{part}$ to $C_{full}$ and then $\itersub(j)$. Thus $j$ ends up it either $C_{full}$ or $C^*$ at the end of this iteration. In either case, we have $\bar{y}(B_j) \leq \bar{y}(F_j) = 1$, so $\bar{y}$ satisfies the corresponding constraint after updating $\lpreroute$. Further, because $\bar{y}(F_j) = 1$, we have:
	\[\sum\limits_{i \in F_j} d'(j,i) \bar{y}_i = \sum\limits_{i \in B_j} d'(j,i) \bar{y}_i + (1 - \bar{y}(B_j)) L(\ell_j),\]
	so contribution of $j$ to the objective before is the same as its contribution after.
	
	In the final case, suppose there exists $j \in C_{full}$ such that $\bar{y}(B_j) = 1$, so we shrink $F_j$ and then $\itersub(j)$. Let $\lpreroute'$ index the data at the end of the iteration, so $F_j' = B_j$. Then we have $\bar{y}(B_j') \leq \bar{y}(F_j') = 1$, so $\bar{y}$ satisfies the corresponding constraint for $j$ whether $j$ is in $C_{full}$ or $C^*$. To compare the contribution of $j$ to the objective in $\lpreroute$ and $\lpreroute'$, we compute:
	
	\begin{align*}
		\sum\limits_{i \in B_j} d'(i,j) \bar{y}_i + (1 - \bar{y}(B_j)) L_{\ell_j} &= 	\sum\limits_{i \in B_j} d'(i,j) \bar{y}_i + 0\\
		&= \sum\limits_{i \in B_j'} d'(i,j) \bar{y}_i + \sum\limits_{i \in F_j' \setminus B_j'} d'(i,j) \bar{y}_i\\
		&= 	\sum\limits_{i \in B_j'} d'(i,j) \bar{y}_i + (1 - \bar{y}(B_j')) L(\ell_j').
	\end{align*}
\end{proof}

Finally, we note that if \iteralg~ terminates, then it is clear that no $C_{part}$-, $C_{full}$-, or non-negativity constraint is tight for $\bar{y}$ by definition of our iterative operations. Thus it suffices to show that \iteralg~ terminates in polynomial time.
	
\begin{lem}\label{lem_iterterminate}
	\iteralg~ terminates in polynomial time.
\end{lem}
\begin{proof}
	It suffices to show that the number of iterations of \iteralg~ is polynomial. In each iteration, we make one of three actions. We either delete a facility from $F$, move a client from $C_{part}$ to $C_{full}$ or shrink a $F$-ball by one radius level for a client in $j \in C_{full}$.
	
	We can delete each facility from $F$ at most once, so we make at most $\lvert F \rvert$ deletions. Each client can move from $C_{part}$ to $C_{full}$ at most once, because we never move clients back from $C_{full}$ to $C_{part}$, so we do this operations at most $\lvert C \rvert$ times.
	
	Finally, observe that $\ell_j \geq -1$ for all $j \in C$ over all iterations by Basic Invariant \ref{invar_basic}(\ref{invar_ell}). We conclude that we can shrink each $F$-ball only polynomially many times.
\end{proof}

\subsection{Analysis of \altreroute}\label{sec_altanalysis}

To prove Lemma \ref{lem_configfacil}, which bounds the number of fractional facilities needed to have a candidate configuration, we first prove a bound on the number of factional clients needed. The bound on the number of facilities will follow.

\begin{lem}\label{lem_configclient}
	Suppose $\lpreroute$ satisfies all Basic Invariants, and let $\bar{y}$ be an optimal extreme point of $\lpreroute$ such that no $C_{part}$-, $C_{full}$-, or non-negativity constraint is tight. If $\lvert C^*_{<1} \rvert \geq 14r$, then there exist a candidate configuration in $C^*_{<1}$.
\end{lem}
\begin{proof}
	We claim that in order for $C^*_{<1}$ to have a candidate configuration, it suffices to have a chain of length at least four in $C^*_{<1}$. To see this, let $(j_1, j_2, j_3, j_4, \dots) \subset C^*_{<1}$ be a chain of length at least four. Then $F_{j_2} \cap F_{j_3} \neq \emptyset$, and by the Distinct Neighbors Property, either $\ell_{j_3} = \ell_{j_2} - 1$ or $\ell_{j_2} = \ell_{j_3} - 1$.
	
	We only consider the former case, because both cases are analogous. Thus, if $\ell_{j_3} = \ell_{j_2} - 1$, then we claim that $(j_2, j_3)$ forms a candidate configuration. We already have the first two properties of a candidate configuration. Now we verify the last two. Because $j_2$ and $j_3$ are part of a chain, we have $\lvert F_{j_2} \rvert = 2$ and $\lvert F_{j_3} \rvert = 2$. Further, $j_2$ has neighbors $j_1$ and $j_3$ along the chain. By Proposition \ref{prop_degree}, each facility in $F_{j_2}$ is in at most two $F$-balls for clients in $C^*$. In particular, one of the facilities in $F_{j_2}$ is shared by $F_{j_1}$ and $F_{j_2}$, and the other must be shared by $F_{j_2}$ and $F_{j_3}$. Thus, each facility in $F_{j_2}$ is in exactly two $F$-balls for clients in $C^*$. An analogous argument holds for $F_{j_3}$, so $(j_2, j_3)$ satisfies all properties of a candidate configuration, as required.
	
	Now suppose $\lvert C^*_{<1} \rvert \geq 14r$. By Theorem \ref{thm_chaindecomp}, $C^*_{<1}$ admits a chain decomposition into at most $3r$ chains and a set of at most $2r$ violating clients. Then at least $12r$ of the clients in $C^*_{<1}$ belong to the $3r$ chains. By averaging, there must exist a chain with size at least $\frac{12r}{3r} = 4$, as required.
\end{proof}

Lemma \ref{lem_configfacil} is a corollary of the above lemma.

\begin{proof}[Proof of Lemma \ref{lem_configfacil}]
	By the previous lemma, it suffices to show that $\lvert F_{<1} \rvert \geq 15r$ implies that $\lvert C^*_{<1} \rvert \geq 14r$. Applying Lemma \ref{lem_fractionalextreme}, we have:
	\[\lvert F_{<1} \rvert \leq dim(C^*_{<1}) + r \leq \lvert C^*_{<1} \rvert + r\]
	, which combined with $\lvert F_{<1} \rvert \geq 15r$ gives the desired result.
\end{proof}

\begin{proof}[Proof of Theorem \ref{thm_mainaltreroute}]
    It is clear that \altreroute~ can be implemented to run in polynomial time and maintains all Basic Invariants, because \altreroute~ only moves clients from $C^*$ to $C_{full}$. Thus it remains to show that \altreroute~ weakly decreases $Opt(\lpreroute)$.
    
    Using the same strategy as in \Cref{lem_iterdecrease}, we let $\lpreroute$ denote the LP at the beginning of \altreroute~ and $\bar{y}$ the optimal extreme point of $\lpreroute$. Then we show that $\bar{y}$ is feasible after the operation and achieves the same objective value.
    
    In this call to \altreroute, we move some client $j$ from $C^*$ to $C_{full}$. We have $\bar{y}(B_j) \leq \bar{y}(F_j) = 1$, so $\bar{y}$ is feasible after \altreroute. Finally, moving $j$ from $C^*$ to $C_{full}$ does not affect its contribution to the objective.
\end{proof}

\subsection{Analysis of \mainalg}\label{sec_pseudoanalysis}

\begin{proof}[Proof of Theorem \ref{thm_pseudo}]
	It is immediate that \mainalg~ maintains all Basic Invariants by Theorems \ref{thm_mainiter} and \ref{thm_mainaltreroute}. Further, both of these sub-routines are polynomial time, so to show that \mainalg~ runs in polynomial time, it suffices to show that the number of calls to  \iteralg~ and \altreroute~ is polynomial.
	
    In every iteration of \mainalg, either we terminate or we are guaranteed to move a client from $C^*$ to $C_{full}$ in \altreroute. Each client can be removed from $C^*$ only polynomially many times, because each time a client is removed, in order to be re-added to $C^*$, it must be the case that we shrunk the $F$-ball of that client. However, by Basic Invariant \ref{invar_basic}(\ref{invar_ell}), we can shrink each $F$-ball only polynomially many times.
    
    Finally, upon termination of \mainalg, there is no candidate configuration, so Lemma \ref{lem_configfacil} implies that $\bar{y}$ has at most $15r$ fractional variables.
\end{proof}

\section{Missing Proofs from \S \ref{sec_trueapprox}}\label{sec_trueproofs}

\begin{proof}[Proof of Lemma \ref{lem_kclosefacil}]
	We let $i^* \in S$ be the closest facility to $i$ in $S$. We show that the cost of connecting $C'$ to $i^*$ is at most $O(\frac{\rho}{\delta}) U$. To do so, we partition $C'$ into two sets of clients: those that are far from $i$ relative to $d(i,i^*)$, and those that are close to $i$. In particular, let $\gamma > 0$ be a constant that we choose later. Then we partition $C'$ into $C'_{far}$ and $C'_{close}$, where:
	\[C'_{far} = \{j \in C' \mid d(j,i) \geq \gamma d(i,i^*)\},\]
	and 
	\[C'_{close} = \{j \in C' \mid d(j,i) > \gamma d(i,i^*)\}.\]
	First we bound the connection cost of $C'_{far}$ to $i^*$ using the fact that $i \notin S_0$, so Extra Invariant \ref{invar_kextra}(\ref{invar_ksparsefacil}) says that $\sum\limits_{j \mid i \in F_j} d(i,j) \leq 2 \rho U$. Thus we compute:
	\[\sum\limits_{j \in C'_{far}} d(j,i^*) \leq (1 + \frac{1}{\gamma}) \sum\limits_{j \in C \mid i \in F_j} d(j,i) \leq (1 + \frac{1}{\gamma}) O(\rho) U\]
	Now suppose $C'_{close} \neq \emptyset$. Fix any $j^* \in C'_{close}$. Then for all $j \in C'_{close}$, we have $d(j,j^*) \leq d(j,i) + d(j^*,i) \leq 2 \gamma d(i,i^*)$. It follows that $C'_{close} \subset B_{C}(j^*, 2\gamma d(i,i^*))$. Our strategy is to use Extra Invariant \ref{invar_kextra}(\ref{invar_ksparserad}):
	\[\lvert B_{C}(j^*, \delta r) \rvert r \leq \rho U,\]
	for all $r \leq R_{j^*}$. Thus we want $2 \gamma d(i,i^*) \leq \delta R_{j^*}$. To lower bound $R_{j^*}$ with respect to $d(i,i^*)$, we use our assumption that there exists some $\bar{i} \in S$ such that $d(j^*,\bar{i}) \leq \alpha L(\ell_{j^*})$, where $L(\ell_{j^*}) \leq \tau R_{j^*}$ by Extra Invariant \ref{invar_kextra}(\ref{invar_ksparserad}). Thus we have:
	\[d(j^*,\bar{i}) \leq \alpha \tau R_{j^*}.\]
	Further, using the triangle inequality and the fact that $i^*$ is the closest facility to $i$ in $S$, we have:
	\[d(j^*, \bar{i}) \geq d(i, \bar{i}) - d(i,j^*) \geq d(i,i^*) - d(i,j^*) \geq (1 - \gamma) d(i,i^*).\]
	Combining these two inequalities gives the lower bound $R_{j^*} \geq \frac{1 - \gamma}{\alpha \tau} d(i,i^*)$.

    Now we are ready to choose $\gamma$. Recall that we want $2 \gamma d(i,i^*) \leq \delta R_{j^*}$, so it suffices to choose $\gamma$ such that:
    \[2 \gamma d(i,i^*) \leq \delta \frac{1 - \gamma}{\alpha \tau} d(i,i^*)\]
    Routine calculations show that we can take $\gamma = \Theta(\delta)$ to satisfy this inequality. Now with this choice of $\gamma$, we can bound:
    \begin{align*}
        \sum\limits_{j \in C'_{close}} d(j, i^*) &\leq (1 + \gamma)\lvert C'_{close} \rvert d(i,i^*)\\
        &\leq (1 + \gamma) \lvert B_{C}(j^*, \delta R_{j^*}) \rvert d(i,i^*)\\
        &\leq (1 + \gamma) \frac{\rho U}{R_{j^*}} d(i,i^*)\\
        &\leq \frac{\rho U}{R_{j^*}} O(1)R_{j^*} = O(\rho) U
    \end{align*}
    
    To conclude the proof, the connection cost of $C'_{far}$ is at most $(1 + \frac{1}{\gamma}) O(\rho) U = O(\frac{\rho}{\delta})U$ and the connection cost of $C'_{close}$ is at most $O(\rho)U$. Summing these costs gives the desired result.
\end{proof}

\begin{proof}[Proof of Lemma \ref{lem_outclosefacil}]
	We let $i^* \in S$ be the closest facility to $i$ in $S$. We show that the cost of connecting $C'$ to $i^*$ is at most $O(\frac{\rho}{\delta}) U$. To do so, we partition $C'$ into two sets of clients: those that are far from $i$ relative to $d(i,i^*)$, and those that are close to $i$. In particular, let $\gamma > 0$ be a constant that we choose later. Then we partition $C'$ into $C'_{far}$ and $C'_{close}$, where:
	\[C'_{far} = \{j \in C' \mid d(j,i) \geq \gamma d(i,i^*)\}\]
	, and 
	\[C'_{close} = \{j \in C' \mid d(j,i) > \gamma d(i,i^*)\}\]
	First we bound the connection cost of $C'_{far}$ to $i^*$ using the fact that $i \notin S_0$, so Extra Invariant \ref{invar_outextra}(\ref{invar_outsparsefacil}) says that $i$ is cheap. Thus we compute:
	\[\sum\limits_{j \in C'_{far}} d(j,i^*) \leq (1 + \frac{1}{\gamma}) \sum\limits_{j \in C \mid i \in F_j} d(j,i) \leq (1 + \frac{1}{\gamma}) O(\rho) U\]
	Now suppose $C'_{close} \neq \emptyset$. Importantly, all of these clients are within distance $\gamma d(i,i^*)$ of $i$, so we have $C'_{close} \subset B_{C}(i, \gamma d(i,i^*))$. Our strategy to bound the connection cost of $C'_{close}$ is to leverage Extra Invariant \ref{invar_outextra}(\ref{invar_outsparserad}), so in particular we want to use the fact:
	\[\lvert \{j \in B_{C}(i, \frac{\delta t}{4 + 3 \delta}) \mid R_j \geq t \} \rvert \leq \frac{\rho (1 + 3\delta / 4)}{1 - \delta/4} \frac{U}{t}\]
	for any $t > 0$. We want to choose $\gamma, t > 0$ such that $\gamma d(i,i^*) \leq \frac{\delta t}{4 + 3 +\delta}$ and $R_j \geq t$ for all $j \in C'_{close}$. To see why this is useful, for such $\gamma$ and $t$, we have $C'_{close} \subset \{j \in B_{C}(i, \frac{\delta t}{4 + 3 \delta}) \mid R_j \geq t \}$. Then we can bound:
	\begin{align*}
		\sum\limits_{j \in C'_{close}} d(j,i^*) &\leq \sum\limits_{j \in C'_{close}}(1 + \gamma)d(i,i^*)\\
		&= (1 + \gamma) \lvert C'_{close} \rvert d(i,i^*)\\
		&\leq (1 + \gamma) \lvert \{j \in B_{C}(i, \frac{\delta t}{4 + 3 \delta}) \mid R_j \geq t \} \rvert d(i,i^*)\\
		&\leq (1 + \gamma) (\frac{\rho (1 + 3\delta / 4)}{1 - \delta/4} \frac{U}{t}) d(i,i^*)
	\end{align*}
	Now we go back and specify our choice of $\gamma$ and $t$, which will allow us to complete the bound of the connection costs. First we lower bound $R_j$ in terms of $d(i,i^*)$ for any $j \in C'_{close}$. We recall that by assumption there exists some $\bar{i} \in S$ such that $d(j,\bar{i}) \leq \alpha L(\ell_j)$, where $L(\ell_j) \leq \tau R_j$ by Extra Invariant \ref{invar_outextra}(\ref{invar_outsparserad}). Thus we have:
	\[d(j,\bar{i}) \leq \alpha \tau R_j\]
	Further, using the triangle inequality and the fact that $i^*$ is the closest facility to $i$ in $S$, we have:
	\[d(j, \bar{i}) \geq d(i, \bar{i}) - d(i,j) \geq d(i,i^*) - d(i,j) \geq (1 - \gamma) d(i,i^*)\]
	Combining this inequality with the upper bound on $d(j, \bar{i})$ gives that $R_j \geq \frac{1 - \gamma}{\alpha \tau} d(i,i^*)$ for all $j \in C'_{close}$. Then we define $t = \frac{1 - \gamma}{\alpha \tau} d(i,i^*)$. This gives us $R_j \geq t$ for all $j \in C'_{close}$. Now we can choose $\gamma > 0$ satisfying:
	\[\gamma d(i,i^*) \leq \frac{\delta t}{4 + 3 \delta} \Rightarrow \gamma \leq \frac{\delta}{4 + 3 \delta} \frac{1 - \gamma}{\alpha \tau}\]
	Taking $\gamma = \frac{\delta}{12 \alpha \tau} = \Theta(\delta)$ suffices.
	
	Using these choices of $\gamma$ and $t$, we can bound:
	\[\sum\limits_{j \in C'_{far}} d(j,i^*) = O(\frac{\rho}{\delta})U\]
	, and
	\[\sum\limits_{j \in C'_{close}} d(j,i^*) \leq (1 + \gamma)(\frac{\rho (1 + 3\delta / 4)}{1 - \delta/4} U)(\frac{\alpha \tau}{1 - \gamma}) = O(\rho) U\]
	Summing these two costs gives the desired result.
\end{proof}

\begin{proof}[Proof of Theorem \ref{thm_mainoutliers}]
    Let $\epsilon' > 0$. We will later choose $\epsilon'$ with respect to the given $\epsilon$ to obtain the desired approximation ratio and runtime. First, we choose parameters $\rho, \delta \in (0,1/2)$ and $U \geq 0$ for our pre-processing algorithm guaranteed by Theorem \ref{thm_kpre}. We take $\rho = {\epsilon'}^2$ and $\delta = \epsilon'$. We require that $U$ is an upper bound on $Opt(\I)$. Using a standard binary search idea, we can guess $Opt(\I)$ up to a multiplicative $(1 + \epsilon')$-factor in time $n^{O(1 / \epsilon')}$, so we guess $U$ such that $Opt(\I) \leq U \leq (1 + \epsilon') Opt(\I)$.
    
    With these choices of parameters, we run the algorithm guaranteed by Theorem \ref{thm_outpre} to obtain $n^{O(1/ \epsilon')}$ many sub-instances such that one such sub-instance is of the form $\I' = (F, C' \subset C, d, k, m' = m - \lvert C^* \setminus C' \rvert)$, where $\lpreroute$ for $\I'$ satisfies all Basic- and Extra Invariants, and we have:
    
    \begin{equation}\label{eq_outcost}
       \frac{\log_e \tau}{(\tau - 1)(1 + \epsilon' / 2)}\mathbb{E}[Opt(\lpreroute)] + \frac{1 - \epsilon'}{1+ \epsilon'} \sum\limits_{j \in C^* \setminus C'} d(j, S_0) \leq U 
    \end{equation}
    
    Then for each sub-instance output by the pre-processing, we run the algorithm guaranteed by Theorem \ref{thm_outlierapprox} to obtain a solution to each sub-instance. Finally, out of these solutions, we output the one that is feasible for the whole instance with smallest cost. This completes the description of our approximation algorithm for $k$-median with outliers. The runtime is $n^{O(1/ \epsilon')}$, so it remains to bound the cost of the output solution and to choose the parameters $\epsilon'$ and $\tau$ and $c$.
    
    To bound the cost, it suffices to consider the solution output on the instance $\I'$ where $\lpreroute$ satisfies all Basic- and Extra Invariants and Equation \ref{eq_outcost}. By running the algorithm guaranteed by Theorem \ref{thm_outlierapprox} on this $\lpreroute$, we obtain a feasible solution $S \subset F$ to $\I'$ such that $S_0 \subset S$, and the cost of connecting $m'$ clients from $C'$ to $S$ is at most $(2 + \alpha)Opt(\lpreroute) + O(\epsilon') U$, where $\alpha = \max(3 + 2\tau^{-c}, 1 + \frac{4 + 2\tau^{-c}}{\tau}, \frac{\tau^3 + 2\tau^2 + 1}{\tau^3 - 1})$. To extend this solution on the sub-instance to a solution on the whole instance $\I$, we must connect $m - m'  = \lvert C^* \setminus C' \rvert$ clients from $C \setminus C'$ to $S$. Because $S_0 \subset S$, applying Equation \ref{eq_kcost} allows us to upper bound the expected cost of connecting $m$ clients to $S$ by:
    \[(2 + \alpha)\mathbb{E}[Opt(\lpreroute)] + O(\epsilon') U + \sum\limits_{j \in C^* \setminus C'} d(j, S_0) \leq (2 + \alpha) \frac{\tau - 1}{\log_e \tau} \frac{(1 + \epsilon')^2}{1 - \epsilon'}U + O(\epsilon') U\]
    Now choosing $\tau > 1$ to minimize $\alpha' = (2 + \max(3, 1 + \frac{4}{\tau}, \frac{\tau^3 + 2\tau^2 + 1}{\tau^3 - 1})) \frac{\tau - 1}{\log_e \tau}$ (note that we ignore the $2\tau^{-c}$ terms), we obtain $\tau = 1.2074$ and $\alpha' = 6.947$. We can choose $c \geq 1$ sufficiently large with respect to $\tau$ such that $2\tau^{-c}$ is sufficiently small to guarantee $(2 + \alpha)\frac{\tau - 1}{\log_e \tau} \leq 6.947 + \epsilon'$
    
     Thus the expected cost of this solution is at most $(6.947 + \epsilon')\frac{(1 + \epsilon')^2}{1 - \epsilon'}U + O(\epsilon') U$, where $U \leq (1 + \epsilon') Opt(\I)$. Finally, by routine calculations, we can choose $\epsilon' = \theta(\epsilon)$ so that expected cost is at most $(6.947 + \epsilon) Opt(\I)$, as required. Note that the runtime of our algorithm is $n^{O(1/ \epsilon')} = n^{O(1/ \epsilon)}$.
\end{proof}

\section{Missing Proofs from \S \ref{sec_postoutlier}: Analysis of \outlieralg}\label{appendix_postproofs}

In this section we present all missing proofs from the analysis of \outlieralg~ and its subroutine \outlierpost.

\subsection{Missing Proofs from Analysis of \outlieralg}

\begin{proof}[Proof of Lemma \ref{lem_extremedisjoint}]
	Without loss of generality, we may assume that no facilities in $S_0$ are co-located with each other, so $\{F_j \mid j \in C_0\}$ is a disjoint family. This implies that $\{F_j \mid j \in C^*_{<1}\}$ is also a disjoint family. Now we construct a basis for $\bar{y}$. For every integral facility $i \in F_{=1}$, we add the constraint $y_i \leq 1$ to our basis. To complete the basis, we need to add $\lvert F_{<1} \rvert$ further linearly independent tight constraints.
	
	We recall that upon termination of \mainalg, no $C_{part}$-, $C_{full}$-, or non-negativity constraint is tight for $\bar{y}$, so the only constraints we can choose are the $C^*$-constraints, the $k$-constraint, or the coverage constraint. We claim that we cannot add any $C^*_{=1}$-constraint to our basis, because such a constraint is supported only on integral facilities, whose constraints we already added to the basis. However, we can add every $C^*_{<1}$-constraint to our basis, because their supports are disjoint and they contain no integral facilities. Thus, our partial basis consists of all tight integrality constraints and all $C^*_{<1}$-constraints.
	
	Now we consider adding the $k$-constraint to our basis. Importantly, the $k$-constraint is linearly independent with the current partial basis only if there exists at least one fractional facility not supported on any $F$-ball for clients in $C^*_{<1}$. Further, we may assume the $k$-constraint is tight (otherwise we cannot add it anyways), so there must be at least two fractional facilities not supported on any $F$-ball for clients in $C^*_{<1}$
	
	However, we note that each $F$-ball for clients in $C^*_{<1}$ contains at least two fractional facilities. Because these $F$-balls are disjoint, we have $\lvert F_{<1} \rvert \geq 2 \lvert C^*_{<1} \rvert$. If we cannot add the $k$-constraint to our basis, then we are done. This is because the coverage constraint is the only further constraint we can add the the basis, so we can bound $\lvert F_{<1} \rvert \leq \lvert C^*_{<1} \rvert +1$. This implies implies $\lvert F_{<1} \rvert \leq \frac{1}{2} \lvert F_{<1} \rvert + 1 \Rightarrow \lvert F_{<1} \rvert \leq 2$ using the previous inequality.
	
	Otherwise, we add the $k$-constraint to our basis, which implies $\lvert F_{<1} \rvert \geq 2 \lvert C^*_{<1} \rvert + 2$ because of the two fractional facilities outside $F(C^*_{<1})$ and $\lvert F_{<1} \rvert \leq \lvert C^*_{<1} \rvert + 2$ because the $k$-constraint and coverage constraint contribute are the only further constraints we can add. Again combining these two inequalities gives $\lvert F_{<1} \rvert \leq 2$.
\end{proof}

\begin{proof}[Proof of Lemma \ref{lem_base1}]
	Let $S = F_{=1}$ be the set of open facilities. It is immediate that $\lvert S \rvert \leq k$. Further, $\lpreroute$ satisfies all Extra Invariants, so $C_0 \subset C^*$. Because $\bar{y}$ is integral, it is clear that we open $S_0$. Thus it remains to show that the connecting $m$ clients to $S$ has cost at most $(2 + \alpha) Opt(\lpreroute)$.
	
	It suffices to show that connecting $C_{full}$ and $C^*$ to $S$ is enough clients and achieves the desired cost. Because $\bar{y}$ is integral and by definition of \mainalg, we have that no $C_{part}$-, $C_{full}$-, or non-negativity constraint is tight for $\bar{y}$. It follows, $F_j = \emptyset$ for all $j \in C_{part}$ and $B_j = \emptyset$ for all $j \in C_{full}$.
	
	Then the coverage constraint of $\lpreroute$ implies:
	\[\sum\limits_{j \in C_{part}} \bar{y}(F_j) \geq m - \lvert C_{full} \cup C^* \rvert\ \Rightarrow \lvert C_{full} \cup C^* \rvert \geq m\]
	, so this solution connects enough clients.
	
	To bound the cost, we compare the connection cost of each client with its contribution to the objective of $\lpreroute$. For all $j \in C^*$ we have $\bar{y}(F_j) = 1$, so $d(j,S) \leq \sum\limits_{i \in F_j} d'(j,i) \bar{y}_i$, which is exactly $j$'s contribution to $\lpreroute$.
	
	For all $j \in C_{full}$, we note that $\bar{y}(B_j) = 0$, so $j$'s contribution to $\lpreroute$ is exactly $L(\ell_j)$. We can apply \Cref{thm_mainreroute} with $\beta = 1$ and set of facilities $S$ to show that $d(j,S) \leq (2 + \alpha) L(\ell_j)$ for all $j \in C_{full}$. To conclude, the connection cost of each client is at most $(2 + \alpha)$ times its contribution to $\lpreroute$, a required.
\end{proof}

\begin{proof}[Proof of Lemma \ref{lem_base2cost}]
    Let $S = F_{=1} \cup \{a\}$ be the output solution. First, note that $\lvert S \rvert \leq k$ because $\bar{y}_a + \bar{y}_b = 1$, and those are the only two fractional variables. Second, because $a \notin S_0$, it must be the case that $b \notin S_0$, because $a,b$ are the only fractional facilities, and by Extra Invariant \ref{invar_outextra} (\ref{invar_outguessfacil}), there is one unit of open facility co-located at each $i \in S_0$. Note that this implies that $S_0 \subset F_{=1} \subset S$.
    
    Now there are two cases, either $a,b \in F_j$ for some $j \in C^*$, or $a,b \notin F(C^*)$. Note that in either case, we close $b$ and open $a$, so we still maintain the property that $\bar{y}(F_j) = 1$ for all $j \in C^*$. Thus, can apply \Cref{thm_mainreroute} with $\beta = 1$ and set of facilities $S$ to show that $d(j,S) \leq (2 + \alpha) L(\ell_j)$ for all $j \in C_{full} \cup C^*$.
	
	We consider connecting the clients $C_{part}(a) \cup C_{full} \cup C^*$ to $S$. First, we show that this is at least $m$ clients. We observe that $a$ and $b$ are the only fractional facilities in $\bar{y}$, and no $C_{part}$-constraint is tight. It follows that for all $j \in C_{part}$, we have $F_j = \{a\}$, $\{b\}$, or $\emptyset$, so we can rewrite the coverage constraint as:
	\[\lvert C_{part}(a) \rvert\bar{y}_a + \lvert C_{part}(b) \rvert \bar{y}_b \geq m - \lvert C_{full} \cup C^* \rvert\]
	Then because $\bar{y}_a + \bar{y}_b = 1$ and $\lvert C_{part}(a) \rvert \geq \lvert C_{part}(b) \rvert$ by assumption, we conclude that $\lvert C_{part}(a) \rvert \geq  m - \lvert C_{full} \cup C^* \rvert$, as required.
	
	Now it remains to show that the cost of connecting $C_{part}(a)$ to $a$ plus the cost of connecting $C_{full} \cup C^*$ to $S$ is at most $\alpha Opt(\lpreroute) + O(\frac{\rho}{\delta}) U$. First we handle $C_{part}(a)$. By assumption, $a \notin S_0$, so by Extra Invariant \ref{invar_outextra}(\ref{invar_outsparsefacil}), we can bound:
	\[\sum\limits_{j \in C_{part}(a)} d(j,a) \leq \sum\limits_{j \in C \mid a \in F_j} d(j,a) = O(\rho) U\]
	
	For the clients in $C_{full} \cup C^*$ that are not supported on $b$, closing $b$ does not affect their connection cost; in particular, each such client either has an integral facility in its $F$-ball to connect to (because we open $a$ and all other facilities are integral), or its $F$-ball is empty, and there exists an integral facility within $(2 +\alpha) L(\ell_j)$ to connect to. In both cases, each client's connection cost is at most $(2 + \alpha)$ times its contribution to the objective of $\lpreroute$.
	
	The only remaining cost to bound is the clients in $C_{full} \cup C^*$ that are supported on $b$. Let $C' = \{j \in C_{full} \cup C^* \mid b \in F_j\}$ be these clients. We show that the cost of connecting all of $C'$ to $S$ is at most $O(\frac{\rho}{\delta}) U$ using Lemma \ref{lem_outclosefacil}. Because every client in $j \in C_{full} \cup C^*$ has an open facility in $S$ within distance $(2 +\alpha) L(\ell_j)$, Lemma \ref{lem_outclosefacil} is applicable to $C'$ with set of facilities $S$ and $i = b \notin S \cup S_0$.

	To summarize, the connection costs of $C_{part}(a)$ and $C'$ are at most $O(\frac{\rho}{\delta}) U$, and the connection cost of all remaining clients in $C_{full} \cup C^*$ that are not supported on $b$ is at most $(2 + \alpha) Opt(\lpreroute)$, so the total connection cost, which is the sum of these terms, it at most the desired bound.
\end{proof}

\subsection{Missing Proofs from Analysis of \outlierpost}

\begin{proof}[Proof of \Cref{lem_greedystay}]
	Assume for contradiction that there exists $\bar{j}$ that is reached by the \textsc{For} loop, but $\bar{j}$ does not remain in $\bar{C}$ until termination. Note that $\bar{j}$ cannot be removed from $\bar{C}$ in the iteration that it is considered in the \textsc{For} loop. Thus there must exist a later iteration for client, say $j$ in which $\bar{j}$ is removed from $\bar{C}$. In the iteration for client $j$, there are only two possible ways that $\bar{j}$ is removed from $\bar{C}$. Either $F_j \cap F_{\bar{j}} \neq \emptyset$ or there exists a client $j' \in C_{part} \setminus C_{covered}$ such that $F_{j'}$ intersects both $F_j$ and  $F_{\bar{j}}$ and $\ell_{j'} \leq \ell_j - c$.
	
	In the former case, because we consider $\bar{j}$ before $j$, it must be the case that we removed $j$ from $\bar{C}$ in $\bar{j}$'s iteration. This is a contradiction. Similarly, in the second case if such a $j'$ exists, then in $\bar{j}$'s iteration, we either remove $j$ from $\bar{C}$ or add $j'$ to $C_{covered}$. In either case, this is a contradiction.
\end{proof}

\begin{proof}[Proof of \Cref{lem_afterinvar}]
	By assumption, the input to \outlierpost, $\lpreroute$, satisfies all Basic and Extra Invariants. To obtain $\lpreroute^1$ from $\lpreroute$, we delete some clients and facilities. Thus the only change to the $F$- and $B$-balls for clients in $C^1$ is that we possibly remove some facilities from their $F$- and $B$-balls; importantly, the radius levels, $\ell_j$ for all clients $j$, remain the same. Thus, it is easy to see that $\lpreroute^1$ satisfies all Basic Invariants.
	
	Similarly, for all remaining clients $j$, we have not changed $\ell_j$ or $R_j$, so the only Extra Invariant that requires some care to verify is Extra Invariant \ref{invar_outextra}(\ref{invar_outguessfacil}). However, we recall that to obtain ${C^*}^1$, we delete all clients in $C^* \setminus C_0$ from the instance, so ${C^*}^1 = C_0$. This is because $C_0 \subset C^*$ by the assumption that $\lpreroute$ satisfies all Extra Invariants.
\end{proof}

\begin{proof}[Proof of \Cref{lem_afterfeasible}]
	We note that $C^1 = C \setminus (C_{part}(\bar{S}) \cup C_{covered} \cup C_{full} \cup (C^* \setminus C_0))$ and $F^1 = F \setminus (F_{=1} \cup F(\bar{C}) \setminus S_0)$. Then we have $C_{part}^1 = C_{part} \setminus (C_{part}(\bar{S}) \cup C_{covered})$, $C_{full}^1 = \emptyset$, and ${C^*}^1 = C_0$.
		
	It suffices to show that all $C_{part}^1$-constraints, ${C^*}^1$-constraint, the $k$-constraint, and the coverage constraint are satisfied by $\bar{y}$ restricted to $F^1$.
	
	Consider any $j \in C_{part}^1$. We observe that the $F_j^1 \subset F_j$ for all $j \in C_{part}^1$ and $C_{part}^1 \subset C_{part}$.Then $\bar{y}(F_j^1) \leq \bar{y}(F_j) \leq 1$ for all $j \in C_{part}^1$.
	
	Now for any $j \in {C^*}^1 = C_0$, it suffices to show that we do not delete any copies of $i \in S_0$ when going from from $F$ to $F^1$, but this is immediate because $S_0 \subset F$, and we do not delete any facility from $S_0$ to obtain $F^1$. Thus, every ${C^*}^1$-constraint is satisfied.

	For the $k$-constraint, we have $\bar{y}(F) \leq k$. We want to show $\bar{y}(F^1) \leq k -  \lvert S \rvert$. By definition $\lvert \bar{S} \rvert = \sum\limits_{j \in \bar{C}} \bar{y}(F_j) = \bar{y}(F(\bar{C}))$, where the final equality follows because $\{F_j \mid j \in \bar{C}\}$ is a disjoint collection by Proposition \ref{prop_partnice}. Then we compute:
	\[\bar{y}(F^1) = \bar{y}(F) - \bar{y}(F(\bar{C})) - \lvert F_{=1} \setminus S_0 \rvert \leq k - \lvert \bar{S} \rvert - \lvert F_{=1} \setminus S_0 \rvert = k - \lvert S\rvert\]
	, as required.

	Finally, for the coverage constraint, we want to show:
	\[\sum\limits_{j \in C_{part}^1} \bar{y}(F_j^1) \geq m^1 - \lvert C_{full}^1 \cup{C^*}^1 \rvert,\]
    where $C_{part}^1 = C_{part} \setminus (C_{part}(\bar{S}) \cup C_{covered})$, $m^1 = m - \lvert C_{part}(\bar{S}) \cup C_{covered} \cup C_{full} \cup (C^* \setminus C_0) \rvert$, and $C_{full}^1 = \emptyset$ and ${C^*}^1 = C_0$, so $\lvert C_{full}^1 \cup{C^*}^1 \rvert = \lvert C_0 \rvert$. Thus we can re-write the coverage constraint as:
    \[\sum\limits_{j \in C_{part}^1} \bar{y}(F_j^1) \geq m - \lvert C_{part}(\bar{S}) \cup C_{covered} \cup C_{full} \cup{C^*} \rvert.\]

	Recall that the coverage constraint of $\lpreroute$ implies:
	\[\sum\limits_{j \in C_{part}} \bar{y}(F_j) \geq m - \lvert C_{full} \cup C^*\rvert.\]
	By splitting this inequality into the contribution by $C_{part}^1$ and $C_{part} \setminus C_{part}^1 = C_{part}(\bar{S}) \cup C_{covered}$, we obtain:
	\begin{align*}
	    \sum\limits_{j \in C_{part}} \bar{y}(F_j) &\geq m - \lvert C_{full} \cup C^*\rvert\\
	    \sum\limits_{j \in C_{part}^1} \bar{y}(F_j) + \sum\limits_{j \in C_{part}(\bar{S}) \cup C_{covered}} \bar{y}(F_j) &\geq m - \lvert C_{full} \cup C^*\rvert\\
	    \sum\limits_{j \in C_{part}^1} \bar{y}(F_j) + \sum\limits_{j \in C_{part}(\bar{S})} \bar{y}(F_j) &\geq m - \lvert C_{covered} C_{full} \cup C^*\rvert
	\end{align*}
	, where in the final inequality we use the fact that $\bar{y}(F_j) \leq 1$ for all $j \in C_{covered} \subset C_{part}$. Now, we recall that $F_j^1 = F_j \setminus F(\bar{C})$ for all $j \in C_{part}^1$, because $C_{part}^1 \subset C_{part}$. We can re-write:
	\[\sum\limits_{j \in C_{part}^1} \bar{y}(F_j) = \sum\limits_{j \in C_{part}^1} (\bar{y}(F_j^1) + \bar{y}(F_j \cap F(\bar{C})).\]
	To show that the coverage constraint is satisfied, it suffices to show:
	\[\sum\limits_{j \in C_{part^1}} \bar{y}(F_j \cap F(\bar{C})) + \sum\limits_{j \in C_{part}(\bar{S})} \bar{y}(F_j) \leq \lvert C_{part}(\bar{S}) \rvert.\]
	To see this, observe that the first sum is over all clients in $C_{part} \setminus C_{covered}$ supported on some facility in $F(\bar{C}) \setminus \bar{S}$ but none in $\bar{S}$ (otherwise these clients would be in $C_{part}(\bar{S})$.) The second sum is over all clients in $C_{part} \setminus C_{covered}$ supported on some facility in $\bar{S}$. Thus, recalling that $w_i = \lvert \{j \in C_{part} \setminus C_{covered} \mid i \in F_j\} \rvert$, we have:
	\[\sum\limits_{j \in C_{part^1}} \bar{y}(F_j \cap F(\bar{C})) + \sum\limits_{j \in C_{part}(\bar{S})} \bar{y}(F_j)  = \sum\limits_{j \in \bar{C}} \sum\limits_{i \in F_j} w_i \bar{y}_i \leq \lvert C_{part}(\bar{S}) \rvert,\]
	where in the final inequality we apply \Cref{prop_greedyoutlier}.
\end{proof}

\section{Missing Proofs from  \S \ref{sec_pathdecomp}: Chain Decomposition}\label{appendix_pathdecomp}

\begin{proof}[Proof of \Cref{lem_fractionalextreme}]
	We construct a basis $\bar{y}$. First, for each integral facility $i \in F_{=1}$, we add the integrality constraint $\bar{y}_i \leq 1$ to our basis. Thus we currently have $\lvert F_{=1} \rvert$ constraints in our basis.
	
	It remains to choose $\lvert F_{<1} \rvert$ further linearly independent constraints to add to our basis. Note that we have already added all tight integrality constraints to our basis, and no non-negativity constraint is tight. Then the only remaining constraints we can add are the $C^*$-constraints and the $r$ constraints of $Ay \leq b$.
	
	We claim that we cannot add any $C^*_{=1}$-constraints, because every $C^*_{=1}$-constraint is of the form $y(F_j) = y_{i_j} = 1$ for the unique integral facility $i_j \in F_1$. Note that here we used the fact that there is no facility that is set to zero. Thus every $C^*_{=1}$-constraint is linearly dependent with the tight integrality constraints, which we already chose.
	
	It follows, the only possible constraints we can choose are the $C^*_{<1}$-constraints and the $r$ constraints of $Ay \leq b$ so:
	\[\lvert F_{< 1} \rvert \leq dim(C^*_{< 1}) + r.\]
\end{proof}

\begin{proof}[Proof of \Cref{prop_violating}]
	It suffices to upper bound the quantity $\sum\limits_{j \in C^*_{<1}} (\lvert F_j \rvert - 2)$ by $2r$, because each client in $V$ contributes at least one to this sum, and every term of this sum is non-negative (because every $F$-ball contains at least two facilities.)
	
	Using Proposition \ref{prop_degree}, we can bound $\sum\limits_{j \in C^*_{<1}} \lvert F_j \rvert \leq 2 \lvert F(C^*_{<1}) \rvert \leq 2 \lvert F_{<1} \rvert$, where in the final inequality, we use the fact that every $F$-ball for clients in $C^*_{<1}$ is supported on \emph{only} fractional facilities. To upper bound our desired quantity, we use this inequality combined with Lemma \ref{lem_fractionalextreme} to obtain:
	\[\sum\limits_{j \in C^*_{<1}} (\lvert F_j \rvert - 2) \leq 2 \lvert F_{<1} \rvert - 2 \lvert C^*_{<1} \rvert \leq 2(dim(C^*_{<1}) + r) - 2 \lvert C^*_{<1} \rvert \leq 2r.\]
\end{proof}

\begin{proof}[Proof of \Cref{prop_compchain}]
	Consider any $V_k$, which is the vertex set of $H_k$. We already established that $H_k$ is either a path or even cycle. In both cases, we can order $V_k = \{j_1, \dots, j_p\}$ such that $j_1 \rightarrow j_2 \rightarrow \dots \rightarrow j_p$ is a path in the intersection graph. We verify that $V_k$ satisfies both properties of a chain.
	
	Because $V_k \subset C^*_{<1} \setminus V$, we have $\lvert F_{j_q} \rvert = 2$ for all $q \in [p]$, and because $j_q \rightarrow j_{q+1}$ is an edge in the intersection graph for all $q \in [p-1]$, we have $F_{j_q} \cap F_{j_{q+1}} \neq \emptyset$ for all $q \in [p-1]$.
\end{proof}

\begin{proof}[Proof of \Cref{prop_path}]
	Because $H_k$ is a path, suppose $V_k = \{j_1, \dots, j_p\}$ such that $j_1 \rightarrow \dots \rightarrow j_p$. There are two cases to consider.
	
	We first handle the degenerate case where there exists $q \in [p-1]$ such that $\lvert F_{j_q} \cap F_{j_{q + 1}} \rvert =  2$. Note that each $F$-ball for clients in $V_k$ has size exactly two, so we have $F_{j_q} = F_{j_{q+1}}$. Then both facilities in $F_{j_q}$ are already in exactly two $F$-balls, so $j_q$ and $j_{q+1}$ can have no other neighbors in the intersection graph. This implies that $H_k$ is path of length two, so $j_1 \rightarrow j_2$, such that $F_{j_1} = F_{j_2}$. To finish this case, we note that $\lvert F(V_k) \rvert = 2$, but $dim(V_k) = 1$, because both constraints $y(F_{j_1}) = 1$ and $y(F_{j_2})= 1$ are the same.
	
	In the second case, for all $q \in [p-1]$, we have $\lvert F_{j_q} \cap F_{j_{q+1}} \rvert = 1$. Using the fact that each facility is in at most two $F$-balls (Proposition \ref{prop_degree}), we have that each non-leaf client on the path has two facilities in its $F$-ball - one that it shares with the previous client on the path and one that it shares with the next. For the leaf clients, they share one facility with their single neighbor in the path, and they have one facility that is not shared with any other client in $V_k$. With these observations, we can bound:
	\[dim(V_k) \leq \lvert V_k \rvert = \frac{1}{2} \sum\limits_{j \in V_k} \lvert F_j \rvert = \frac{1}{2} (2 \lvert F(V_k) \rvert - 2) = \lvert F(V_k) \rvert - 1.\]
\end{proof}

\begin{proof}[Proof of \Cref{prop_cycle}]
	Each $F$-ball for clients along this cycle contains exactly two facilities. Using Proposition \ref{prop_degree}, each client along this cycle shares one of its facilities with the previous client in the cycle and one with the next client. This implies that each facility in $F(V_k)$ is in exactly two $F$-balls. Combining these two observations gives:
	\[\lvert V_k \rvert = \frac{1}{2} \sum\limits_{j \in V_k} \lvert F_j \rvert = \frac{1}{2}(2 \lvert F(V_k) \rvert) = \lvert F(V_k) \rvert.\]
	Thus, to prove $\lvert F(V_k) \rvert > dim(V_k)$, it suffices to show that $dim(V_k) < \lvert V_k \rvert$. We do this by showing that the constraints $\{y(F_j) = 1 \mid j \in V_k\}$ are not linearly independent. By assumption, $H_k$ is bipartite with bipartition, say $L \cup R = V_k$. Consider the linear combination of the constraints $\{y(F_j) = 1 \mid j \in V_k\}$, where every constraint indexed by $L$ has coefficient $1$ and every constraint indexed by $R$ has coefficient $-1$. Then for every facility in $F(V_k)$, it is in exactly two $F$-balls, and these two $F$-balls must be on opposite sides of the bipartition, so each facility in $F(V_k)$ has coefficient $0$ in this linear combination. In conclusion, we have constructed a non-trivial linear combination of the constraints $\{y(F_j) \mid j \in V_k\}$ whose left hand side is the zero vector, so $dim(V_k) \leq \lvert V_k \rvert$.
\end{proof}

\section{Proof of Theorem \ref{thm_outpre}: $k$-Median with Outliers Pre-Processing}\label{appendix_preoutlier}

The goal of this section is to prove Theorem \ref{thm_outpre} using the relevant theorems from \cite{DBLP:conf/stoc/KrishnaswamyLS18}. Note that we follow exactly the same pre-processing steps; the only difference is that we summarize the results of their pre-processing in a single theorem.

The proof of the knapsack pre-processing, \Cref{thm_kpre}, follows analogously from the pre-processing steps in \cite{DBLP:conf/stoc/KrishnaswamyLS18} as well.

\subsection{Preliminaries}

We define the notions of \emph{extended instances} and \emph{sparse extended instances} for $k$-median with outliers. These definitions are useful to capture the properties of our pre-processing. 

Extended instances are used to handle the fact that in our pre-processing, we will guess some facilities to pre-open. Then $S_0$ is the set of guessed facilities.

\begin{definition}[Extended Instance for $k$-Median with Outliers]\label{def_extendedoutliers}
    An extended instance for $k$-median with outliers is of the form $\I = (F,C,d,k,m,S_0)$, where $F$, $C$, $d$, $k$, and $m$ are defined as in a standard $k$-median with outliers instance (see Definition \ref{def_instanceoutlier}), and $S_0 \subset F$.
    
    As in $k$-median with outliers, the goal is to choose a set of at most $k$ open facilities $S \subset F$ and at least $m$ clients $C' \subset C$ to serve to minimize the connection costs of the served clients to the open facilities, so $\sum\limits_{j \in C'} d(j,S)$. However, we add the additional constraint that the set of open facilities must include $S_0$.
\end{definition}

Further, sparse extended instances give our properties for what it means for the facilities and clients to be ``cheap" (see the second and third properties in the next definition, respectively.)

\begin{definition}[Sparse Extended Instance for $k$-Median with Outliers]\label{def_sparseinstance}
    Let $\I' = (F ,C', d, k, m', S_0)$ be an extended $k$-median with outliers instance and $\rho, \delta \in (0, \frac{1}{2})$, $U \geq 0$ be parameters. We say that $\I'$ is $(\rho,\delta, U)$-sparse with respect to solution $(S^*, {C^*}')$ if the following three properties hold:
    \begin{enumerate}
        \item the cost of the solution $(S^*, {C^*}')$ to $\I'$ is at most $U$
        \item for all $i \in S^* \setminus S_0$, we have $\sum\limits_{j \in {C^*}' \mid d(j,S^*) = d(j,i)} d(j, i) \leq \rho U$
        \item for all $p \in F \cup C'$, we have $\lvert B_{C'}(p, \delta d(p, S^*))\rvert d(p, S^*) \leq \rho U$
    \end{enumerate}
\end{definition}

\subsection{Sparsification}

In this section, we pass from the input $k$-median with outliers instance to a sparse extended sub-instance by guessing the expensive parts of the input instance. Then on this sparse extended sub-instance, we can strengthen $\lpbasic$. The following theorems are directly from \cite{DBLP:conf/stoc/KrishnaswamyLS18}, so we omit the proofs in this paper. The first theorem states that we can efficiently compute a sparse extended sub-instance at the cost of a small increase in approximation ratio.

\begin{thm}\label{thm_preguess}
	Let $\I = (F,C,d,m,k)$ be an instance of $k$-median with outliers with optimal solution $(S^*, C^*)$ and $\rho, \delta \in (0,1/2)$ be parameters. Then there exists a $n^{O(1/\rho)}$-time algorithm that given $\I$, $\rho$, $\delta$, and an upper bound $U$ on the cost of the optimal solution $(S^*, C^*)$\footnote{Note that we are given $U$, but not the solution $(S^*, C^*)$}, outputs $n^{O(1/\rho)}$-many extended $k$-median with outliers instances of the form $\I' = (F,C', d, m',k, S_0)$ such that $C' \subset C$, $m' = \lvert C^* \cap C' \rvert$, and $S_0 \subset S$. Further, one such instance $\I'$ is $(\rho, \delta, U)$-sparse with respect to the solution $(S^*, C^* \cap C')$ and satisfies:
    \begin{equation}\label{eq_partialcost}
        \frac{1 - \delta}{1 + \delta} \sum\limits_{j \in C^* \setminus C'} d(j, S_0) + \sum\limits_{j \in C^* \cap C'} d(j, S^*) \leq U
    \end{equation}
\end{thm}

Once we have our sparse extended sub-instance, say $\I'$, we use these sparsity properties to compute the $R$-vector, which is needed for our Extra Invariants.

\begin{thm}\label{thm_prerad}
	Let $\I' = (F, C', d, m', k, S_0)$ be an extended $k$-median with outliers instance and $\rho, \delta \in (0, 1/2)$ and $U \geq 0$. Suppose $\I'$ is $(\rho, \delta, U)$-sparse instance with respect to solution $(S^*, {C^*}')$ to $\I'$ such that $(S^*, {C^*}')$ has cost $U'$ on $\I'$. Then there exists a polynomial time algorithm that takes as input $\I'$, $\rho$, $\delta$, and $U$ and outputs $R \in \mathbb{R}_+^{C'}$ satisfying:
	\begin{enumerate}
		\item For every $t > 0$ and $p \in F \cup C'$, we have:
		\[\lvert \{j \in B_{C'}(p, \frac{\delta t}{4 + 3 \delta}) \mid R_j \geq t \} \rvert \leq \frac{\rho (1 + 3\delta / 4)}{1 - \delta/4} \frac{U}{t}\]
		\item There exists a solution to $\I'$ of cost at most $(1 + \delta / 2) U'$ such that if client $j$ is connected to facility $i$, then $d(j,i) \leq R_j$ and for any facility $i \notin S_0$, the total cost of clients connected to $i$ in this solution is at most $\rho (1 + \delta / 2) U$\label{prop_sparsesolution}
	\end{enumerate}
\end{thm}

\subsection{Putting it all Together: Proving Theorem \ref{thm_outpre}}

Combining the algorithms guaranteed by these above two theorems, we show how to construct $\lpreroute$ with the desired properties.

    Suppose we are given a $k$-median with outliers instance $\I = (F,C,d,m,k)$, parameters $\rho, \delta \in (0, \frac{1}{2})$, and an upper bound $U$ of $Opt(\I)$. First we run the algorithm guaranteed by Theorem \ref{thm_preguess} to obtain $n^{O(1/\rho)}$-many extended $k$-median with outliers instances. Then for each instance, we run the algorithm guaranteed by Theorem \ref{thm_prerad} to obtain a vector $R$ for each such instance.
    
    By Theorem \ref{thm_preguess}, let $\I' = (F,C' \subset C,d,m' = m - \lvert C^* \setminus C' \rvert,k, S_0)$ be the instance output by the first algorithm such that $\I'$ is $(\rho, \delta, U)$-sparse with respect to the solution $(S^*, C^* \cap C')$ and satisfies Equation \ref{eq_partialcost}. This sub-instance will be the one that is guaranteed by Theorem \ref{thm_outpre}, so from here we need to compute the $R$-vector, and construct $\lpreroute$ with the desired properties.
    
    Note that the cost of solution $(S^*, C^* \cap C')$ to $\I'$ is exactly $U' = \sum\limits_{j \in C^* \cap C'} d(j, S^*)$. It follows, on this instance $\I'$, the algorithm guaranteed by Theorem \ref{thm_prerad} outputs a vector $R \in \mathbb{R}_+^{C'}$ such that for every $t > 0$ and $p \in F \cup C'$, we have:
    \[\lvert \{j \in B_{C'}(p, \frac{\delta t}{4 + 3 \delta}) \mid R_j \geq t \} \rvert \leq \frac{\rho (1 + 3\delta / 4)}{1 - \delta/4} \frac{U}{t}\]
    , and there exists a solution, say $(\bar{S}, \bar{C})$ to $\I'$ of cost at most $(1 + \delta /2) U'$ such that if $j$ is connected to facility $i$, then $d(j,i) \leq R_j$ and for any facility $i \in \bar{S} \setminus S_0$, the total cost of clients connected to $i$ in this solution is at most $\rho(1 + \delta/2) U$.
    
    It remains to construct $\lpreroute$. To do so, first we construct a strengthened LP for the instance $\I'$ such that $(\bar{S}, \bar{C})$ is feasible for the strengthened LP, which we call $\lpextra$:
    
\begin{align*}\tag{$\lpextra$}
	\min\limits_{x,y}~~ &\sum\limits_{i \in F, j \in C'} d(i,j) x_{ij}\\
	\text{s.t.}~~ &x_{ij} \leq y_i \quad \forall i \in F, j \in C' &  &y_i = 1 \quad \forall i \in S_0\\
	&\sum\limits_{i \in F} x_{ij} \leq 1 \quad \forall j \in C' & &x_{ij} = 0 \quad \forall i \in F, j \in C' \text{ s.t. } d(i,j) > R_j\\
	&\sum\limits_{i \in F} y_i \leq k & & \sum\limits_{j \in C'} d(i,j) x_{ij} \leq \rho (1 + \delta / 2) U y_i \quad \forall i \notin S_0\\
	&\sum\limits_{j \in C', i \in F} x_{ij} \geq m & & x_{ij} = 0 \quad \forall i \notin S_0, j \in C' \text{ s.t. } d(i,j) > \rho (1 + \delta/2) U\\
	&0 \leq x,y \leq 1 \\
\end{align*}

The left column of constraints are the same as $\lpbasic$ and the right column of constraints are extra constraints that are valid for the solution $(\bar{S}, \bar{C})$ to our sub-instance $\I'$. Because these constraints are valid for the solution $(\bar{S}, \bar{C})$, the following proposition is immediate.

\begin{prop}\label{prop_lpextracost}
	$Opt(\lpextra) \leq (1 + \delta / 2)U'$.
\end{prop}

From here, we want to carry out a similar construction as in \S \ref{sec_iteroverview}, where we construct $\lpreroute$ satisfying all Basic Invariants from $\lpbasic$. We note that the main difference in our procedure here when compared to \S \ref{sec_iteroverview} is how we eliminate the $x$-variables. To compute the $F$-balls for $\lpextra$, we must carefully duplicate facilities to capture the constraints: $\sum\limits_{j \in C'} d(i,j) x_{ij} \leq \rho (1 + \delta / 2) U y_i \quad \forall i \notin S_0$. We pass from $\lpextra$ to $\lpfacil$ with the next lemma.

\begin{lem}\label{lem_extrasplit}
    There exists a polynomial time algorithm that takes as input $\lpextra$, duplicates facilities in $F$, and outputs a vector $\bar{y} \in [0,1]^F$ and sets $F_j \subset B_F(j,R_j)$ for all $j \in C'$ such that:
    \begin{enumerate}
        \item $\bar{y}(F_j) \leq 1$ for all $j \in C'$
        \item $\bar{y}(F) \leq k$
        \item $\sum\limits_{j \in C'} \bar{y}(F_j) \geq m$
        \item $\sum\limits_{j \in C} \sum\limits_{i \in F_j} d(i,j) \bar{y}_i \leq Opt(\lpextra)$
        \item For all $i \in S_0$, there is one unit of open facility co-located with $i$ in $\bar{y}$
        \item For every facility $i$ not co-located with a facility in $S_0$, we have $\sum\limits_{j \in C' \mid i \in F_j} d(i,j) \leq 2 \rho ( 1 + \delta / 2) U$
    \end{enumerate}
\end{lem}

Applying the algorithm guaranteed by the above lemma to $\lpextra$, we can obtain the $F_j$-sets. Using these $F$-balls, we proceed similarly as in \S \ref{sec_iteroverview}. Thus, next we randomly discretize the distances to powers of $\tau > 1$ (up to a random offset) to obtain $d'(p,q)$ for all $p,q \in F \cup C$. Again, the possible discretized distances are $L(-2) = -1, L(-1) = 0, \dots, L(\ell) = \alpha \tau^\ell$ for all $\ell \in \mathbb{N}$, and $d'$ satisfies Lemma \ref{lem_disc}.

Then we define the radius levels and inner balls in the exact same way, so:
\[\ell_j = \min\limits_{\ell \geq -1} \{\ell \mid d'(j,i) \leq L(\ell) \quad \forall i \in F_j\}\]
\[B_j = \{i \in F_j \mid d'(j,i) \leq L(\ell_j - 1)\}\]

To complete the data of $\lpreroute$ for $\I'$, we need to define the sets $C_{part}$, $C_{full}$, and $C^*$. Here we must slightly modify the construction of \S \ref{sec_iteroverview} to accommodate the set of pre-opened facilities, $S_0$. To satisfy Extra Invariant \ref{invar_outextra}(\ref{invar_outguessfacil}), we create a set $C_0$ of dummy clients such that for each $i \in S_0$, there exists a dummy client $j(i) \in C_0$ that is co-located with $i$ such that $F_{j(i)}$ has radius level $-1$ and consists of all co-located copies of $i$. Thus, we define $C_{part} = C'$, $C_{full} = \emptyset$, and $C^* = C^0$.

This completes the description of $\lpreroute$ for sub-instance $\I'$. To complete our algorithm, we output each $\I'$ along with $\lpreroute$, $S_0$, and $R$.

To summarize, our algorithm is to first run the algorithm guaranteed by Theorem \ref{thm_preguess} to obtain $n^{O(1/\rho)}$-many sub-instances. For each sub-instance, we compute $R$ using Theorem \ref{thm_prerad}, construct $\lpextra$, construct the $F$-balls using Lemma \ref{lem_extrasplit}, and define the rest of the data of $\lpreroute$ as in \S \ref{sec_iteroverview}. The runtime of our algorithm is immediate, so it suffices to show that one of the outputs has the desired properties.

In particular, we consider the sub-instance $\I' = (F,C' \subset C,d,m' = m - \lvert C^* \setminus C' \rvert,k, S_0)$ output by the algorithm guaranteed by Theorem \ref{thm_preguess} such that $\I'$ is $(\rho, \delta, U)$-sparse with respect to the solution $(S^*, C^* \cap C')$ and satisfies Equation \ref{eq_partialcost}. For the remainder of this section, we consider the $\lpreroute$ constructed for this specific sub-instance. To complete the proof, we verify that $\lpreroute$ satisfies the two desired properties.

\begin{prop}\label{prop_extralp}
    $\lpreroute$ satisfies all Basic and Extra Invariants.
\end{prop}
\begin{proof}
    It is easy to verify that $\lpreroute$ satisfies all Basic Invariants by construction. For the Extra Invariants, we handle them one-by-one.
    
     Extra Invariant \ref{invar_outextra}(\ref{invar_outguessfacil}) holds by construction of $\lpreroute$. To show Extra Invariant \ref{invar_outextra}(\ref{invar_outsparsefacil}), we again apply Lemma \ref{lem_extrasplit}, which states that the $F$-balls have the desired property.
    
    To show Extra Invariant \ref{invar_outextra}(\ref{invar_outupperbound}), we note that in Lemma \ref{lem_extrasplit}, the $F$-balls are constructed such that $F_j \subset B_F(j, R_j)$, for all $j \in C'$. Thus, for all $j \in C'$ and $i \in F_j$, we have $d(i,j) \leq R_j$. Then by definition of the radius levels, we have $L(\ell_j) \leq \tau R_j$, as required. Finally, Extra Invariant \ref{invar_outextra}(\ref{invar_outsparserad}) follows from the guarantee of Theorem \ref{thm_prerad}.
\end{proof}

\begin{prop}\label{prop_extralpopt}
    $\frac{\log_e \tau}{(\tau - 1)(1 + \delta/2)}\mathbb{E}[Opt(\lpreroute)] + \frac{1 - \delta}{1+ \delta} \sum\limits_{j \in C^* \setminus C'} d(j, S_0) \leq U$
\end{prop}
\begin{proof}
    We first show that $\mathbb{E}[Opt(\lpreroute)] \leq \frac{\tau - 1}{\log_e \tau} Opt(\lpextra)$. We have $C_{part} = C'$, $C_{full} = \emptyset$, and $C^* = C_0$. Note that the dummy clients in $C^*$ contribute zero to the objective of $\lpreroute$, because they are co-located with one unit of open facility in their $F$-balls. Thus Lemma \ref{lem_extrasplit} implies that there exists a feasible solution $\bar{y}$ to $\lpreroute$ of cost at most $Opt(\lpextra)$ up to the discretization of the distances. The cost of discretization is bounded by Lemma \ref{lem_disc}, and immediately gives the extra $\frac{\tau - 1}{\log_e \tau}$-factor. The factor of $\frac{\tau - 1}{\log_e \tau}$ is due to the cost of discretization, which is bounded by Lemma \ref{lem_disc}.
    
    Now we relate $Opt(\lpextra)$ to $U' = \sum\limits_{j \in C^* \cap C'} d(j,S^*)$, which satisfies Equation \ref{eq_partialcost} by the guarantees of Theorem \ref{thm_preguess}. Combining Equation \ref{eq_partialcost} with Proposition \ref{prop_lpextracost}, we can obtain our final bound:
    \begin{align*}
        \frac{\log_e \tau}{(\tau - 1)(1 + \delta/2)}\mathbb{E}[Opt(\lpreroute)] + \frac{1 - \delta}{1+ \delta} \sum\limits_{j \in C^* \setminus C'} d(j, S_0) &\leq \frac{1}{1 + \delta/2} Opt(\lpextra) + \frac{1 - \delta}{1+ \delta} \sum\limits_{j \in C^* \setminus C'} d(j, S_0)\\
        &\leq \sum\limits_{j \in C^* \cap C'} d(j,S^*) + \frac{1 - \delta}{1+ \delta} \sum\limits_{j \in C^* \setminus C'} d(j, S_0)\\
        &\leq U
    \end{align*}
\end{proof}

\end{document}